\documentclass[11pt,letterpaper]{article}

\usepackage[letterpaper,margin=1in]{geometry}

\usepackage{amssymb,amsmath,amsthm,amsfonts}
\usepackage{mathtools}
\usepackage{enumitem}
\usepackage[numbers,comma,sort&compress]{natbib}
\usepackage{authblk}
\usepackage{graphicx}
\usepackage{caption}
\usepackage{tabularx,array}

\usepackage[labelformat=simple]{subcaption}
\usepackage{float}
\usepackage[ruled,vlined,linesnumbered]{algorithm2e}
\usepackage{physics}
\usepackage{footnote}
\usepackage{bbding}

\usepackage{tikz}
\usetikzlibrary{calc,patterns,angles,quotes}
\usepackage{siunitx}
\usetikzlibrary{quantikz}
\usepackage{hyperref}
\usepackage{cleveref}
\usepackage{thm-restate}
\usepackage{tablefootnote}

\crefname{section}{Section}{Sections}
\Crefname{section}{Section}{Sections}
\crefname{subsection}{Subsection}{Subsections}
\Crefname{subsection}{Subsection}{Subsections}
\crefname{theorem}{Theorem}{Theorems}
\Crefname{theorem}{Theorem}{Theorems}
\crefname{lemma}{Lemma}{Lemmas}
\Crefname{lemma}{Lemma}{Lemmas}
\crefname{proposition}{Proposition}{Propositions}
\Crefname{proposition}{Proposition}{Propositions}
\crefname{corollary}{Corollary}{Corollaries}
\Crefname{corollary}{Corollary}{Corollaries}
\crefname{definition}{Definition}{Definitions}
\Crefname{definition}{Definition}{Definitions}
\crefname{remark}{Remark}{Remarks}
\Crefname{remark}{Remark}{Remarks}
\crefname{fact}{Fact}{Facts}
\Crefname{fact}{Fact}{Facts}
\crefname{equation}{Equation}{Equations}
\Crefname{equation}{Equation}{Equations}
\crefname{figure}{Figure}{Figures}
\Crefname{figure}{Figure}{Figures}
\crefname{table}{Table}{Tables}
\Crefname{table}{Table}{Tables}
\crefname{algocf}{Algorithm}{Algorithms}
\Crefname{algocf}{Algorithm}{Algorithms}

\newtheorem{theorem}{Theorem}[section]
\newtheorem{definition}[theorem]{Definition}
\newtheorem{lemma}[theorem]{Lemma}
\newtheorem{proposition}[theorem]{Proposition}

\newtheorem{corollary}[theorem]{Corollary}

\newtheorem{fact}[theorem]{Fact}

\Crefname{algocf}{Algorithm}{Algorithms}
\crefname{algorithm}{Algorithm}{Algorithms}

\newcommand{\eq}[1]{Eq.~(\ref{eq:#1})}

\newcommand{\thm}[1]{\hyperref[thm:#1]{Theorem~\ref*{thm:#1}}}
\newcommand{\cor}[1]{\hyperref[cor:#1]{Corollary~\ref*{cor:#1}}}
\newcommand{\defn}[1]{\hyperref[defn:#1]{Definition~\ref*{defn:#1}}}
\newcommand{\lem}[1]{\hyperref[lem:#1]{Lemma~\ref*{lem:#1}}}
\newcommand{\prop}[1]{\hyperref[prop:#1]{Proposition~\ref*{prop:#1}}}
\newcommand{\fig}[1]{\hyperref[fig:#1]{Figure~\ref*{fig:#1}}}
\newcommand{\tab}[1]{\hyperref[tab:#1]{Table~\ref*{tab:#1}}}
\newcommand{\algo}[1]{\hyperref[algo:#1]{Algorithm~\ref*{algo:#1}}}
\renewcommand{\sec}[1]{\hyperref[sec:#1]{Section~\ref*{sec:#1}}}
\newcommand{\append}[1]{\hyperref[append:#1]{Appendix~\ref*{append:#1}}}
\newcommand{\fac}[1]{\hyperref[fac:#1]{Fact~\ref*{fac:#1}}}
\newcommand{\lin}[1]{\hyperref[lin:#1]{Line~\ref*{lin:#1}}}

\newcommand{\br}[1]{\left( #1 \right)}
\newcommand{\set}[1]{ \{ #1 \} }

\renewcommand{\abs}[1]{\left|#1\right|}

\renewcommand{\arraystretch}{1.25}

\def\>{\rangle}
\def\<{\langle}

\renewcommand{\bra}[1]{\langle#1|}
\renewcommand{\ket}[1]{|#1\rangle}
\newcommand{\reg}[1]{\mathsf{#1}}

\renewcommand{\bf}{\mathbf}
\renewcommand{\i}{\mathrm{i}}

\DeclareMathOperator{\spn}{Span}

\DeclareMathOperator{\supp}{supp}

\renewcommand{\emptyset}{\varnothing}
\def\Tr{\operatorname{Tr}}\def\:{\hbox{\bf:}}

\makeatletter
\newcommand{\doubletilde}[1]{{%
  \mathpalette\double@tilde{#1}%
}}
\newcommand{\double@tilde}[2]{%
  \sbox\z@{$\m@th#1\tilde{#2}$}%
  \ht\z@=.9\ht\z@
  \tilde{\box\z@}%
}
\makeatother

\let\oldnl\nl
\newcommand{\nonl}{\renewcommand{\nl}{\let\nl\oldnl}}

\SetKwComment{Comment}{/* }{ */}
\SetKwInput{KwInput}{Input}                
\SetKwInput{KwOutput}{Output}              
\SetKwInput{KwNotation}{Notation}              
\SetKwInput{KwRegister}{Registers}              
\SetKwInput{KwParameter}{Parameter}             
\newcommand{\AlgoDisplayStep}[2]{%
\parbox[t]{0.92\linewidth}{%
#1:
$$ #2 $$
}\;
}

\numberwithin{equation}{section}

\makeatletter
\@addtoreset{algocf}{section}
\makeatother

\newcommand{\T}{\ensuremath{\mathrm{T}}}

\title{Optimal T Counts under Sparsity: from QROM to State Preparation and Block Encoding}
\author[1,2]{Tongyang Li\thanks{\texttt{tongyangli@pku.edu.cn}}}
\author[3]{Fengning Ou\thanks{\texttt{reverymoon@gmail.com}}}
\author[1,2]{Xinzhao Wang\thanks{\texttt{wangxz@stu.pku.edu.cn}}}
\author[3,4]{Penghui Yao\thanks{\texttt{phyao1985@gmail.com}}}
\author[5]{Pei Yuan\thanks{\texttt{peiyuan@tencent.com}}}
\author[5]{Shengyu Zhang\thanks{\texttt{shengyzhang@tencent.com}}}
\affil[1]{Center on Frontiers of Computing Studies, Peking University, Beijing 100871, China}
\affil[2]{School of Computer Science, Peking University, Beijing 100871, China}
\affil[3]{State Key Laboratory of Novel Software Technology, Nanjing University, Nanjing 210023, China}
\affil[4]{Hefei National Laboratory, Hefei 230088, China}
\affil[5]{Tencent Quantum Laboratory}
\date{}

\begin{document}

\maketitle

\begin{abstract}
Many quantum algorithms require coherent access to classical data, often
modeled by quantum read-only memory (QROM). We initiate the study of the
$\mathrm{T}$ count of sparse QROM, in which only $s$ of the $2^n$ addresses
store nonzero data. We prove that the optimal $\mathrm{T}$ count is
$\Theta\left(n+\min\left\{s,\sqrt{s\left(m+\log(2^{n+1}/s)\right)}\right\}\right)$.
Our upper bounds use a multilevel hashing scheme, while
our lower bounds reduce sparse QROM to state preparation and use counting
arguments for adaptive Clifford+$\mathrm{T}$ circuits. The lower bounds thus
hold even when mid-circuit measurements and classically controlled operations
are allowed. As applications, we obtain matching $\mathrm{T}$-count bounds
$\Theta\left(\min\left\{s,\sqrt{s\log(2^{n+1}/s)}\right\}
+\sqrt{s\log(1/\varepsilon)}+\log(1/\varepsilon)\right)$
for $s$-sparse state preparation and
$\Theta\left(\sqrt{2^n s\left(n+\log(1/\varepsilon_{\rm BE})\right)}
+\log(1/\varepsilon_{\rm BE})\right)$
for block encoding of $s$-sparse matrices,
where $\varepsilon$ and $\varepsilon_{\rm BE}$ are the precision of state
preparation and block encoding, respectively.
\end{abstract}
\newpage

\tableofcontents
\clearpage
\section{Introduction}

Quantum algorithms often require coherent access to classical data. A problem
instance may be specified by a Hamiltonian, a data set, a probability
distribution, or the entries of a matrix, and the algorithm is formulated in
an oracle model where this data is queried in superposition. Such coherent
data access appears in Hamiltonian simulation~\cite{low2019hamiltonian},
quantum linear algebra~\cite{GSLW19}, and state
preparation~\cite{david2026quantum}. Implementing these oracles requires the
relevant classical data to be loaded coherently, and the cost of data loading
is part of the resource cost of the algorithm.

The standard abstraction for coherent data loading is quantum read-only memory
(QROM)~\cite{babbush2018encoding}, also referred to as quantum random-access
classical memory (QRACM) in parts of the literature
\cite{kuperberg2013subexp,jaques2025qram}. Given an \(n\)-bit address space
and \(m\)-bit data values, let
\(d=(d_x)_{x\in\{0,1\}^n}\) with \(d_x\in\{0,1\}^m\). The associated QROM
unitary is defined on basis states by
\[
    O_d\ket{x}\ket{0^m}=\ket{x}\ket{d_x},
    \qquad x\in\{0,1\}^n.
\]
We write \(N=2^n\) for the table size. By linearity, the address \(x\) may also
be queried in superposition. Several
models of coherent memory access have been studied, including hardware quantum
random-access memory (QRAM) architectures such as the fanout and bucket-brigade proposals~\cite{giovannetti2008quantum},
as well as circuit-level analyses of QRAM/QROM implementations~\cite{dimatteo2019qramft,hann2021resilience,Low2024tradingtgatesdirty,haner2022tablelookup,zhu2024qlut}.
In this work we focus on the circuit model: the data \(d\) are fixed in
advance, and the task is to compile the QROM unitary \(O_d\) into a circuit over
the Clifford+\(\T\) gate set. This gate set is standard in fault-tolerant
quantum computation.

We measure this compilation cost by the \emph{\T\ count}, namely the number of
\(\T\) gates in the resulting circuit. Circuits consisting only of Clifford
gates are classically simulable by the Gottesman--Knill
theorem~\cite{gottesman1998heisenberg}, whereas universal quantum computation
requires resources beyond Clifford operations. In the Clifford+\(\T\) model, this non-Clifford resource is supplied
by the \(\T\) gate, which in many fault-tolerant architectures is implemented
using magic-state distillation~\cite{bravyi2005universal,litinski2019magicstate}
or cultivation~\cite{gidney2024magicstatecultivation}. Thus, \T\ count is a
standard measure of the non-Clifford resource used by a Clifford+\(\T\)
implementation. Although implementing the QROM unitary for a table of \(N\)
entries requires \(\Omega(N)\) total gates in
general~\cite{zhangCircuitComplexityQuantum2024a,jaques2025qram}, the corresponding \T\ count can be much
smaller: the SELECT--SWAP construction of
Low, Kliuchnikov, and Schaeffer~\cite{Low2024tradingtgatesdirty} achieves a square-root dependence on \(N\).

This QROM bound applies to arbitrary tables, but many data-loading tasks
have additional structures. In particular, the full address space may be large
while only \(s\ll N\) entries are nonzero. We call the set of nonzero entries
the support of the table, and refer to the corresponding QROM as
\emph{sparse QROM}. Sparse classical data arise in sparse Hamiltonian simulation~\cite{berry2015simulating},
sparse linear-system solving~\cite{HHL2009,childs2017qls}, and sparse state preparation~\cite{gleinig2021state,li2025nearly}. 
In such
settings, the amount of relevant classical information is proportional to
\(s\), rather than to \(N\). This motivates the central question of sparse
coherent data loading:
\[
\textit{What is the optimal \T\ count for loading sparse classical data?}
\]

We answer this question by giving {\em tight bounds} on \T\ count for exact
and approximate sparse QROM throughout the sparsity range. Our constructions
use the number \(s\) of nonzero entries to reduce the cost of loading the data.
The lower bounds
hold even in the more general adaptive implementation model introduced later.
These sparse-QROM bounds also yield optimal \T\ counts for
sparse state preparation and sparse block encoding, with applications to sparse Hamiltonian simulation, sparse linear-system solving, and quantum
rejection sampling.
\subsection{Main results}

In this paper, we first study the \T\ count needed to implement sparse QROM.
For a data vector \(d\), let \(\supp(d)=\{x:d_x\neq0^m\}\) and
\(|\supp(d)|\le s\). The standard sparse
QROM must implement \(O_d\) on every input address. We also use a weaker model,
called \emph{promised sparse QROM}, which is only required to act correctly on
addresses in \(\supp(d)\):
\[
    O_d^{\mathrm{prom}}\ket{x}\ket{0^m}
    =
    \ket{x}\ket{d_x},
    \qquad x\in\supp(d).
\]
Here \emph{exact} means zero error, whereas \emph{approximate} allows a fixed
sufficiently small constant diamond-norm error on the relevant input subspace.
All lower bounds allow adaptive circuits with mid-circuit
measurements and classically controlled Clifford+\(\T\) operations; their cost
is the expected number of \T\ gates, maximized over valid inputs. The main
\T-count bounds appearing in this paper are summarized in
\cref{tab:qrom-tcount-summary}.

\begin{table}[!htbp]
\makebox[\textwidth][c]{%
\small
\setlength{\tabcolsep}{5pt}
\renewcommand{\arraystretch}{1.45}
\begin{tabular}{|>{\raggedright\arraybackslash}m{0.32\textwidth}|>{\raggedright\arraybackslash}m{0.6\textwidth}|}
\hline
Task & \T\ count \\ \hline
Dense QROM &
$\displaystyle\Theta\!\left(\min\{2^n,\sqrt{2^n m}\}\right)$
(\Cref{cor:dense-qrom-all-widths,thm:lb-dense-qrom}) \\ \hline
Promised sparse QROM &
$\displaystyle\Theta\!\left(\min\{s,\sqrt{sm}\}\right)$
(\Cref{thm:Algorithm-promised-sparse-QROM,thm:lb-promised-qrom}) \\ \hline
Sparse QROM (exact) &
$\displaystyle\Theta\!\left(n+\min\left\{s,\sqrt{s\left(m+\log(2^{n+1}/s)\right)}\right\}\right)$
(\Cref{thm:Algorithm-sparse-QROM,thm:lb-exact-sparse-qrom-support}) \\ \hline
Sparse QROM (approximate) &
$\displaystyle\Theta\!\left(\min\{s,\sqrt{sm}\}\right)$
(\Cref{thm:Algorithm-approximate-sparse-QROM,cor:lb-approximate-sparse-qrom}) \\ \hline
\end{tabular}%
}
\caption{Worst-case \T-count bounds. All bounds assume $m\ge2$, and all sparse bounds hold for
$4\le s\le2^n$.
For dense and promised sparse QROM, the upper bounds are exact,
and the matching lower bounds hold even with small constant error. All upper
bounds use unitary circuits, except for approximate sparse QROM, whose error
bound applies to the channel averaged over classical randomness.}
\label{tab:qrom-tcount-summary}
\label{tab:applications}
\end{table}

The applications below use dense QROM for unrestricted addresses and promised
sparse QROM when the queried address is guaranteed to lie in \(\supp(d)\). The
promised sparse QROM bound \(\Theta(\min\{s,\sqrt{sm}\})\) has the same form
as the dense QROM bound, with \(2^n\) replaced by \(s\) in such steps.

We first consider sparse state preparation. The input is a
classical description of an \(n\)-qubit state
\[
    \ket{\psi}
    =
    \sum_{x=0}^{s-1}\alpha_x\ket{a_x},
    \qquad
    a_x\in\{0,1\}^n,
\]
where the basis states \(a_0,\ldots,a_{s-1}\) are distinct, and the goal is to
prepare \(\ket{\psi}\) starting from the all-zero state within trace-distance
error \(\varepsilon\). We give a sparse state preparation algorithm and prove
a matching lower bound throughout \(4\le s\le2^n\).

\begin{theorem}[Informal version of
\Cref{thm:sparse-state-preparation,thm:sparse-state-lower-bound}]
Any \(s\)-sparse \(n\)-qubit state can be prepared within trace-distance error \(\varepsilon\) by a Clifford\,$+$\,\T\ circuit using 
\[
    O\!\left(
        \min\{s,\sqrt{s\log(2^{n+1}/s)}\}
        +
        \sqrt{s\log(1/\varepsilon)}
        +
        \log(1/\varepsilon)
    \right)
\] \T\ gates. Furthermore, if \(4\le s\le2^n\) and
\(0<\varepsilon\le1/64\), no adaptive
Clifford\,$+$\,\T\ circuit can use asymptotically fewer \T\ gates in the
worst case.
\end{theorem}

To see why promised sparse QROM is essential for the upper bound, first prepare the
compressed state \(\sum_{x=0}^{s-1}\alpha_x\ket{x}\ket{0^n}\) using the dense
state-preparation algorithm of Gosset, Kothari, and
Wu~\cite{david2026quantum}, and then use dense QROM to load the basis labels
\(a_x\), producing $\sum_{x=0}^{s-1}\alpha_x\ket{x}\ket{a_x}$.
It remains to erase the compressed label \(x\) while keeping \(\ket{a_x}\).
This can be done by applying the inverse of the promised sparse QROM
\(\ket{a_x}\ket{0}\mapsto\ket{a_x}\ket{(1,x)}\), where the extra flag only
separates the label \(x=0\) from the zero message. The address is promised to
lie in \(\{a_0,\ldots,a_{s-1}\}\). If this erase step were implemented as a dense
QROM over all \(2^n\) basis labels, the \T\ count would be
\(O(\sqrt{2^n\log s})\). The promise reduces the address set to the \(s\)
labels that actually occur, giving \T\ count
\(O(\min\{s,\sqrt{s\log s}\})\). Loading the labels \(a_x\) costs
\(O(\min\{s,\sqrt{sn}\})\), which dominates the erasure cost, while
dense state preparation contributes
\(O(\sqrt{s\log(1/\varepsilon)}+\log(1/\varepsilon))\).
For large \(s\), we use the multilevel hash values as part of the compressed
labels, so QROM only needs to load the remaining address coordinates.
This reduces the loading and erasure cost to
\(O(\sqrt{s\log(2^{n+1}/s)})\).
Converting the compressed state into the target sparse state therefore uses
\(O(\min\{s,\sqrt{s\log(2^{n+1}/s)}\})\) \T\ gates, compared with the
linear-in-\(s\) Toffoli cost of Rupprecht and W\"olk~\cite{rupprecht2026sparsequantumstatepreparation}
for this conversion. Together with dense state preparation, the total \T\ count matches
the lower bound in \Cref{thm:sparse-state-lower-bound} throughout the
sparsity range.

We next consider sparse block encoding. Here row- and column-\(s\)-sparse
means that every row and every column has at most \(s\) nonzero entries, and a
unitary \(U\) is an \((\alpha,a,\varepsilon_{\rm BE})\)-block-encoding of a
matrix \(A\) if
\[
    \left\|
        A-\alpha
        (\bra{0^a}\otimes I)U(\ket{0^a}\otimes I)
    \right\|
    \le \varepsilon_{\rm BE}.
\]
We give a sparse block encoding implementation with square-root \T\ count in
\(2^n s\), up to logarithmic factors. The same asymptotic \T\ count also
implements its controlled version. We also prove a matching adaptive lower
bound throughout \(4\le s\le2^n\) for sufficiently small constant matrix error.
An adaptive implementation realizes the block-encoding unitary \(U\) exactly;
\(\varepsilon_{\rm BE}\) is the error of the encoded matrix.

\begin{theorem}[Informal version of
\Cref{lem:complex-sparse-block-encoding,thm:controlled-block-encoding,thm:sparse-block-encoding-lower-bound}]
For any row- and column-\(s\)-sparse
\(A=(a_{ij})\in\mathbb{C}^{2^n\times2^n}\) with \(|a_{ij}|\le1\), both a
\((\Theta(s),O(n+\log(s/\varepsilon_{\rm BE})),\varepsilon_{\rm BE})\)-block-encoding of \(A\) and its controlled version can be implemented by a Clifford\,$ + $\,\T\ circuit using
\[
    O\!\left(
        \sqrt{2^n s(n+\log(1/\varepsilon_{\rm BE}))}
        +
        \log(1/\varepsilon_{\rm BE})
    \right)
\]
\T\ gates. There is an absolute constant \(c>0\) such that, if
\(4\le s\le2^n\) and \(0<\varepsilon_{\rm BE}\le c\), no
adaptive Clifford$+$\T\ circuit can use asymptotically fewer \T\ gates in
the worst case.
\end{theorem}
Promised sparse QROM enters the block-encoding upper bound through the operations that
load the nonzero positions in each row and column. For the row operation, let
\(r_{i,k}\) list the nonzero column indices in row \(i\), padded with
distinct zero-entry positions to length \(s\). We first use a dense QROM
to load \(r_{i,k}\in[2^n]\) from
\((i,k)\in[2^n]\times[s]\), with \T\ count
\(O(\sqrt{2^n s n})\), and then erase \(k\) from the pair
\((i,r_{i,k})\). This erasure operation is implemented by the inverse of a promised
sparse QROM, since \((i,r_{i,k})\) is promised to belong to these row lists.
A dense QROM for the erase map would range over all pairs
\((i,j)\in[2^n]\times[2^n]\), with \T\ count
\(O(2^n\sqrt{\log s})\), whereas the promised version only ranges over the
\(2^n s\) listed positions and has \T\ count
\(O(\sqrt{2^n s\log s})\), which is dominated by the
\(O(\sqrt{2^n s n})\) \T-count for loading \(r_{i,k}\). The column operation is
identical, with rows and columns interchanged. The remaining terms in the
block-encoding bound come from loading
\(\log(s/\varepsilon_{\rm BE})\)-bit approximations of the nonzero matrix
entries and converting these values into amplitudes using the subsampling
method of~\cite{babbush2018encoding}.

Our QROM, state-preparation, and block-encoding circuits give \T-count bounds for
quantum singular value transformation (QSVT), sparse Hamiltonian simulation,
sparse linear-system solving, and quantum rejection sampling. These algorithms in turn
lead to applications such as matrix inversion and quantum machine learning via
QSVT~\cite{GSLW19,Chakraborty2022QRLs}, quantum chemistry and linear
differential equations via Hamiltonian simulation~\cite{Kassal2010chemistry,Babbush2015fermions,Berry2017lode},
data fitting, electromagnetic scattering, and finite-element methods via
linear-system solvers~\cite{Wiebe2012datafitting,Clader2013preconditioned,Montanaro2016fem},
and quantum linear systems, quantum Metropolis sampling, and recent lattice
discrete Gaussian sampling and related cryptographic tasks via quantum
rejection sampling~\cite{HHL2009,Temme2011metropolis,Chevignard2026DGS,Ling2026QRSTrapdoor}.

The main additional application bounds are summarized in
\Cref{tab:application-tcount-summary}.
\begin{table}[htp]
\makebox[\textwidth][c]{%
\small
\renewcommand{\arraystretch}{2.1}
\begin{tabular}{|
>{\raggedright\arraybackslash}m{0.29\textwidth}|
p{0.72\textwidth}|}
\hline
Application & \T\ count \\ \hline

Singular value transformation &
\(O\br{d \br{\sqrt{2^n s(n+\log(sd/\varepsilon))}+\log(sd/\varepsilon) } }\)
(\Cref{lem:block-encoding-for-QSVT}) \\ \hline

Hamiltonian simulation &
\(O \br{ \br{s|\tau|+\log(1/\varepsilon)} \br{
\sqrt{2^n s\bigl(n+\log(s|\tau|/\varepsilon)\bigr)}
+\log(s|\tau|/\varepsilon)} }\)
(\Cref{thm:sparse-hamiltonian-simulation}) \\ \hline

Sparse linear-system solving &
\(O\br{s\kappa\log(1/\varepsilon)\br{\sqrt{2^n s(n+\log(s\kappa/\varepsilon))}
+\log(s\kappa/\varepsilon)} }\)
(\Cref{thm:sparse-qlss}) \\ \hline

Quantum rejection sampling &
\(O \br{ \br{\min\{s,\sqrt{s\log(2/\delta)}\}+\log(2/\delta)+q} / \|\epsilon\|_2}\)
(\Cref{thm:sparse-qrs-filter}) \\ \hline

\end{tabular}%
}

\caption{Additional application bounds obtained from sparse QROM and sparse
block encoding. Here $\varepsilon,\delta$ are the error parameters, \(d\) is the polynomial degree, \(\tau\) is the
simulation time, \(\kappa\) is the condition number, $\epsilon$ is the rejection
sampling vector, and $q$ is the number of qubits in its given preparation
unitary. Matrix entries are
approximated to the precision required by the target error.}
\label{tab:application-tcount-summary}
\end{table}

\subsection{Techniques}
\paragraph{Sparse QROM algorithms.}

To handle sparse QROMs, a natural idea is to compress the support before
loading the data. Let $S\subseteq \set{0,1}^n$ be the support set, with $|S|=s$.
In the promised sparse QROM setting, one would like to construct an injective
hash
$
h:S\to \set{0,1}^{\lceil\log s\rceil+O(1)}
$
and store the value $d_x$ at the compressed address $h(x)$. If $h(x)$ can be evaluated coherently, then on an input $x\in S$ one can first compute $h(x)$,
and then invoke a dense QROM on the $O(\log s)$-bit address $h(x)$ to load $d_x$.
The difficulty is to implement the injective map $h$ on $S$ with low \T\ count.
Two approaches to constructing a single perfect hash help explain our
multilevel construction.

\begin{itemize}
\item A first approach is random linear hashing. One may choose two random
linear maps $h_1,h_2$, where $h_1$ partitions the support $S$ into buckets and
$h_2$ is used to process labels within each bucket in parallel. This approach
is attractive because linear maps contribute no $\T$-gates. However, random linear
hashing does not give sufficiently strong worst-case control over the bucket
sizes. Processing all buckets in superposition can therefore require a large \T\ count.

\item 
Another approach is Pagh's perfect hash scheme \cite{pagh1999hash}.
In this scheme,
the hash has the form
$$
h(x)=h_1(x)\oplus g(h_2(x)),
$$
where $h_1$ and $h_2$ are linear maps, and the table $g$ is chosen
so that the final map is collision-free on the support $S$.
The difficulty is that coherent
evaluation of $h$ now requires querying $g(h_2(x))$ in superposition. This query is itself a
dense QROM on $\log s$-bit addresses, contributing an additional
$
O(\sqrt{s\log s})
$
$\T$ count overhead.

\end{itemize}

Instead of requiring all buckets to be small, we use linear hashes to
resolve only a constant fraction of the remaining support at each level.
Colliding elements are left for later
levels, so no table $g$ is needed.
At level $i$, let $s_i$ be the number of unresolved support addresses.
We choose a linear hash $h_i$ into $O(s_i)$ buckets so that a constant
fraction of these addresses lie in singleton buckets, each containing
just one unresolved address.

For each such address $x$, a dense QROM addressed by $h_i(x)$ loads $d_x$
and a flag recording that $x$ is resolved, with \T\ count
$O(\min\{s_i,\sqrt{s_i m}\})$.
If the bucket is not a singleton, no data are loaded at this level, and the
index remains unresolved. The next level is controlled by the active flag,
so it acts only on unresolved indices.
The number of unresolved indices $s_i$ decreases geometrically across levels. Thus, the total \T\ count is bounded by
$$
\sum_{i\ge 0} O\left(\min\{s_i,\sqrt{s_i m}\}\right)
=
O(\min\{s,\sqrt{sm}\}).
$$
Processing unresolved addresses recursively is related to classical
perfect-hashing schemes such as RecSplit~\cite{esposito2020recsplit}, but
here the recursion is designed to reduce the total \T\ count of the QROM circuit.

To implement sparse QROM on addresses outside $S$ as well, we also load the
support address $y$ selected by the hashing circuit and copy $d_y$ to the
output only if $y=x$. For exact implementations, we store only the address
coordinates complementary to each linear hash and reconstruct the candidate
address using Clifford gates. We also express the data words in a basis of
their linear span and choose between unary iteration and SELECT--SWAP at
each level, with linear and square-root \T-count scaling in the table size,
respectively. Together with one final address
comparison, this gives \T\ count
$O(n+\min\{s,\sqrt{s(m+\log(2^{n+1}/s))}\})$.

For approximate implementations with a fixed small constant error, we
replace the full address comparison by comparing constant-length random
linear hashes of $x$ and $y$. The hash is
chosen independently of the circuit that selects $y$. Unequal addresses have equal hashes with only a small constant
probability. Computing the hash requires only Clifford gates, and averaging
over its choice gives a channel with small constant diamond-norm error.
The resulting \T\ count is $O(\min\{s,\sqrt{sm}\})$.

\paragraph{Lower bounds}

Our lower bounds use the adaptive Clifford+\(\T\) model of
Gosset, Kothari, and Wu~\cite{david2026quantum}, which builds on the
postselected Clifford canonical form of Beverland, Campbell, Howard, and
Kliuchnikov~\cite{beverland2020lowerbounds}. An adaptive Clifford+\(\T\)
circuit with expected \T\ count $t$ that prepares $\ket\phi$ to
trace-distance error $\varepsilon$ yields a Clifford circuit with Pauli
postselections that prepares a state within trace distance $O(\sqrt\varepsilon)$ of
$\ket\phi$. Its canonical form acts only on the output qubits and $O(t)$
magic-state qubits.
We then reduce each task to state preparation and construct a large family
of instances. For approximate implementations, the chosen output states are
separated at the required error scale; for exact implementations, distinct
output states suffice. Distinct instances therefore give distinct canonical
forms, and counting these forms yields the \T-count lower bounds.

\begin{itemize}
\item  \textbf{Sparse state preparation.}
We use two hard families. First, by fixing the support and varying the
amplitudes, the problem reduces to dense state preparation on an
\(s\)-dimensional subspace. The lower bound of
Gosset, Kothari, and Wu~\cite{david2026quantum} then gives the
\(\sqrt{s\log(1/\varepsilon)}+\log(1/\varepsilon)\) terms. Second, by varying
the support itself and embedding lower-dimensional instances, we obtain the
\(\min\{s,\sqrt{s\log(2^{n+1}/s)}\}\) term.

\item \textbf{Sparse QROM.}
We fix the support $S=\{0,1\}^r\times\{0^{n-r}\}$, where
$r=\lfloor\log s\rfloor$, and use a uniform superposition over $S$ as the
input. We choose a family of data vectors that differ pairwise on a constant
fraction of $S$. Consequently,
the QROM output states have pairwise trace distance bounded below by a
positive constant. Restricting the number of varying output bits gives
$\Omega(\min\{s,\sqrt{sm}\})$ throughout the sparsity range.
For exact sparse QROM, we use indicator functions $1_S$ and apply their QROMs
to $\ket{+}^{\otimes n}\ket{0^m}$. Distinct supports give distinct exact
output states. A table with a single nonzero entry also gives an $\Omega(n)$
lower bound, which absorbs the $O(n^2)$ subtraction in the bound on $t^2$.

\item \textbf{Sparse block encoding.}
We reduce block encoding to preparation of a normalized state proportional
to the vectorization $\operatorname{vec}(B)=\sum_{i,j}B_{ij}\ket{i}\ket{j}$
of its encoded block $B$. To obtain this state, we apply the block-encoding
unitary to half of a maximally entangled state, with the ancillas initialized
and postselected in $\ket{0^a}$.
Including the output projection in the Pauli-postselected circuit makes the
canonical-form count independent of $a$.
We then construct three families of hard instances: the first varies the support, the
second varies the nonzero values, and the third is obtained from the
single-qubit state preparation hard instances of Beverland, Campbell, Howard,
and Kliuchnikov~\cite{beverland2020lowerbounds}. 

\end{itemize}

\subsection{Related works}
\paragraph{QROM.}
Prior work studies coherent classical-data access in several related
models~\cite{jaques2025qram,zhangCircuitComplexityQuantum2024a}, including
tradeoffs in query
time~\cite{giovannetti2008quantum,giovannetti2008architectures,paler2020parallelisingqram,xu2025fattreeqram,cesa2025resourcestateqram},
ancilla
size~\cite{dimatteo2019qramft,haner2022tablelookup,zhu2024qlut}, noise
resilience~\cite{hann2021resilience,cesa2025resourcestateqram}, and
architectural
assumptions~\cite{giovannetti2008quantum,giovannetti2008architectures,jaques2025qram}.
We focus on the \T\ count of compiling a fixed classical data table into a
QROM circuit. For dense QROM, Low, Kliuchnikov, and
Schaeffer~\cite{Low2024tradingtgatesdirty} introduced the SELECT-SWAP
architecture, which yields square-root \T-count scaling for dense QROM.
Motlagh and Pocrnic~\cite{motlagh2026halving} recently improved the leading
constant by about a factor of two. Low, Kliuchnikov, and Schaeffer also gave a
matching lower bound when the number of ancillas is restricted.
Prior work on QROM focused primarily on dense QROM. In this work, we study
sparse QROM and give asymptotically optimal constructions. Our sparse-QROM
lower bounds also recover dense QROM, yielding an unconditional \T-count
lower bound without restricting the number of ancillas.

\paragraph{State preparation.}

For general \(n\)-qubit states, Low, Kliuchnikov, and
Schaeffer~\cite{Low2024tradingtgatesdirty} proved an ancilla-assisted upper
bound \(\widetilde{O}(\sqrt{2^n\log(1/\varepsilon)})\) on \T\ count. Gosset,
Kothari, and Wu~\cite{david2026quantum} later showed that the optimal
worst-case \T\ count is
\(\Theta(\sqrt{2^n\log(1/\varepsilon)}+\log(1/\varepsilon))\), and their lower
bound already holds for adaptive Clifford+\(\T\) circuits.
Prior work on sparse state preparation studies circuit
size~\cite{gleinig2021state,deveras2022double,mozafari2022state,malvetti2021quantumcircuits,ramacciotti2024state,mao2024state,li2025nearly}
and circuit depth~\cite{zhang2022state,sun2023state,yeo2025state,zi2025state}.
Rupprecht and W\"olk~\cite{rupprecht2026sparsequantumstatepreparation} give a
construction with \T\ count
\(O(s+\sqrt{s\log(1/\varepsilon)}+\log(1/\varepsilon))\). Li and
Luo~\cite{li2025nearly} proved nearly optimal circuit-size bounds for sparse
state preparation. To the best of our knowledge, before the first version
of this paper, no matching asymptotic upper and lower bounds were known for
the \(\T\) count of sparse state preparation.

In concurrent independent work posted shortly after the first version of
this paper, Luo and Li~\cite{LuoLi26} obtained a constant-error \T-count
upper bound \(\widetilde O(\min\{s,n^{3/4}\sqrt{s}\})\) for sparse state
preparation using sparse Boolean function synthesis, and proved a lower
bound \(\Omega(\min\{s,\sqrt{ns}\})\) for
\(2\le s\le2^{n/2}\).

\paragraph{Block encoding.}
Circuit-level implementations of block encodings have also been studied. For
dense matrices, SELECT--SWAP-based constructions have \T\ count linear in the
matrix dimension~\cite{clader2022dense}. For sparse matrices, Gily\'en, Su,
Low, and Wiebe~\cite{GSLW19} construct block encodings from sparse query-access
oracles. Zhang and Yuan~\cite{zhangCircuitComplexityQuantum2024a} analyze the
gate complexity of implementing such oracles, obtaining nearly linear
dependence on the matrix dimension and sparsity.
More efficient circuits are known when the nonzero positions can be computed
efficiently~\cite{Sunderhauf2024blockencoding}, including banded circulant
matrices~\cite{campsExplicitQuantumCircuits2024a} and matrices with periodic
diagonal structure~\cite{zecchi2026blockencodingsparsematrices}.
Dictionary-based sparse block encodings give another way to reduce circuit
depth~\cite{Yang2025dictionarybased}. Our block-encoding results apply to
general sparse matrices specified by classical lists of nonzero positions and
values, and bound the fault-tolerant \T\ count for both uncontrolled and
controlled block encodings.




\paragraph{Organization.}
After introducing the circuit models and counting lemmas in
Section~2, we present the sparse-QROM constructions in Section~3 and then
their applications to sparse state preparation, block encoding, and related
tasks in Section~4. We defer all lower-bound arguments to Section~5.

\section*{Acknowledgments}

Tongyang Li and Xinzhao Wang are supported by the National Natural Science
Foundation of China under Grant No.~62372006.
Penghui Yao and Fengning Ou are supported by the National Natural Science Foundation of China (Grant Nos.~62332009 and 12347104), the Quantum Science and Technology--National Science and Technology Major Project (Grant No.~2021ZD0302901), the NSFC/RGC Joint Research Scheme (Grant No.~12461160276), the Natural Science Foundation of Jiangsu Province (No.~BK20243060), the Fundamental and Interdisciplinary Disciplines Breakthrough Plan of the Ministry of Education of China (No.~JYB2025XDXM118), the ``111 Center''(No.~B26023), and the Fundamental Reseach Funds for the Central Universities (Grant No.~2026300376).

OpenAI’s ChatGPT was used for language editing and for checking the clarity
and completeness of parts of the proofs. In preparing the second version,
GPT-6 Astra was used to develop the constructions and proof arguments for
the approximate sparse QROM implementation and the extension of the matching
bounds to the high-sparsity regime
$2^{(1-\delta)n}\le s\le2^n$. The authors subsequently verified, refined,
and rewrote these constructions and proof arguments.

\section{Preliminaries}

In this section, we introduce the circuit models and prove the counting
lemmas used below.

For a positive integer $N$, we write $[N]=\{0,1,\ldots,N-1\}$. For an indexed
object \(f\), such as a data vector or a function, we write
\(\supp(f)=\{x:f_x\neq 0\}\), where \(0\) denotes the all-zero string when the
entries are bit strings. For a state
\(\ket{\psi}=\sum_x\alpha_x\ket{x}\), its support is the computational-basis
set \(\{x:\alpha_x\neq0\}\). For a matrix \(A\), we write
\(\supp(A)=\{(i,j):A_{ij}\neq0\}\); when needed, we call this the entry
support of \(A\). For a probability distribution, support means the set of
outcomes with nonzero probability.

Quantum registers are denoted by sans-serif letters, such as
\(\reg{L}\) and \(\reg{S}\).

An $N\times N$ matrix is row- and column-$s$-sparse if every row and every
column contains at most $s$ nonzero entries. We write $\|A\|$ for the operator
norm of a matrix $A$, namely $\|A\|=\max_{\|x\|_2=1}\|Ax\|_2$.
For an $N\times N$ matrix $A$, we write
$\operatorname{vec}(A)=\sum_{i,j=0}^{N-1}A_{ij}\ket{i}\ket{j}$. We also write
$\|A\|_F=\big(\sum_{i,j=0}^{N-1}|A_{ij}|^2\big)^{1/2}$ for the Frobenius
norm, so that $\|\operatorname{vec}(A)\|_2=\|A\|_F$. For an operator $X$, we
write $\|X\|_1=\operatorname{Tr}\sqrt{X^\dagger X}$ for the trace norm. For a
linear map $\mathcal{M}$ on operators, we write
\[
    \|\mathcal{M}\|_\diamond
    =
    \sup_{k\ge 1}\ \sup_{X\neq 0}
    \frac{\|(\mathcal{M}\otimes I_k)(X)\|_1}{\|X\|_1}
\]
for the diamond norm, where \(I_k\) denotes the identity map of dimension $k$.

\subsection{Clifford+\T\ circuits}

\T\ count is the complexity measure used throughout this paper. We first define Clifford+\T\ circuits and then introduce the adaptive Clifford+\T\ model.

\begin{definition}[Clifford+\texorpdfstring{\(\T\)}{T} circuit]
A Clifford\,$+$\,\T\ circuit is a quantum circuit composed of Clifford gates,
\T\ gates, and clean ancillas initialized to \(\ket{0}\). The \T\ count of such
a circuit is the number of \T\ gates appearing in the circuit.
\end{definition}

A stabilizer state is a state $C\ket{0^q}$ for a Clifford unitary $C$.
We write $\ket T=T\ket+ = (\ket0+e^{i\pi/4}\ket1)/\sqrt2$ for the
single-qubit magic state used below.

We use the adaptive Clifford+\(\T\) model and the expected \T-count
notation of Gosset, Kothari, and Wu~\cite{david2026quantum}. An adaptive Clifford+\(\T\) circuit
extends the non-adaptive model by allowing single-qubit
computational-basis measurements, where later operations may depend on earlier measurement
outcomes.
If the circuit performs \(\ell\) measurements, then each measurement record
\(r=(r_1,\ldots,r_\ell)\in\{0,1\}^\ell\) fixes the subsequent operations that
are controlled on measurement outcomes. Let \(U^{(j)}(r_{<j})\) be the unitary
applied between the \((j-1)\)-st and \(j\)-th measurements after observing
\(r_{<j}=(r_1,\ldots,r_{j-1})\), and let \(\Pi^{(j)}_{r_j}\) be the projector
corresponding to outcome \(r_j\) of the \(j\)-th measurement. The corresponding
Kraus operator is
\[
    K_r
    =
    U^{(\ell+1)}(r_{\le \ell})
    \Pi^{(\ell)}_{r_\ell}
    U^{(\ell)}(r_{<\ell})
    \cdots
    \Pi^{(1)}_{r_1}
    U^{(1)} .
\]
This operator acts on the input register together with the clean ancillas, and
the output channel is
\[
    \mathcal{A}(\rho)
    =
    \sum_r
    K_r
    \bigl(\rho\otimes \ket{0^a}\!\bra{0^a}\bigr)
    K_r^\dagger .
\]
For an input state \(\rho\) on the designated input register, the expected
\T\ count of an adaptive Clifford+\T\ algorithm \(\mathcal{A}\) on input
\(\rho\) is
\[
    \mathcal{T}_{\rho}(\mathcal{A})
    :=
    \sum_r
    \Tr\!\left[
        K_r
        \bigl(\rho\otimes \ket{0^a}\!\bra{0^a}\bigr)
        K_r^\dagger
    \right]
    \T_r ,
\]
where \(\T_r\) is the number of \(\T\) gates used when the measurement record is
\(r\), and the trace term is the probability of observing \(r\).

If \(\mathcal{H}\) is a subspace of inputs, we write
\[
    \mathcal{T}_{\mathcal{H}}(\mathcal{A})
    :=
    \sup_{\rho\in\mathcal{D}(\mathcal{H})}
    \mathcal{T}_{\rho}(\mathcal{A}),
\]
where \(\mathcal{D}(\mathcal{H})\) denotes the set of density operators
supported on \(\mathcal{H}\).
For state-preparation tasks we always take the input state to be
\(\rho=\ket{0^n}\!\bra{0^n}\), and abbreviate
\(\mathcal{T}(\mathcal{A})=\mathcal{T}_{\ket{0^n}\!\bra{0^n}}(\mathcal{A})\).
All lower bounds below allow adaptive circuits and use this expected cost,
maximized over valid inputs when the task has an input register.

Below, we say that \(\mathcal{A}\) prepares the state \(\ket{\psi}\) within
trace-distance error \(\varepsilon\) if
\[
    \frac12
    \left\|
        \mathcal{A}(\ket{0^n}\!\bra{0^n})
        -
        \ket{\psi}\!\bra{\psi}\otimes
        \ket{0^a}\!\bra{0^a}
    \right\|_1
    \le \varepsilon .
\]
A Clifford circuit with \(m\) Pauli postselections, \(a\) ancillas, and
\(t\) magic states, denoted by \(\mathcal{C}\), maps
\[
    \mathcal{C}(\ket{\phi_{\rm in}}\ket{T}^{\otimes t}\ket{0^a})
    =
    \ket{\phi_{\rm out}}\ket{0^t}\ket{0^a}.
\]
It can be explicitly written as
\begin{align*}
   \ket{\phi_{\rm out}}\ket{0^t}\ket{0^a}
   \propto
   C_{m+1}(I+P_m)C_m\cdots(I+P_1)C_1
   \ket{\phi_{\rm in}}\ket{T}^{\otimes t}\ket{0^a},
\end{align*}
where each \(C_j\) is a Clifford unitary and each \(P_j\) is a Hermitian
Pauli. Here \(\propto\) denotes equality up to a nonzero scalar. We will use
the following reduction from adaptive Clifford+\(T\) circuits to Clifford
circuits with Pauli postselections from Gosset, Kothari, and Wu~\cite{david2026quantum}.
In statements involving $2t$ magic states, $t$ is an integer upper bound
on the expected \T\ count; a noninteger bound can be rounded up.

\begin{lemma}[Adaptive-circuit reduction {\cite[Claim~4.5]{david2026quantum}}]
\label{lem:gkw-adaptive-reduction}
Let \(\mathcal{A}\) be an adaptive Clifford\,$+$\,\T\ circuit that prepares an
\(n\)-qubit state \(\ket{\phi}\) within trace-distance error \(\varepsilon\)
and has expected \T\ count at most \(t\). Then, there
are a state \(\ket{\psi}\) and a Clifford circuit \(\mathcal{C}\) with Pauli
postselections and \(a\) clean ancillas such that
\[
    \mathcal{C}
    \left(
        \ket{0^n}\ket{T}^{\otimes 2t}\ket{0^a}
    \right)
    =
    \ket{\psi}\ket{0^{2t+a}},
\]
and
\[
    \frac12
    \left\|
        \ketbra{\psi}-\ketbra{\phi}
    \right\|_1
    \le
    \sqrt{6\varepsilon}.
\]
\end{lemma}

The following result was established in \cite[Section~A.7]{beverland2020lowerbounds} and
we use the formulation of \cite[Lemma~4.8]{david2026quantum}. It converts a
Clifford circuit with Pauli postselections to a canonical form. The
canonical form removes the clean ancillas and has size independent of the
number of Pauli postselections in the original circuit.

\begin{lemma}[{see \cite[Lemma~4.8]{david2026quantum}}]
\label{lem:gkw-canonical-form}
Let \(\mathcal{C}\) be a Clifford circuit with any number of Pauli
postselections and \(a\) clean ancillas. Suppose
\[
    \mathcal{C}
    \left(
        \ket{0^n}\ket{T}^{\otimes q}\ket{0^a}
    \right)
    =
    \ket{\phi}\ket{0^{q+a}}.
\]
Then there exist an \((n+q)\)-qubit Clifford unitary \(C\) and Hermitian
Paulis \(P_1,\ldots,P_q\) on \(n+q\) qubits such that
\[
    \ket{\phi}\ket{0^q}
    \propto
    C(I+P_q)\cdots(I+P_1)
    \left(\ket{0^n}\ket{T}^{\otimes q}\right).
\]
\end{lemma}

Combining \cref{lem:gkw-adaptive-reduction,lem:gkw-canonical-form} gives
\begin{corollary}
\label{lem:gkw-pauli-postselection}
Let \(\mathcal{A}\) be an adaptive Clifford\,$+$\,\T\ circuit that prepares an
\(n\)-qubit state \(\ket{\phi}\) within trace-distance error \(\varepsilon\)
and has expected \T\ count at most \(t\). Then, there
exist a state \(\ket{\psi}\), an \((n+2t)\)-qubit Clifford unitary \(C\), and
Hermitian Paulis \(P_1,\ldots,P_{2t}\) such that
\[
    \ket{\psi}\ket{0^{2t}}
    \propto
    C(I+P_{2t})\cdots(I+P_1)
    \left(\ket{0^n}\ket{T}^{\otimes 2t}\right),
\]
and
\[
    \frac12
    \left\|
        \ketbra{\psi}-\ketbra{\phi}
    \right\|_1
    \le
    \sqrt{6\varepsilon}.
\]
\end{corollary}

\subsection{QROM}

Quantum read-only memory (QROM) provides coherent access to classical data.
Given a classical table \(d\), a QROM implementation realizes
the map \(\ket{x}\ket{0}\mapsto \ket{x}\ket{d_x}\) coherently, so the
same map acts linearly when the address register is in superposition.
We mainly focus on the sparse QROM. A sparse QROM is simply a QROM whose underlying data table has only a small
support, namely at most \(s\) nonzero entries.

\begin{definition}[QROM and sparse QROM] \label{def:QROM}
     Suppose we have a data vector $(d_0,\cdots,d_{2^n - 1})$ where each $d_i$ is a binary string of length $m$.
     A \emph{QROM} for $d$ is a unitary oracle $O_d$ such that
     \begin{align*}
        O_d\ket{x}\ket{0^m} = \ket{x}\ket{d_x}, \quad \forall x \in \set{0,1}^n.
     \end{align*}
    If $|\supp(d)|\leq s$, we call $O_d$ an $s$-sparse QROM.  
\end{definition}
Low, Kliuchnikov, and Schaeffer provided a dense-QROM implementation whose \T\ count scales as the square root of the table size.

\begin{theorem}[Dense QROM \T\ count {\cite{Low2024tradingtgatesdirty}}]
\label{thm:LB-dense-qrom}
    Given $n \geq 1$, 
    suppose we have a data vector $d=(d_0,\cdots,d_{2^n-1})$ where each $d_x$ is a binary string of length $m$. Then the dense QROM oracle $O_d$ 
    can be implemented exactly with $O(\sqrt{2^n m} + n + m)$ ancilla qubits and \T\ count\footnote{The SELECT--SWAP construction has \T\ count $O(2^n/\lambda+m\lambda)$ for an integer $1\le\lambda\le2^n$. Choosing $\lambda$ within a constant factor of $\max\{1,\sqrt{2^n/m}\}$ gives the bound stated here.}
    \begin{align*}
        O\bigl(\sqrt{2^n m} + m\bigr).
    \end{align*} 
\end{theorem}

Combining this with the $O(2^n)$ \T\ count of unary iteration gives the
following bound.

\begin{corollary}\label{cor:dense-qrom-all-widths}
Given an $n$-bit address and $m$-bit data values, the dense QROM oracle
can be implemented exactly by a unitary Clifford+$\T$ circuit with
$O(\min\{2^n,\sqrt{2^n m}\})$ $\T$ gates.
\end{corollary}

\begin{proof}
The unary-iteration construction~\cite{babbush2018encoding} successively
computes the address indicators $\mathbf1[x=j]$ for $j\in[2^n]$, reusing
tests of common prefixes. Computing and uncomputing these indicators takes
$O(2^n)$ Toffoli gates and $O(n)$ clean ancillas. Writing each fixed word
$d_j$ into the output uses only CNOT gates controlled by its indicator. Thus the circuit
uses $O(2^n)$ \T\ gates, independently of $m$.
If $m\le2^n$, \Cref{thm:LB-dense-qrom} also gives
$O(\sqrt{2^n m})$ $\T$ gates. Choosing the smaller construction proves
the claim. Both constructions XOR the data into the output register.
\end{proof}

It is also useful to distinguish the promised and non-promised settings. In the
promised setting, the input is assumed to lie in the support of the sparse data
table. Without this promise, the implementation is required to act correctly on
all inputs.

\begin{definition}[Promised sparse QROM] \label{def:promised-sparse-QROM}
     Suppose we have a data vector $d = (d_0,\cdots,d_{2^n - 1})$ where each $d_i$ is a binary string of length $m$, and assume $d$ is $s$-sparse.
     A \emph{promised sparse QROM} is a unitary oracle $O_d^{\mathrm{prom}}$ such that
     \begin{align*}
        O_d^{\mathrm{prom}}\ket{x}\ket{0^m} = \ket{x}\ket{d_x},\quad  \forall x \in \supp(d).
     \end{align*}
\end{definition}

We allow exact and approximate implementations in the adaptive
Clifford+\T\ model. Such an implementation may use mid-circuit
measurements and classically controlled operations and induces a quantum
channel.

\begin{definition}[Exact and approximate QROM]
\label{def:adaptive-QROM}
 Suppose we have a data vector $(d_0,\cdots,d_{2^n - 1})$ where each $d_i$ is a binary string of length $m$.
Let
$$
\mathcal{H}
=
\spn\set{\ket{x}\ket{0^m}\ket{0^a}:x\in\set{0,1}^n}.
$$
An adaptive Clifford\,$+$\,\T\ circuit \(\mathcal{A}_d\) induces a channel
\(\mathcal{E}_d\) on these registers. An adaptive QROM implementation
with error \(\varepsilon\) for \(d\) is such a circuit \(\mathcal{A}_d\) satisfying
$$
\left\|
\left.\mathcal{E}_d\right|_{\mathcal{H}}
-
\Phi_d
\right\|_\diamond
\le \varepsilon,
$$
where $\Phi_d(\rho)=O_d\rho O_d^\dagger$ for the ideal QROM unitary $O_d$ satisfying
$$
O_d\ket{x}\ket{0^m}\ket{0^a}
=
\ket{x}\ket{d_x}\ket{0^a},
\quad \forall x\in\{0,1\}^n,
$$ 
Its expected \T\ count is \(\mathcal{T}_{\mathcal{H}}(\mathcal{A}_d)\).
The implementation is exact when $\varepsilon=0$ and approximate otherwise.
If $|\supp(d)|\le s$, we call \(\mathcal{A}_d\) an adaptive
\(s\)-sparse QROM.
\end{definition}

\begin{definition}[Exact and approximate promised sparse QROM]
\label{def:adaptive-promised-sparse-QROM}
 Suppose we have a data vector $(d_0,\cdots,d_{2^n - 1})$ where each $d_i$ is a binary string of length $m$. 
 Let $S=\supp(d)$. Let
$$
\mathcal{H}_S
=
\spn\set{\ket{x}\ket{0^m}\ket{0^a}:x\in S}.
$$
An adaptive Clifford\,$+$\,\T\ circuit \(\mathcal{A}_d\) induces a channel
\(\mathcal{E}_d\) on these registers. An adaptive promised sparse QROM
implementation with error \(\varepsilon\) for \(d\) is such a circuit
\(\mathcal{A}_d\) satisfying
$$
\left\|
\left.\mathcal{E}_d\right|_{\mathcal{H}_S}
-
\Phi_d^{\mathrm{prom}}
\right\|_\diamond
\le \varepsilon,
$$
where \(\Phi_d^{\mathrm{prom}}(\rho)=O_d^{\mathrm{prom}}\rho\br{O_d^{\mathrm{prom}}}^\dagger\)
for the ideal promised sparse QROM unitary \(O_d^{\mathrm{prom}}\) satisfying
$$
O_d^{\mathrm{prom}}\ket{x}\ket{0^m}\ket{0^a}
=
\ket{x}\ket{d_x}\ket{0^a},
\quad \forall x\in S.
$$
Its expected \T\ count is \(\mathcal{T}_{\mathcal{H}_S}(\mathcal{A}_d)\).
\end{definition}

\subsection{Counting lemmas for lower bounds}
The following counting lemmas will be used in our lower-bound arguments.
We begin with a bound on the size of an independent set.
\begin{lemma}
\label{lem:greedy-independent-set}
Let \(G\) be a finite graph with \(M\) vertices and maximum degree at most
\(D\). Then \(G\) has an independent set of size at least \(M/(D+1)\).
\end{lemma}

\begin{proof}
Repeatedly choose one remaining vertex and delete it together with all of its
neighbors. Each step deletes at most \(D+1\) vertices, and the chosen vertices
form an independent set.
\end{proof}

For a finite set $S$ and functions $c,c': S \to \set{0,1}^m$, define their Hamming
distance by
\[
    d_H(c,c') = \left|\{t\in S:c(t)\neq c'(t)\}\right| .
\]

\begin{lemma}
\label{lem:alphabet-packing}
For every finite set $S$ of size $s$ and every integer $m\ge1$, there
exists a family $\mathcal C$ of functions $S\to\set{0,1}^m$ with
minimum Hamming distance at least $s/4$ and
$\log_2|\mathcal C|=\Omega(sm)$.
\end{lemma}

\begin{proof}
Take all $2^{sm}$ functions as vertices and put an edge between two functions
if they differ in fewer than $s/4$ coordinates. A Hamming ball of radius
$s/4$ has size
\[
    \sum_{j=0}^{\lfloor s/4\rfloor}\binom sj(2^m-1)^j
    \le 2^{sH_2(1/4)}(2^m-1)^{s/4},
\]
where $H_2(x)=-x\log_2x-(1-x)\log_2(1-x)$ is the binary entropy.
Here $\sum_{j\le s/4}\binom sj\le 3^{s/4}(4/3)^s
=2^{sH_2(1/4)}$ follows by expanding $(1+1/3)^s$.
By \cref{lem:greedy-independent-set}, an independent set has logarithmic
size at least $s(m-H_2(1/4)-\frac14\log_2(2^m-1))=\Omega(sm)$.
The last bound holds for $m=1$ since $H_2(1/4)<1$, and for $m\ge2$
since $\log_2(2^m-1)\le m$.
\end{proof}

We will also need estimates for the number of supports of a given size.
\begin{fact}
\label{fact:binomial-estimate}
For integers \(1\le r\le N\),
\[
    \log\binom Nr = r\log(N/r)+O(r).
\]
\end{fact}

\begin{proof}
Note that
\[
    \left(\frac Nr\right)^r \leq   \prod_{j=0}^{r-1}\frac{N-j}{r-j} = \binom Nr  \le
    \frac{N^r}{r!}
    \le
    \left(\frac{eN}{r}\right)^r
\]
where we use the inequalities \((N-j)/(r-j)\ge N/r\) for $0 \leq j < r$ and \(r!\ge (r/e)^r\). Taking logarithms gives that \[
    r\log(N/r)
    \le
    \log\binom Nr
    \le
    r\log(eN/r).
\] 
This proves the estimate.
\end{proof}

The same graph argument gives a large family of sparse supports with small
pairwise intersections.
\begin{lemma}
\label{lem:sparse-support-packing}
There is an absolute constant $K>2$ such that, for integers
$1\le s\le N/K$, the following holds.
There exists a family \(\mathcal{F}\) of \(s\)-element subsets of $[N]$ such that
\[
    |S\cap S'|\le \frac{s}{2},
    \quad
    \forall S\neq S'\in\mathcal{F},
\]
and
\[
    \log_2|\mathcal{F}|
    =
    \Omega\!\left(s\log(N/s)\right).
\]
\end{lemma}

\begin{proof}

Consider the graph \(G\) whose vertices are the \(s\)-element subsets of
\([N]\). We put an edge between two distinct subsets \(S\) and \(S'\) when
\(|S\cap S'|>s/2\). An independent set in this graph is exactly the desired
family.

We first bound the maximum degree of \(G\). Fix a subset \(S\). If
\(|S\cap S'|=s-j\), then \(S'\) is obtained by removing \(j\) elements from
\(S\) and adding \(j\) elements from $S^c$. Thus, the number of
subsets \(S'\) with \(|S\cap S'|>s/2\), including \(S'=S\), is at most
\begin{align*}
 D+1 \le \sum_{j=0}^{\lfloor s/2\rfloor}
    \binom{s}{j}\binom{N-s}{j}
    \le
    2^s(s+1)\binom{N-s}{\lfloor s/2\rfloor}
    \le
    2^{O(s)}\left(\frac{N}{s}\right)^{s/2}.
\end{align*}
The second inequality uses $N-s\ge s$, so
$\binom{N-s}{j}$ increases for $0\le j\le s/2$.

By \cref{lem:greedy-independent-set}, \(G\) has an independent set
\(\mathcal{F}\) satisfying
\[
    \log |\mathcal{F}|
    \ge
    \log \binom{N}{s} - \log(D+1)
    \ge
    \log \binom{N}{s} - \frac{s}{2}\log(N/s)-O(s).
\]
By \cref{fact:binomial-estimate},
\[
    \log\binom{N}{s}=s\log(N/s)+O(s).
\]
Substituting this into the last display gives
\[
    |\mathcal{F}| \ge
    \exp\!\left(\frac{s}{2}\log(N/s)-O(s)\right).
\]
Choosing $K$ sufficiently large gives $\log(N/s)\ge\log K$,
so the $O(s)$ term is absorbed into $\frac{s}{2}\log(N/s)$.
Thus $|\mathcal F|=2^{\Omega(s\log(N/s))}$.

\end{proof}
For matrix supports, we must also bound the number of positions in each row
and column. For sets \(A\) and \(B\), write
\[
    A\triangle B=(A\setminus B)\cup(B\setminus A)
\]
for their symmetric difference.
\begin{lemma}
\label{lem:sparse-matrix-support-packing}
There exists a constant \(\gamma>0\)
such that the following holds for all sufficiently large \(n\). Let
\(N=2^n\), and let \(s\) be an integer satisfying
\(4\le s\le N\).
Then there exists a family \(\mathcal{P}\) of subsets of \([N]\times[N]\)
such that every \(P\in\mathcal{P}\) contains $r$
positions in each row and at most \(s\) positions in each column,
where $s/4\le r\le s$,
\[
    |P\triangle P'|\ge\gamma Ns,
    \quad
    \forall P\neq P'\in\mathcal{P},
\]
and
\[
    \log|\mathcal{P}|=\Omega(Ns\log(2N/s)).
\]
\end{lemma}

\begin{proof}
First suppose $s\ge Cn$ for a sufficiently large constant $C$.
Let $r = \lfloor s/2\rfloor$.
Choose independently and uniformly \(r\) positions in every row. For a fixed
column, its number of selected positions has distribution
\(\operatorname{Bin}(N,r/N)\), so a Chernoff bound and a union bound show that
the probability that some column contains more than \(s\) positions is at
most \(Ne^{-r/3}\). Since \(s\ge Cn\), this is at most \(1/2\) for sufficiently
large \(C\). Thus, at least half of the \(\binom Nr^N\) choices satisfy both
sparsity conditions. Moreover, by \cref{fact:binomial-estimate},
\[
    \log\binom Nr
    =
    \Omega(s\log(2N/s)).
\]
Here we use \(r=\lfloor s/2\rfloor=\Omega(s)\), \(r\le s/2\), and
\(\log(N/r)=\Theta(\log(2N/s))\) for \(s\le N\).
Hence the logarithm of the number of admissible supports is
\(\Omega(Ns\log(2N/s))\).

For $4\le s<Cn$, let $r$ be the largest power of two at most $s/2$,
so that $s/4<r\le s/2$, and put $\ell=N/r$.
Partition the rows and columns into $r$ groups of size $\ell$.
Choose an $\ell\times\ell$ permutation matrix independently in each
of the $r^2$ blocks. Each resulting support has $r$ positions per
row and column, and there are $(\ell!)^{r^2}$ choices. Since
$\log(\ell!)=\Omega(\ell\log\ell)$ for $\ell\ge2$, their logarithmic
number is $\Omega(Nr\log(N/r))=\Omega(Ns\log(2N/s))$.

It remains to impose pairwise separation. Put a graph on the admissible
supports, joining two supports whose symmetric-difference distance is at most
\(\gamma Ns\). For a fixed support, the number of admissible supports within
this distance is at most
\[
    D
    :=
    \sum_{j\le\gamma Ns}\binom{Nr}{j}\binom{N^2}{j}
    \le
    2^{O(\gamma Ns\log(2N/s)+\gamma Ns\log(1/\gamma))}.
\]
The first binomial chooses positions removed from a support and the
second chooses positions added to it. To get the displayed bound, set
\(M=\lfloor\gamma Ns\rfloor\) and use
\[
    \sum_{j\le M}\binom{R}{j}\le (M+1)(eR/M)^M,
    \qquad 1\le M\le R/2,
\]
with \(R=Nr\) and \(R=N^2\). For sufficiently small $\gamma$, both
conditions hold for all sufficiently large $n$, since $r\ge s/4$.
The bound follows from $\binom Rj\le\binom RM\le(eR/M)^M$ for $j\le M$;
the factors \(M+1\) are absorbed into the exponent above.

By \cref{lem:greedy-independent-set}, the graph has an independent set of size
\[
    |\mathcal{P}|
    \ge
    \frac{2^{\Omega(Ns\log(2N/s))}}{D+1}
    =
    2^{\Omega(Ns\log(2N/s))}
\]
where we choose \(\gamma\) sufficiently small 
so that the factor  $D+1$ is absorbed into the implicit constant in the exponent. 
This proves the claimed bound.
\end{proof}

\section{Sparse QROM}

In this section, we prove \cref{thm:Algorithm-promised-sparse-QROM,thm:Algorithm-sparse-QROM,thm:Algorithm-approximate-sparse-QROM}.

\subsection{From dense QROM to promised sparse QROM}

Under the support promise, the \T\ count can depend on $s$ and $m$ alone,
even when the input address has many more than $\log s$ bits.
\begin{restatable}{theorem}{promisedQROM}\label{thm:Algorithm-promised-sparse-QROM}
    Given a set $S \subseteq \{0,1\}^n$ of size $s\ge2$, suppose $d$ is an $s$-sparse data vector with support $S$ and message length $m$, and let $0<\delta<1$. Then, with failure probability at most $\delta$ in classical preprocessing, we can construct an exact unitary implementation of the promised sparse QROM oracle $O_d^{\mathrm{prom}}$ with \T\ count
    $$  O\br{\min\{s,\sqrt{sm}\}} $$ 
    and $ O\br{\min\{s,\sqrt{sm}\} + m + \log s}$ ancilla qubits.
    The hashes can be selected classically in time
    $$  O\br{s n \log s\log\br{\frac{\log s}{\delta}}}. $$
\end{restatable}

The construction uses multilevel hashing to reduce promised sparse QROM to
dense QROM. Let $S_0=S$, and let $S_i$ contain the support addresses
not yet handled at level $i$. A linear hash
$h_i:\mathbb F_2^n\to\mathbb F_2^{m_i}$ groups these addresses into sets
$S_i\cap h_i^{-1}(u)$, called buckets. We choose $h_i$ so that at least one
quarter of the addresses lie in buckets containing just one element.
Their data are loaded at this level, and the remaining addresses form $S_{i+1}$.
Thus \(|S_{i+1}|\le 3|S_i|/4\), and the procedure terminates after
\(O(\log s)\) levels.
Classical preprocessing checks $|S_{i+1}|\le3|S_i|/4$ for each chosen hash.

Let $L$ be such that $S_L=\emptyset$.
\Cref{alg:promised-qrom-level} implements one level; \cref{alg:promised-qrom}
applies all $L$ levels, copies the answer to the output, and uncomputes the
auxiliary registers.

For each level $i$, define $D_i(h_i(x))=(c_i(x),d_i(x))$,
where $c_i(x)\in\{0,1\}$ and $d_i(x)\in \{0,1\}^m$ are defined by
$$
c_i(x)=
\begin{cases}
1, & \text{if } |S_i\cap h_i^{-1}(h_i(x))|=1,\\
0, & \text{otherwise,}
\end{cases}, \quad d_i(x)=
\begin{cases}
d_y, & \text{if } c_i(x)=1,\\
0^m, & \text{if } c_i(x)=0.
\end{cases}
$$
where $y$ is the unique element of $S_i$ in the bucket $h_i(x)$ when
$c_i(x)=1$. Thus, $D_i$ records whether the current bucket is a singleton
bucket, and if it is, which data value should be written into the temporary
answer register.
The algorithms below specify their action on the input states used by the
construction. Since every step is reversible, these actions extend to
unitaries on the full register space.
In all algorithms below, the initial active flag $\reg{A}_0$ is prepared in
$\ket1$ and restored to $\ket0$ by $X$ gates.

\LinesNumbered
\begin{algorithm}[!htb]
\caption{Level-$i$ subroutine for promised sparse QROM}
\label{alg:promised-qrom-level}

\DontPrintSemicolon
\KwInput{The level-$i$ table $D_i$.}
\KwRegister{\(\reg{I},\reg{T},\reg{A}_i,\reg{A}_{i+1},\reg{C}_i,\reg{H}\), assumed to be in
\(\ket{x}_{\reg{I}}\ket{t}_{\reg{T}}\ket{a_i}_{\reg{A}_i}
\ket{0}_{\reg{A}_{i+1}}\ket{0}_{\reg{C}_i}\ket{0^{m_i}}_{\reg{H}}\).}
\KwOutput{The level-$i$ unitary on registers \(\reg{I},\reg{T},\reg{A}_i,\reg{A}_{i+1},\reg{C}_i,\reg{H}\).}

\AlgoDisplayStep{Hash into $\reg{H}$}{
\ket{x}_{\reg{I}}\ket{0^{m_i}}_{\reg{H}}
\mapsto
\ket{x}_{\reg{I}}\ket{h_i(x)}_{\reg{H}} .
}
\AlgoDisplayStep{Apply the dense QROM of \Cref{cor:dense-qrom-all-widths} for $D_i$ on $(\reg{A}_i,\reg{H})$ with targets $(\reg{C}_i,\reg{T})$\nllabel{lin:promised-level-dense-qrom}}{
\ket{a_i}_{\reg{A}_i}\ket{h_i(x)}_{\reg{H}}\ket{0}_{\reg{C}_i}\ket{t}_{\reg{T}}
\mapsto
\ket{a_i}_{\reg{A}_i}\ket{h_i(x)}_{\reg{H}}
\ket{a_i c_i(x)}_{\reg{C}_i}
\ket{t\oplus a_i d_i(x)}_{\reg{T}} .
}
\AlgoDisplayStep{Compute the next active flag}{
\ket{a_i}_{\reg{A}_i}\ket{a_i c_i(x)}_{\reg{C}_i}\ket{0}_{\reg{A}_{i+1}}
\mapsto
\ket{a_i}_{\reg{A}_i}\ket{a_i c_i(x)}_{\reg{C}_i}
\ket{a_i \oplus a_i c_i(x)}_{\reg{A}_{i+1}} .
}
\AlgoDisplayStep{Uncompute the hash register}{
\ket{x}_{\reg{I}}\ket{h_i(x)}_{\reg{H}}
\mapsto
\ket{x}_{\reg{I}}\ket{0^{m_i}}_{\reg{H}} .
}
\end{algorithm}
\LinesNotNumbered

\LinesNumbered
\begin{algorithm}[!htb]
\caption{Promised sparse QROM}
\label{alg:promised-qrom}
\DontPrintSemicolon
\KwInput{The tables $D_0,\dots,D_{L-1}$.}
\KwRegister{\(\reg{I},\reg{O},\reg{T},\reg{A}_0,\ldots,\reg{A}_L\), and \(\reg{C}_i,\reg{H}_i\) for
\(0\le i<L\), assumed to be in
\(\ket{x}_{\reg{I}}\ket{0^m}_{\reg{O}}\ket{0^m}_{\reg{T}}\ket{1}_{\reg{A}_0}\), with
\(\reg{A}_1,\ldots,\reg{A}_L,\reg{C}_i,\reg{H}_i\) in \(\ket0\).}
\KwOutput{The promised sparse QROM unitary on registers \(\reg{I}\) and \(\reg{O}\).}

\For{$i\gets 0$ \KwTo $L-1$}{
    Run \cref{alg:promised-qrom-level} for level $i$\;
}

\AlgoDisplayStep{Copy $\reg{T}$ into $\reg{O}$}{
\ket{t}_{\reg{T}}\ket{0^m}_{\reg{O}}
\mapsto
\ket{t}_{\reg{T}}\ket{t}_{\reg{O}} .
}

\For{$i\gets L-1$ \KwTo $0$}{
    Reverse \cref{alg:promised-qrom-level} for level $i$\;
}
\end{algorithm}
\LinesNotNumbered

We now prove correctness of \cref{alg:promised-qrom}.

Fix a promised input $x\in S$. Since the sets $S=S_0\supseteq S_1\supseteq \cdots$ record the unresolved inputs, there is a unique level $i^\star$ such that $x\in S_{i^\star}\setminus S_{i^\star+1}$. 
For $i<i^\star$, we have $a_i=1$ and $c_i(x)=0$, so $\reg{T}$ is unchanged.
At level $i^\star$, we have $c_i(x)=1$ and $d_i(x)=d_x$, so the subroutine
writes $d_x$ into $\reg{T}$ and sets $a_{i^\star+1}=0$. All later levels have
$a_i=0$ and leave $\reg{T}$ unchanged. Thus the forward pass leaves $\reg{T}$ in
$\ket{d_x}$.

The algorithm then copies $\reg{T}$ into the output register $\reg{O}$ and runs all level subroutines in reverse order, which uncomputes all other auxiliary registers while leaving the copied value in $\reg{O}$ unchanged. Hence, on every promised input $x\in S$, the final output is $\ket{x}\ket{d_x}$, so the construction indeed implements the promised sparse QROM.

To complete the proof of \cref{thm:Algorithm-promised-sparse-QROM}, we show
how to select the hashes $h_i$ and bound the cost of \cref{alg:promised-qrom}.

\begin{proof}[Proof of \cref{thm:Algorithm-promised-sparse-QROM}]
We use the probabilistic method to select the hashes. Let $S_0=S$, and
write $s_i=|S_i|$. Suppose $S_i$ has been defined. If $s_i=0$, we stop.
Otherwise, choose $m_i=\left\lceil \log (2s_i) \right\rceil$ and let $h_i:\{0,1\}^n\to\{0,1\}^{m_i}$ be a uniformly random linear hash function over $\mathbb F_2$. Define
$$
S_{i+1}
=
\{x\in S_i:\exists y\in S_i\setminus\{x\}\text{ such that }h_i(x)=h_i(y)\}.
$$
Thus, $S_{i+1}$ is exactly the set of elements of $S_i$ that lie in non-singleton buckets.

For every fixed $x\in S_i$, we have
$$
\Pr_{h_i}[x\in S_{i+1}]
\leq
\sum_{y\in S_i\setminus\{x\}}\Pr_{h_i}[h_i(x)=h_i(y)].
$$
Since $h_i$ is a uniformly random linear map, for every $y\neq x$,
$$
\Pr_{h_i}[h_i(x)=h_i(y)]
=
2^{-m_i}.
$$
Therefore,
$$
\mathbb{E}_{h_i}[|S_{i+1}|]
=
\sum_{x\in S_i}\Pr_{h_i}[x\in S_{i+1}]
\leq
s_i(s_i-1)2^{-m_i}
\leq
\frac{s_i}{2}.
$$
Hence, by Markov's inequality,
$$
    \Pr_{h_i}\br{|S_{i+1}|>3s_i/4}
    \le
    \frac{\mathbb{E}_{h_i}[|S_{i+1}|]}{3s_i/4}
    \le
    \frac{2}{3}.
$$
Therefore, a uniformly random choice of $h_i$ satisfies
$|S_{i+1}|\le 3s_i/4$ with probability at least $1/3$. In particular, a
suitable hash $h_i$ exists at every level.
Conditioned on making such a valid choice at every level, we have
$s_{i+1}\le 3s_i/4$. Hence, after $L=O(\log s)$ levels, we reach $S_L=\emptyset$.

We next analyze the classical time for selecting the hashes. At each level $i$, we sample
independent random linear hash functions $O(\log(L/\delta))$ times. Since each
sample succeeds with probability at least $1/3$, the probability that all
samples fail at this level is at most $O(\delta/L)$. With a union bound over the $L$
levels, the total failure probability is at most $\delta$.

For each sampled hash function at level $i$, we only need to check which
buckets are singleton buckets, and this can be done in time $O(s_i n \log(2s_i))$. Hence, the total classical processing time is
$$
    \sum_{i=0}^{L-1} O\br{s_i n \log(2s_i) \log(L/\delta)}
    =
    O\br{s n \log s\log(L/\delta)}.
$$

We then analyze the \T\ count. In \cref{alg:promised-qrom-level}, the hash computation and the uncomputation of $\reg{H}$ use only CNOT gates because each $h_i$ is linear over $\mathbb{F}_2$, so these steps contribute no \T\ gates. The update of the active flag also uses only CNOT gates.

The main cost comes from the dense QROM step in
\lin{promised-level-dense-qrom}. Formally, \cref{alg:promised-qrom-level} asks
for a dense QROM controlled by the active flag $\reg{A}_i$. Instead of implementing a
controlled dense QROM directly, we absorb the control bit into the input and
implement the dense QROM for
$$
\tilde{D_i} : \set{0,1}^{1+m_i}\to \set{0,1}^{1+m}, \quad (a_i,h_i(x))\mapsto a_i\cdot D_i(h_i(x)).
$$
By \cref{cor:dense-qrom-all-widths}, this costs $O\left(\min\{2^{m_i+1},\sqrt{2^{m_i+1}(m+1)}\}\right)$. Since $m_i=\left\lceil \log (2s_i) \right\rceil$, we have $2^{m_i}=O(s_i)$, and hence the dense QROM cost at level $i$ is $O(\min\{s_i,\sqrt{s_i m}\})$.

Summing over all levels, the total forward cost is
$$
\sum_{i=0}^{L-1} O(\min\{s_i,\sqrt{s_i m}\})
\leq
O\left(\min\left\{s\sum_{i=0}^{L-1}(3/4)^i,
\sqrt{sm}\sum_{i=0}^{L-1}(3/4)^{i/2}\right\}\right)
=O(\min\{s,\sqrt{sm}\}),
$$
where we use the fact $s_i\leq s(3/4)^i$ and $L = O(\log s)$.

The reverse pass contributes only a constant factor, and the copying step from $\reg{T}$ to $\reg{O}$ uses only CNOT gates. Therefore, the total \T\ count of \cref{alg:promised-qrom} is
$$O(\min\{s,\sqrt{sm}\}).$$

It remains to bound the ancilla count. The temporary answer register $\reg{T}$
contributes $m$ clean ancillas. The active-flag registers
$\reg{A}_0,\ldots,\reg{A}_L$ and singleton flags $\reg{C}_0,\ldots,\reg{C}_{L-1}$ contribute
$O(L)=O(\log s)$ additional clean ancillas. Since each level uncomputes its
hash register before the next level begins, a single hash register of size
$\max_i m_i=O(\log s)$ can be reused throughout the circuit.

For the dense QROM in \lin{promised-level-dense-qrom}, we use the
SELECT--SWAP implementation underlying \cite{Low2024tradingtgatesdirty}
when $m+1\le2^{m_i+1}$, and unary iteration otherwise. For an
input of length $m_i+1$ and output length $m+1$, these use
$O\br{\log s_i + \min\{s_i,\sqrt{s_i m}\}}$ ancillas beyond the registers
$(\reg{A}_i,\reg{H},\reg{C}_i,\reg{T})$. This workspace can also be reused across levels and
between the forward and reverse passes. Since $s_i\le s$ and $m_i=O(\log s)$
for all $i$, the total ancilla count is
$$O\br{\min\{s,\sqrt{sm}\} + m + \log s}.$$
\end{proof}

\subsection{From promised sparse QROM to sparse QROM}

To act correctly outside the support, the circuit must verify that the
address selected by the promised construction equals the input. We first
give a zero-error implementation of this check.
\begin{restatable}{theorem}{sparseQROM}\label{thm:Algorithm-sparse-QROM}
Given a set $S\subseteq\{0,1\}^n$ of size $s\ge2$, suppose $d$ is an
$s$-sparse data vector with support $S$ and message length $m$, and let
$0<\delta<1$. Then, with failure probability at most $\delta$ in classical
preprocessing, we can construct an exact unitary implementation of the
sparse QROM oracle $O_d$ with \T\ count
\[
O\left(n+\min\left\{s,\sqrt{s\left(m+\log(2^{n+1}/s)\right)}\right\}\right)
\]
and $O\left(n+m+\min\{s,\sqrt{s(m+\log(2^{n+1}/s))}\}\right)$ ancilla
qubits. The classical preprocessing takes polynomial time.
\end{restatable}

We first state a simple consequence of the promised sparse QROM construction.

\begin{proposition}
\label{prop:promised-qrom-behavior}
Given a set $S \subseteq \{0,1\}^n$ of size $s$, and suppose $d$ is an $s$-sparse data vector with support $S$ and message length $m$.
For every $x\in S$, the promised sparse QROM construction of \cref{alg:promised-qrom} satisfies
$\ket{x}\ket{0^m}\mapsto \ket{x}\ket{d_x}.$
Moreover, for every $x\notin S$, its output is either of the form
$\ket{x}\ket{0^m}\mapsto \ket{x}\ket{d_y}$ for some $y\in S$, or of the form
$\ket{x}\ket{0^m}\mapsto \ket{x}\ket{0^m}.$
\end{proposition}

\begin{proof}
Recall that, at level $i$, the promised construction writes into the temporary register only when $a_i c_i=1$.
There can be at most one such level.
Whenever a write occurs, it comes from a singleton bucket of some set $S_i\subseteq S$. Hence, the value written into the temporary register is $d_y$ for the unique element $y\in S_i$ in that bucket, and in particular $y\in S$.
If no write ever occurs, the temporary register remains initialized to $\ket{0^m}$.
\end{proof}

\begin{proof}[Proof of \cref{thm:Algorithm-sparse-QROM}]
We augment the promised output with the index itself. If $x\in S$, this
returns exactly $(d_x,x)$. If $x\notin S$,
\cref{prop:promised-qrom-behavior} implies that the promised procedure may
return some pair $(d_y,y)$ with $y\in S$, or it may return zero. We therefore
copy the returned data only when the recovered index equals the input.

We first choose an invertible linear change of
coordinates in $\mathbb F_2^m$ that maps every data word to
$(\bar d_x,0^{m-r})$, where $r\le\min\{m,s\}$ is the dimension of their
span. Both this change of coordinates and its inverse use only Clifford
gates. We perform the QROM and controlled copy on these $r$ data
coordinates.

We next define the hashes and the tables at each level. Write $s_i=|S_i|$ and
set $m_i=\min\{n,\lceil\log(2s_i)\rceil\}$. If $m_i=n$, choose $h_i$ to
be the identity, which resolves all remaining elements. Otherwise, choose
$h_i$ uniformly among full-rank linear maps to $m_i$ bits. For $x\ne y$,
\[
\Pr[h_i(x)=h_i(y)]
=\frac{2^{n-m_i}-1}{2^n-1}\le2^{-m_i}.
\]
The proof of \cref{thm:Algorithm-promised-sparse-QROM} therefore still
gives $s_{i+1}\le3s_i/4$. Extend $h_i$ by a linear map
$g_i:\mathbb F_2^n\to\mathbb F_2^{n-m_i}$ so that
$x\mapsto(h_i(x),g_i(x))$ is invertible. For a singleton bucket containing $y\in S_i$,
store $(1,g_i(y),\bar d_y)$, and store zero for other buckets. Only
$n-m_i$ additional address bits are needed, since $h_i(y)=h_i(x)$.

We now describe the algorithm. As in \cref{alg:promised-qrom-level}, the
active bit is included in the dense-table address. Its output $(c,v,w)$
is $(1,g_i(y),\bar d_y)$ when the current bucket is active and singleton,
and is zero otherwise.

\begingroup
\renewcommand{\AlgoDisplayStep}[2]{\parbox[t]{0.88\linewidth}{#1:$$#2$$}\;}
\LinesNumbered
\begin{algorithm}[!htb]
\caption{Exact sparse QROM}
\label{alg:sparse-qrom}
\DontPrintSemicolon
\KwInput{The level tables storing $(1,g_i(y),\bar d_y)$ at singleton buckets.}
\KwRegister{$\reg{I},\reg{T}_d,\reg{T}_x,\reg{F},\reg{O},\reg{A}_0,\ldots,\reg{A}_L$, and temporary registers
$\reg{H},\reg{C},\reg{V},\reg{W}$, initially in
$\ket{x}_{\reg{I}}\ket{0^r}_{\reg{T}_d}\ket{0^n}_{\reg{T}_x}\ket0_{\reg{F}}\ket{0^m}_{\reg{O}}\ket1_{\reg{A}_0}$,
with all other registers in $\ket0$, and workspace for the QROM
and controlled hash coordinates.}
\KwOutput{The sparse QROM unitary on registers $\reg{I}$ and $\reg{O}$.}
Apply the data-coordinate change to $\reg{O}$\;
\For{$i\gets0$ \KwTo $L-1$}{
\AlgoDisplayStep{Compute $h_i(x)$ and query the level table with address $(a_i,h_i(x))$}{
\ket{x}_{\reg{I}}\ket{a_i}_{\reg{A}_i}\ket0_{\reg{H}}\ket0_{\reg{C},\reg{V},\reg{W}}
\mapsto
\ket{x}_{\reg{I}}\ket{a_i}_{\reg{A}_i}\ket{h_i(x)}_{\reg{H}}\ket{c,v,w}_{\reg{C},\reg{V},\reg{W}}.
}
\AlgoDisplayStep{Compute $c h_i(x)$ and add the candidate address and data}{
\ket{t_x}_{\reg{T}_x}\ket{t_d}_{\reg{T}_d}\ket0_{\reg{A}_{i+1}}
\mapsto
\ket{t_x\oplus(h_i,g_i)^{-1}(c h_i(x),v)}_{\reg{T}_x}
\ket{t_d\oplus w}_{\reg{T}_d}\ket{a_i\oplus c}_{\reg{A}_{i+1}}.
}
Uncompute $c h_i(x)$, the level-table output, and the hash register\;
}
\AlgoDisplayStep{Compute the equality flag $\reg{F}$ between $\reg{I}$ and $\reg{T}_x$}{
\ket{x}_{\reg{I}}\ket{\tilde x}_{\reg{T}_x}\ket0_{\reg{F}}
\mapsto
\ket{x}_{\reg{I}}\ket{\tilde x}_{\reg{T}_x}\ket{\mathbf1[x=\tilde x]}_{\reg{F}}.
}
\AlgoDisplayStep{Controlled on $\reg{F}$, copy $\reg{T}_d$ into the first $r$ bits of $\reg{O}$}{
\ket{\tilde d}_{\reg{T}_d}\ket f_{\reg{F}}\ket{0^m}_{\reg{O}}
\mapsto
\ket{\tilde d}_{\reg{T}_d}\ket f_{\reg{F}}\ket{f\tilde d,0^{m-r}}_{\reg{O}}.
}
Apply the inverse data-coordinate change to $\reg{O}$\;
Reverse the equality test and then reverse all level steps\;
\end{algorithm}
\LinesNotNumbered

\endgroup

We now explain why \cref{alg:sparse-qrom} is correct. First, suppose
$x\in S$. At its unique resolving level, the table returns
$(1,g_i(x),\bar d_x)$, so the candidate registers contain
$(\bar d_x,x)$. Hence, the equality flag is set to $1$, and the
controlled copy and inverse coordinate change write $d_x$ into $\reg{O}$.

Next, suppose $x\notin S$. By \cref{prop:promised-qrom-behavior}, after
the forward pass the candidate registers contain either
$(\bar d_y,y)$ for some $y\in S$, or $(0^r,0^n)$. In the first case,
$y\ne x$, so the equality flag is zero. In the second case, the
candidate data are zero. Thus, in both cases the output register
remains $\ket{0^m}$.

Finally, the last step reverses the equality test and all level steps.
This cleans all temporary and auxiliary registers while leaving the
output register unchanged. Therefore, on every input
$x\in\{0,1\}^n$, the final state is $\ket{x}_{\reg{I}}\ket{d_x}_{\reg{O}}$.

It remains to analyze the \T\ count. By
\cref{cor:dense-qrom-all-widths}, the dense-table cost at level $i$ is
$O(\min\{s_i,\sqrt{s_i(r+n-m_i+1)}\})$. Its sum is at most $O(s)$.
For the square-root bound, use $n-m_i\le n-\log s_i$ to obtain
\[
\begin{aligned}
\sum_{i=0}^{L-1}\sqrt{s_i(r+n-m_i+1)}
&\le
\sum_{i=0}^{L-1}\sqrt{s_i\bigl(r+\log(2^{n+1}/s)+\log(s/s_i)\bigr)}\\
&=O\left(\sqrt{s\bigl(r+\log(2^{n+1}/s)\bigr)}\right).
\end{aligned}
\]
Indeed, the terms without $\log(s/s_i)$ form a geometric sum, while
$(s_i/s)\log(s/s_i)=O(\sqrt{s_i/s})$ bounds the remaining sum.
Computing and uncomputing $c h_i(x)$ uses $O(m_i)$ \T\ gates per level;
its total cost is $O(\log^2 s)=O(\sqrt s)$. All other coordinate changes
use only Clifford gates. The final comparison costs $O(n)$, and the
controlled copy costs $O(r)$. Since $r\le\min\{m,s\}$, both the copy
and the controlled hash computations are absorbed into the QROM bound.
Reversing the computation changes the cost by only a constant factor.
This proves the stated \T-count bound.

The candidate registers, temporary QROM output registers, and coordinate
changes use $O(n+m)$ ancillas. The flags use $O(\log s)$ ancillas, and
the dense-table workspace is reused across levels. Choosing unary
iteration or SELECT--SWAP as in
\cref{thm:Algorithm-promised-sparse-QROM} gives the stated ancilla bound.
As in \cref{thm:Algorithm-promised-sparse-QROM}, each sampled hash satisfies
$s_{i+1}\le3s_i/4$ with probability at least $1/3$. Testing full rank, completing the linear maps,
compressing the data words, and filling the tables take polynomial
classical time. Thus the preprocessing succeeds with probability at
least $1-\delta$ after $O(\log(L/\delta))$ trials per level.
\end{proof}

Allowing a small channel error removes the $O(n)$ cost of the full address
comparison. We compare short random linear hashes of the input and candidate
instead; the error is bounded after averaging over the random hash.
\begin{restatable}{theorem}{approximateSparseQROM}
\label{thm:Algorithm-approximate-sparse-QROM}
Given a set $S\subseteq\{0,1\}^n$ of size $s\ge2$, suppose $d$ is an
$s$-sparse data vector with support $S$ and message length $m$. For every
fixed $0<\varepsilon<1/128$, there is an approximate sparse QROM
implementation with diamond-norm error at most $\varepsilon$, using
an adaptive Clifford+\T\ circuit with
$O(\min\{s,\sqrt{sm}\})$ \T\ gates for every choice of the classical randomness and
$O(m+\log s+\min\{s,\sqrt{sm}\})$ quantum ancillas.
\end{restatable}

\begin{proof}
First fix the hashes and singleton buckets of
\cref{thm:Algorithm-promised-sparse-QROM}. Let $y(x)$ be the support
element in the first singleton bucket selected by input $x$, and leave
$y(x)$ undefined if no such bucket is selected. By
\cref{prop:promised-qrom-behavior}, $y(x)=x$ for $x\in S$ and every
defined $y(x)$ belongs to $S$. The function $y(x)$ is fixed before we choose
the verification hash $f$ below.

Compress the data words to $r\le\min\{m,s\}$ coordinates as in the
exact construction. Let $k=\min\{n,\lceil2\log(4/\varepsilon)\rceil\}$.
If $k=n$, take $f$ to be the identity. Otherwise, choose a uniformly
random linear map $f:\mathbb F_2^n\to\mathbb F_2^k$. For distinct
$x,y$, the probability that $f(x)=f(y)$ is $2^{-k}$. At level $i$,
store $(1,f(y),\bar d_y)$ in each singleton bucket.

\LinesNumbered
\begin{algorithm}[!htb]
\caption{Approximate sparse QROM}
\label{alg:approximate-sparse-qrom}
\DontPrintSemicolon
\KwInput{The fixed hashes and singleton buckets, and the compressed data $\bar d$.}
\KwRegister{$\reg{I},\reg{T}_d,\reg{T}_f,\reg{H},\reg{F},\reg{O}$, together with the promised sparse QROM
workspace, initially in
$\ket{x}_{\reg{I}}\ket{0^r}_{\reg{T}_d}\ket{0^k}_{\reg{T}_f}\ket{0^k}_{\reg{H}}\ket0_{\reg{F}}\ket{0^m}_{\reg{O}}$.}
\KwOutput{An approximate sparse QROM channel on registers $\reg{I}$ and $\reg{O}$.}
Choose $f$ independently by classical randomness\;
Apply the data-coordinate change to $\reg{O}$\;
\AlgoDisplayStep{Run \cref{alg:promised-qrom} storing $(f(y),\bar d_y)$ in each singleton bucket}{
\ket{x}_{\reg{I}}\ket{0^r}_{\reg{T}_d}\ket{0^k}_{\reg{T}_f}
\mapsto
\ket{x}_{\reg{I}}\ket{\tilde d}_{\reg{T}_d}\ket{\tilde f}_{\reg{T}_f}.
}
\AlgoDisplayStep{Compute $f(x)$ and compare it with $\reg{T}_f$}{
\ket{x}_{\reg{I}}\ket{\tilde f}_{\reg{T}_f}\ket{0^k}_{\reg{H}}\ket0_{\reg{F}}
\mapsto
\ket{x}_{\reg{I}}\ket{\tilde f}_{\reg{T}_f}\ket{f(x)}_{\reg{H}}
\ket{\mathbf1[f(x)=\tilde f]}_{\reg{F}}.
}
\AlgoDisplayStep{Controlled on $\reg{F}$, copy $\reg{T}_d$ into the first $r$ bits of $\reg{O}$}{
\ket{\tilde d}_{\reg{T}_d}\ket b_{\reg{F}}\ket{0^m}_{\reg{O}}
\mapsto
\ket{\tilde d}_{\reg{T}_d}\ket b_{\reg{F}}\ket{b\tilde d,0^{m-r}}_{\reg{O}}.
}
Apply the inverse data-coordinate change to $\reg{O}$\;
Reverse the comparison, the computation of $f(x)$, and the promised sparse QROM\;
Discard the classical random seed\;
\end{algorithm}
\LinesNotNumbered

For every choice of $f$, \cref{alg:approximate-sparse-qrom} is a clean
unitary circuit before discarding the seed. It is exact on $S$. For
$x\notin S$, it can write nonzero data only if a candidate is selected
and $f(x)=f(y(x))$. If no candidate is selected, the data register is
zero. In particular, when $f$ is the identity the circuit is exact.

We now bound the error when $k<n$. Let $P_f$ project onto the addresses
where $x\notin S$, $y(x)$ is defined, and $f(x)=f(y(x))$. Each diagonal
entry of $\mathbb E_f P_f$ is a collision probability, at most $2^{-k}$,
because $y(x)$ is fixed independently of $f$. Thus
$\|\mathbb E_f P_f\|\le2^{-k}$. Write $U_f$
for the resulting unitary for hash $f$, and use the extension
$O_d\ket{x}\ket{z}=\ket{x}\ket{z\oplus d_x}$.
The circuit $U_f$ has the same XOR form for arbitrary output-register states.
For any pure input $\ket\psi$, including an arbitrary
reference system, the two unitaries agree outside $P_f$, so
\[
\|(U_f-O_d)\ket\psi\|_2^2
\le 4 \cdot \bra\psi(P_f\otimes I)\ket\psi.
\]
The trace norm between pure-state projectors is at most twice the
Euclidean distance between their state vectors. Let
$\mathcal A=\mathbb E_f[U_f(\cdot)U_f^\dagger]$ and
$\Phi_d(\rho)=O_d\rho O_d^\dagger$. Convexity and Cauchy--Schwarz give
\[
\begin{aligned}
\big\|((\mathcal A-\Phi_d)\otimes I)(\ketbra{\psi})\big\|_1
&\le 2\mathbb E_f\|(U_f-O_d)\ket\psi\|_2\\
&\le4\sqrt{\bra\psi(\mathbb E_f P_f\otimes I)\ket\psi}
\le4\cdot2^{-k/2},
\end{aligned}
\]
where identity operators on the reference system are implicit in
$U_f-O_d$. The bound extends to mixed inputs by convexity. Maximizing
over input states and reference systems yields
\[
\|\mathcal A-\Phi_d\|_\diamond
\le4\|\mathbb E_f P_f\|^{1/2}
\le4\cdot 2^{-k/2}\le\varepsilon.
\]
It remains to bound the cost. Since $k=O(1)$,
\cref{thm:Algorithm-promised-sparse-QROM} bounds the QROM cost by
$O(\min\{s,\sqrt{s(r+k)}\})=O(\min\{s,\sqrt{sm}\})$.
Computing $f(x)$ uses only Clifford gates, the comparison costs $O(k)$,
and the controlled copy costs $O(r)$, which is absorbed since
$r\le\min\{m,s\}$. Reversing the steps changes the cost by a constant
factor. The promised sparse QROM workspace gives the stated ancilla bound.
The random bits can be generated one at a time by Hadamard gates and
measurements and then stored classically. Thus every choice of $f$ has the
claimed \T\ count and ancilla count.
\end{proof}

\section{Applications}

\subsection{Sparse state preparation}

We first formalize the sparse state preparation problem. Let
\begin{align*}
    \ket{\psi}
    =
    \sum_{x=0}^{s-1}\alpha_x\ket{a_x},
    \qquad
    a_x\in\{0,1\}^n,
\end{align*}
be an $n$-qubit state, where the basis states \(a_0,\dots,a_{s-1}\) are
distinct and \(\sum_{x=0}^{s-1}|\alpha_x|^2=1\). The sparse state preparation
task is to
construct a quantum circuit that maps $\ket{0^n}$ to $\ket{\psi}$. The
\T\ count of sparse state preparation is the \T\ count of the circuit.

Both constructions below first prepare the amplitudes on compressed
labels, load the support addresses, and erase the compressed labels.
They differ in the choice of labels and the QROM operations used to load
and erase them. We first recall the optimal dense state-preparation theorem.

\begin{theorem}[{\cite[Theorem~1.1]{david2026quantum}}]
\label{thm:dense-state-prepare}
Any $n$-qubit state can be prepared up to trace-distance error $\varepsilon$
with \T\ count and ancilla count both bounded by
$$
    O\br{\sqrt{2^n\log(1/\varepsilon)}+\log(1/\varepsilon)}.
$$
The ancillas are returned to zero.
\end{theorem}

We next show how to combine this theorem with our QROM constructions to obtain
an algorithm for sparse state preparation.

\begin{theorem}[Sparse state preparation]
\label{thm:sparse-state-preparation}
Any $s$-sparse $n$-qubit state can be prepared up to trace-distance error
$\varepsilon$ with \T\ count
\begin{align*}
    O\br{\min\{s,\sqrt{s\log(2^{n+1}/s)}\}
    +\sqrt{s\log(1/\varepsilon)}+\log(1/\varepsilon)}.
\end{align*}
\end{theorem}

We first use the consecutive labels $0,\ldots,s-1$, for which loading the
support addresses costs $O(\min\{s,\sqrt{sn}\})$. For large $s$, we instead
form the compressed labels using the multilevel hashes from the exact sparse
QROM construction. If $a_x$ is first resolved at level $i$, its label records
$i$ and $h_i(a_x)$. Since $h_i(a_x)$ provides $m_i$ linear coordinates of
$a_x$, the QROM only needs to load the remaining $n-m_i$ complementary
coordinates $g_i(a_x)$, where $(h_i,g_i)$ is an invertible linear map.
The full address $a_x$ can then be recovered using only Clifford gates.

For the consecutive-label construction, let $r=\lceil\log s\rceil$.
Define a dense data vector $d$ on $\{0,1\}^r$ of message length
$n$ by
\begin{align*}
    d_x
    =
    \begin{cases}
        a_x, & 0\le x<s,\\
        0^n, & \text{otherwise.}
    \end{cases}
\end{align*}
Also define an $s$-sparse data vector $b$ on $\{0,1\}^n$ of message length
\(r+1\) by
\begin{align*}
    b_y
    =
    \begin{cases}
        (1,x), & y=a_x \text{ for some } 0\le x<s,\\
        0^{r+1}, & \text{otherwise.}
    \end{cases}
\end{align*}
Here \(x\) is encoded in \(r\) bits, and the leading flag ensures that the
entry for \(x=0\) is not confused with the zero message.

\LinesNumbered
\begin{algorithm}[!htb]
\caption{Sparse state preparation using consecutive labels}
\label{alg:sparse-state-preparation}
\DontPrintSemicolon
\KwInput{The support indices
\(\{a_x\}_{0\le x<s}\) and amplitudes \(\{\alpha_x\}_{0\le x<s}\).}
\KwRegister{\(\reg{L},\reg{F},\reg{S}\), initialized to
\(\ket{0^r}_{\reg{L}}\ket0_{\reg{F}}\ket{0^n}_{\reg{S}}\).}
\KwOutput{Prepare $\ket{\psi}$ in the state register $\reg{S}$}

\AlgoDisplayStep{Prepare the compressed state on $\reg{L}$ using \Cref{thm:dense-state-prepare}\nllabel{lin:ssp-compressed-state}}{
\ket{0^r}_{\reg{L}}
\mapsto
\sum_{x=0}^{s-1}\alpha_x\ket{x}_{\reg{L}} .
}
\AlgoDisplayStep{Run the dense QROM of \Cref{cor:dense-qrom-all-widths} for $d$ on registers $\reg{L},\reg{S}$\nllabel{lin:ssp-dense-qrom}}{
\sum_{x=0}^{s-1}\alpha_x\ket{x}_{\reg{L}}\ket{0^n}_{\reg{S}}
\mapsto
\sum_{x=0}^{s-1}\alpha_x\ket{x}_{\reg{L}}\ket{d_x}_{\reg{S}} .
}
\AlgoDisplayStep{Set \(\reg{F}\) to \(\ket1\), then run the inverse promised sparse QROM with \Cref{alg:promised-qrom} for $b$ on registers $\reg{S},\reg{F},\reg{L}$\nllabel{lin:ssp-promised-qrom}}{
\sum_{x=0}^{s-1}\alpha_x\ket{x}_{\reg{L}}\ket0_{\reg{F}}\ket{a_x}_{\reg{S}}
\mapsto
\ket{0^r}_{\reg{L}}\ket0_{\reg{F}}\sum_{x=0}^{s-1}\alpha_x\ket{a_x}_{\reg{S}} .
}
\end{algorithm}
\LinesNotNumbered

\begin{proof}[Proof of \Cref{thm:sparse-state-preparation}]
To prove correctness of \Cref{alg:sparse-state-preparation}, it remains to
justify the erasure in \lin{ssp-promised-qrom}. Since \(b_{a_x}=(1,x)\)
for every \(0\le x<s\), the
promised sparse QROM from \Cref{alg:promised-qrom}, applied to $b$, maps
\begin{align*}
    \sum_{x=0}^{s-1}\alpha_x\ket{a_x}_{\reg{S}}\ket0_{\reg{F}}\ket{0^r}_{\reg{L}}
    \mapsto
    \sum_{x=0}^{s-1}\alpha_x\ket{a_x}_{\reg{S}}\ket1_{\reg{F}}\ket{x}_{\reg{L}}.
\end{align*}
After flipping \(\reg{F}\) from \(\ket0\) to \(\ket1\), its reverse map exactly
accomplishes \lin{ssp-promised-qrom}.

We next analyze the \T\ count. By \Cref{thm:dense-state-prepare}, the \T\ count
of \lin{ssp-compressed-state} is
\begin{align*}
    O\br{\sqrt{2^r\log(1/\varepsilon)}+\log(1/\varepsilon)}
    =
    O\br{\sqrt{s\log(1/\varepsilon)}+\log(1/\varepsilon)}.
\end{align*}
By \Cref{cor:dense-qrom-all-widths}, the \T\ count of
\lin{ssp-dense-qrom} is
$O(\min\{s,\sqrt{sn}\})$. By
\Cref{thm:Algorithm-promised-sparse-QROM}, the \T\ count of
\lin{ssp-promised-qrom} is
\(O(\min\{s,\sqrt{s(r+1)}\})=O(\min\{s,\sqrt{sn}\})\).
Therefore, the total \T\ count of \Cref{alg:sparse-state-preparation} is
\begin{align*}
    O\br{\min\{s,\sqrt{sn}\}+\sqrt{s\log(1/\varepsilon)}+\log(1/\varepsilon)}.
\end{align*}
For $s\le n$, this gives the required $O(s)$ support cost. For
$n<s\le2^{n/2}$, we have $\log(2^{n+1}/s)=\Theta(n)$, so it also gives
the claimed bound in this range.

It remains to consider $s>\max\{n,2^{n/2}\}$. We label each support address
by the level at which it first lies in a singleton bucket and its hash value
at that level. After recovering the address, we recompute the active and
singleton flags to identify this level and erase the label.
Apply the full-rank hashes $h_i$ and complementary coordinates $g_i$ from
the exact sparse QROM construction to $S_0=\{a_x:0\le x<s\}$.
Write $s_i=|S_i|$, $m_i=\min\{n,\lceil\log(2s_i)\rceil\}$,
and let $R_i=S_i\setminus S_{i+1}$ be the elements resolved at level $i$.
Set $p_i=\sum_{j<i}2^{m_j}$ and, in this construction, take
$r=\lceil\log(\sum_{i<L}2^{m_i})\rceil$.
Since $s_i\le(3/4)^i s$, we have $2^r=O(s)$.
The intervals $[p_i,p_i+2^{m_i})$ are disjoint, and each $p_i$ is a
multiple of $2^{m_i}$ because the $m_i$ are nonincreasing.
Thus each interval consists of a fixed prefix
followed by an arbitrary $m_i$-bit suffix. For $x\in R_i$, define the
compressed label $y(x)=p_i+h_i(x)$. These labels are distinct.
\Cref{alg:sparse-state-preparation-hashing} uses these labels to prepare
the state.

\LinesNumbered
\begin{algorithm}[!htb]
\caption{Sparse state preparation using multilevel hashing}
\label{alg:sparse-state-preparation-hashing}
\DontPrintSemicolon
\KwInput{The support indices \(a_x\), amplitudes \(\alpha_x\), error
\(\varepsilon\), and precomputed \(h_i,g_i,R_i,p_i\).}
\KwRegister{\(\reg{L},\reg{S}\), initialized to \(\ket{0^r}_{\reg{L}}\ket{0^n}_{\reg{S}}\), and clean
work registers \(\reg{F},\reg{H},\reg{G},\reg{A}_i,\reg{C}_i\). At level \(i\), \(\reg{H},\reg{G}\) use
\(m_i,n-m_i\) qubits, respectively.}
\KwOutput{Prepare \(\ket{0^r}_{\reg{L}}\ket{\psi}_{\reg{S}}\) to trace-distance error
\(\varepsilon\), with clean work registers.}

\AlgoDisplayStep{Prepare the compressed state to error \(\varepsilon\) using
\Cref{thm:dense-state-prepare}\nllabel{lin:ssp-hash-prepare}}{
\ket{0^r}_{\reg{L}}\mapsto\sum_{x=0}^{s-1}\alpha_x\ket{y(a_x)}_{\reg{L}}.
}
\For{$i\gets0$ \KwTo $L-1$}{
Compare the first \(r-m_i\) bits of \(\reg{L}\) with those of
\(p_i\), storing the equality flag in \(\reg{F}\). Controlled on \(\reg{F}\), copy
the last \(m_i\) bits of \(\reg{L}\) into \(\reg{H}\)\nllabel{lin:ssp-hash-prefix}\;
Run the dense QROM of \Cref{cor:dense-qrom-all-widths} with
address \((\reg{F},\reg{H})\) and target \(\reg{G}\), storing \(g_i(x)\) at
\((1,h_i(x))\) for \(x\in R_i\) and zero elsewhere\nllabel{lin:ssp-hash-qrom}\;
Apply \((h_i,g_i)^{-1}\) to \(\reg{H}\reg{G}\) and XOR the resulting
address into \(\reg{S}\). Undo the linear map, QROM, suffix copy, and prefix
comparison to clear \(\reg{F},\reg{H},\reg{G}\)\nllabel{lin:ssp-hash-copy}\;
}
Set \(\reg{A}_0\) to \(\ket1\). Compute the active flags \(\reg{A}_i\)
and singleton flags \(\reg{C}_i\) from \(\reg{S}\) by the forward levels of
\Cref{alg:promised-qrom-level}, omitting the data register
\(\reg{T}\)\nllabel{lin:ssp-hash-flags}\;
\For{$i\gets0$ \KwTo $L-1$}{
Compute \(h_i(x)\) in \(\reg{H}\) from \(\ket{x}_{\reg{S}}\).
Controlled on \(\reg{C}_i\), XOR the fixed prefix of \(p_i\) and the suffix
\(h_i(x)\) into \(\reg{L}\). Uncompute \(\reg{H}\)\nllabel{lin:ssp-hash-erase}\;
}
Reverse the flag computation and reset \(\reg{A}_0\) to \(\ket0\)\;
\end{algorithm}
\LinesNotNumbered

For each supported address $x\in R_i$, only level $i$ has $\reg{F}=1$ in the
first loop. At this level, $\reg{H}\reg{G}$ contains $(h_i(x),g_i(x))$, so
\lin{ssp-hash-copy} writes $x$ into $\reg{S}$. At every other level, $\reg{H}\reg{G}$ is zero
and $\reg{S}$ is unchanged. In the second loop, $\reg{C}_i$ is the unique nonzero
singleton flag, so \lin{ssp-hash-erase} erases precisely $y(x)$.
All work registers are uncomputed, giving
\[
    \ket{y(x)}_{\reg{L}}\ket0_{\reg{S}}
    \longmapsto\ket{y(x)}_{\reg{L}}\ket{x}_{\reg{S}}
    \longmapsto\ket0_{\reg{L}}\ket{x}_{\reg{S}}.
\]
Thus the ideal compressed state is mapped to $\ket\psi$.

Since $2^r=O(s)$, \lin{ssp-hash-prepare} has \T\ count
$O(\sqrt{s\log(1/\varepsilon)}+\log(1/\varepsilon))$.
The QROM in \lin{ssp-hash-qrom} loads only the $n-m_i$ bits of $g_i(x)$,
since the label already contains $h_i(x)$. By
\Cref{cor:dense-qrom-all-widths}, its cost is
$O(\sqrt{2^{m_i}(n-m_i)})$ per level.
Using $n-m_i\le\log(2^n/s_i)$, the same geometric-sum estimate as in the proof of
\Cref{thm:Algorithm-sparse-QROM} gives
\[
    \sum_{i<L}\sqrt{2^{m_i}(n-m_i)}
    =O\!\left(\sqrt{s\log(2^{n+1}/s)}\right).
\]
The one-bit QROMs used to compute and uncompute the flags in
\lin{ssp-hash-flags} cost $O(\sum_i\sqrt{2^{m_i}})=O(\sqrt{s})$.
The prefix tests and controlled copies in
\lin{ssp-hash-prefix} and \lin{ssp-hash-erase} cost
$O(nL)=O(n\log(s+1))=O(\sqrt{s})$, since $s>2^{n/2}$.
All remaining operations use only Clifford gates. Also $s>n$ implies
$\sqrt{s\log(2^{n+1}/s)}\le s$. This proves the required support cost
for the remaining range.

In both constructions, the compressed state is prepared to error
$\varepsilon$, and the subsequent operations are unitary and map its
ideal value to $\ket\psi$ with clean ancillas. Invariance of trace
distance under unitaries therefore gives the stated error, including
the work registers. Combining the support cost with the compressed
state-preparation cost proves the theorem.
\end{proof}

\subsection{Block encoding of sparse matrices}

\label{subsec:sparse-block-encoding}

In this subsection, we focus on the \T\ count of block encodings for sparse
matrices. We use the following definition.
\begin{definition}[Block encoding]
    Let $A$ be an operator acting on $n$ qubits. A unitary $U$ is an
    $(\alpha,a,\varepsilon)$-block-encoding of $A$ if
    \[
        \left\|
            A-\alpha
            \left(\bra{0^a}\otimes I\right)
            U
            \left(\ket{0^a}\otimes I\right)
        \right\|\le \varepsilon .
    \]
\end{definition}

An adaptive implementation must realize $U$ exactly; the error
$\varepsilon$ is measured on the encoded matrix. Its \T\ count is the
expected number of \T\ gates, maximized over input states.
The parameter $a$ counts the ancillas projected onto zero in the encoded
block. Circuit implementations may use additional workspace returned to zero;
this workspace is included in the total qubit count.
Throughout this subsection, $0<\varepsilon_{\rm BE}\le1/2$.

Let \(N=2^n\) and \(A=(a_{ij})\in\mathbb{R}^{N\times N}\) be
row- and column-\(s\)-sparse with \(|a_{ij}|\le1\). Set
\(\bar{s}=2^{\lceil\log s\rceil}\), so \(s\le \bar{s}\le2s\), and pad every row and
column to length \(\bar{s}\).

For each row \(i\), choose distinct labels \(r_{i,0},\ldots,r_{i,\bar{s}-1}\in[N]\)
whose set contains the row support \(\{j:a_{ij}\neq0\}\); any remaining slots
are filled with zero entries of that row. Similarly, for each column \(j\),
choose distinct labels \(c_{0,j},\ldots,c_{\bar{s}-1,j}\in[N]\) whose set contains
the column support \(\{i:a_{ij}\neq0\}\), with any remaining slots filled by
zero entries of that column. For the inverse promised sparse QROM calls below, we
associate the address \((i,r_{i,k})\) with the nonzero message \((1,k)\), and
similarly associate \((c_{\ell,j},j)\) with \((1,\ell)\). The leading flag
ensures that the message is nonzero also when \(k=0\) or \(\ell=0\).
Choose a precision parameter \(b\) and set \(B=2^b\). For \(i,j\in[N]\), let
\[
    m_{ij}=\lfloor B|a_{ij}|\rfloor,
    \qquad
    \tau_{ij}=\mathbf{1}[a_{ij}<0]
\]
be the magnitude and sign data for \(a_{ij}\). The register storing $m_{ij}$
has $b+1$ bits, since $m_{ij}=B$ is possible. Define
\begin{align}
\label{eq:def-tilde-a}
    \tilde a_{ij}=(-1)^{\tau_{ij}}\frac{m_{ij}}{B}.
\end{align}
Following \cite[Lemma~27]{GSLW19}, we express the matrix entries as inner
products. In \algo{sparse-block-encoding-unitary}, we construct
two state-preparation unitaries \(U_r\ket{i}\ket{0^a}=\ket{\psi_i}\) and
\(U_c\ket{j}\ket{0^a}=\ket{\phi_j}\) such that
\(\braket{\psi_i}{\phi_j}=\tilde a_{ij}/\bar{s}\). Then
\(U_{\rm BE}=U_r^\dagger U_c\) is an \(\bar{s}\)-normalized block encoding of
\(\tilde A=(\tilde a_{ij})\). The subroutines in the algorithm are written on
the input states used in the block encoding. Since every step is reversible,
they extend to unitaries on the full register space.

\begingroup
\renewcommand{\AlgoDisplayStep}[2]{%
\parbox[t]{0.88\linewidth}{#1:%
\setlength{\abovedisplayskip}{4pt}%
\setlength{\belowdisplayskip}{4pt}%
\setlength{\abovedisplayshortskip}{4pt}%
\setlength{\belowdisplayshortskip}{4pt}%
$$#2$$}\;}
\SetKwProg{Subroutine}{Subroutine}{:}{}

\LinesNumbered
\begin{algorithm}[!htb]
\caption{Sparse block-encoding unitary}
\label{algo:sparse-block-encoding-unitary}
\DontPrintSemicolon
\KwInput{The padded row and column tables \(r_{i,k}\), \(c_{\ell,j}\), and the
value data \(m_{ij},\tau_{ij}\).}
\KwRegister{\(\reg{X},\reg{Y},\reg{K},\reg{L},\reg{\Sigma},\reg{M},\reg{T},\reg{Q},\reg{F}\).}
\KwOutput{The unitary \(U_{\rm BE}=U_r^\dagger U_c\).}

\Subroutine{\(U_r\) acting on \(\ket{i}_{\reg{X}}\ket{0}_{\reg{K},\reg{\Sigma} ,\reg{Y},\reg{Q},\reg{F}}\)}{
\AlgoDisplayStep{Prepare the uniform superposition on $\reg{K},\reg{\Sigma}$\nllabel{lin:be-row-uniform}}{
\ket0_{\reg{K}}\ket0_{\reg{\Sigma}}\mapsto
\frac1{\sqrt{\bar s B}}\sum_{k=0}^{\bar s-1}\sum_{\sigma=0}^{B-1}
\ket{k}_{\reg{K}}\ket{\sigma}_{\reg{\Sigma}}.
}
\AlgoDisplayStep{Run the dense QROM of \Cref{thm:LB-dense-qrom} on $\reg{X},\reg{K},\reg{Y}$\nllabel{lin:be-row-dense}}{
\ket{i}_{\reg{X}}\ket{k}_{\reg{K}}\ket0_{\reg{Y}}\mapsto
\ket{i}_{\reg{X}}\ket{k}_{\reg{K}}\ket{r_{i,k}}_{\reg{Y}}.
}
Set \(\reg{F}\) to \(\ket{1}\)\;
\AlgoDisplayStep{Run the inverse promised sparse QROM from \Cref{alg:promised-qrom} on $\reg{X},\reg{Y},\reg{F},\reg{K}$\nllabel{lin:be-row-inverse}}{
\ket{i}_{\reg{X}}\ket{r_{i,k}}_{\reg{Y}}\ket1_{\reg{F}}\ket{k}_{\reg{K}}\mapsto
\ket{i}_{\reg{X}}\ket{r_{i,k}}_{\reg{Y}}\ket0_{\reg{F}}\ket0_{\reg{K}}.
}
Set \(\reg{Q}\) to \(\ket{1}\)\nllabel{lin:be-row-q}\;
}
\BlankLine
\Subroutine{\(U_c\) acting on \(\ket{j}_{\reg{Y}}\ket{0}_{\reg{L},\reg{\Sigma},\reg{X},\reg{M},\reg{T},\reg{Q},\reg{F}}\)}{
\AlgoDisplayStep{Prepare the uniform superposition on $\reg{L},\reg{\Sigma}$\nllabel{lin:be-column-uniform}}{
\ket0_{\reg{L}}\ket0_{\reg{\Sigma}}\mapsto
\frac1{\sqrt{\bar s B}}\sum_{\ell=0}^{\bar s-1}\sum_{\sigma=0}^{B-1}
\ket{\ell}_{\reg{L}}\ket{\sigma}_{\reg{\Sigma}}.
}
\AlgoDisplayStep{Run the dense QROM of \Cref{thm:LB-dense-qrom} on $\reg{L},\reg{Y},\reg{X}$\nllabel{lin:be-column-support}}{
\ket{\ell}_{\reg{L}}\ket{j}_{\reg{Y}}\ket0_{\reg{X}}\mapsto
\ket{\ell}_{\reg{L}}\ket{j}_{\reg{Y}}\ket{c_{\ell,j}}_{\reg{X}}.
}
\AlgoDisplayStep{Run the dense QROM of \Cref{thm:LB-dense-qrom} on $\reg{L},\reg{Y},\reg{M},\reg{T}$\nllabel{lin:be-column-value}}{
\ket{\ell}_{\reg{L}}\ket{j}_{\reg{Y}}\ket0_{\reg{M}}\ket0_{\reg{T}}\mapsto
\ket{\ell}_{\reg{L}}\ket{j}_{\reg{Y}}\ket{m_{c_{\ell,j},j}}_{\reg{M}}\ket{\tau_{c_{\ell,j},j}}_{\reg{T}}.
}
\AlgoDisplayStep{Compare $\sigma$ with $m_{c_{\ell,j},j}$ and write the result into $\reg{Q}$\nllabel{lin:be-column-compare}}{
\ket{\sigma}_{\reg{\Sigma}}\ket{m_{c_{\ell,j},j}}_{\reg{M}}\ket0_{\reg{Q}}\mapsto
\ket{\sigma}_{\reg{\Sigma}}\ket{m_{c_{\ell,j},j}}_{\reg{M}}\ket{\mathbf1[\sigma<m_{c_{\ell,j},j}]}_{\reg{Q}}.
}
Apply the phase \((-1)^{\tau_{c_{\ell,j},j}}\) controlled by \(\reg{T}\)\nllabel{lin:be-column-phase}\;
Run the same dense QROM again to erase \(\reg{M}\) and \(\reg{T}\)\nllabel{lin:be-column-value-uncompute}\;
Set \(\reg{F}\) to \(\ket{1}\)\;
\AlgoDisplayStep{Run the inverse promised sparse QROM from \Cref{alg:promised-qrom} on $\reg{X},\reg{Y},\reg{F},\reg{L}$\nllabel{lin:be-column-inverse}}{
\ket{c_{\ell,j}}_{\reg{X}}\ket{j}_{\reg{Y}}\ket1_{\reg{F}}\ket{\ell}_{\reg{L}}\mapsto
\ket{c_{\ell,j}}_{\reg{X}}\ket{j}_{\reg{Y}}\ket0_{\reg{F}}\ket0_{\reg{L}}.
}
}
\BlankLine
Apply \(U_c\), then apply \(U_r^\dagger\), obtaining
\(U_{\rm BE}=U_r^\dagger U_c\)\nllabel{lin:be-compose}\;

\end{algorithm}
\LinesNotNumbered
\endgroup

We now show that this unitary block-encodes \(A\).

\begin{theorem}
\label{thm:block-encoding}
Let $A\in\mathbb{R}^{2^n\times 2^n}$ be a row- and column-$s$-sparse matrix
with $|a_{ij}|\le 1$, and set \(\bar{s}=2^{\lceil\log s\rceil}\). There exists an
$(\bar{s}, O(n+\log(s/\varepsilon_{\rm BE})), \varepsilon_{\rm BE})$
block encoding of $A$ with \T\ count and qubit count both bounded by 
$$
O\br{\sqrt{2^n s(n+\log(1/\varepsilon_{\rm BE}))}+\log(1/\varepsilon_{\rm BE})}.
$$
\end{theorem}

\begin{proof}
We first prove correctness. The row and column subroutines in
\algo{sparse-block-encoding-unitary}, from \lin{be-row-uniform} to
\lin{be-row-q} and from \lin{be-column-uniform} to \lin{be-column-inverse},
return the work registers \(\reg{K},\reg{L},\reg{M},\reg{T},\reg{F}\) to \(\ket{0}\). Omitting these clean
registers and writing the remaining registers in the order \(\reg{X},\reg{Y},\reg{\Sigma},\reg{Q}\),
they implement
\[
\begin{aligned}
    U_r\ket{i}_{\reg{X}}\ket{0}_{\reg{Y}\reg{\Sigma} \reg{Q}}
    &=
    \frac{1}{\sqrt{\bar{s}B}}
    \sum_{k=0}^{\bar{s}-1}\sum_{\sigma=0}^{B-1}
    \ket{i}_{\reg{X}}\ket{r_{i,k}}_{\reg{Y}}\ket{\sigma}_{\reg{\Sigma}}\ket{1}_{\reg{Q}},\\
    U_c\ket{j}_{\reg{Y}}\ket{0}_{\reg{X}\reg{\Sigma} \reg{Q}}
    &=
    \frac{1}{\sqrt{\bar{s}B}}
    \sum_{\ell=0}^{\bar{s}-1}\sum_{\sigma=0}^{B-1}
    (-1)^{\tau_{c_{\ell,j},j}}
    \ket{c_{\ell,j}}_{\reg{X}}\ket{j}_{\reg{Y}}\ket{\sigma}_{\reg{\Sigma}}
    \ket{\mathbf{1}[\sigma<m_{c_{\ell,j},j}]}_{\reg{Q}} .
\end{aligned}
\]
By \lin{be-compose}, the block-encoding unitary is
\(U_{\rm BE}=U_r^\dagger U_c\). Here the input system is placed in $\reg{X}$ for
$U_r$ and in $\reg{Y}$ for $U_c$; fixed SWAP gates identify both inputs with the
same system register. These Clifford gates are implicit in the notation.
Its top-left block has entries
\[
    (\bra{i}\bra{0^a})U_r^\dagger U_c(\ket{j}\ket{0^a}).
\]
For $i,j\in[N]$, if $(i,j)\in\supp(A)$, then there is a unique
pair $(k,\ell)$ such that $r_{i,k}=j$ and $c_{\ell,j}=i$.
Therefore,
\begin{align*}
    (\bra{i}\bra{0^a})U_r^\dagger U_c(\ket{j}\ket{0^a})
    &=
    \frac{1}{\bar{s}B}
    \sum_{k,\ell=0}^{\bar{s}-1}\sum_{\sigma=0}^{B-1}
    (-1)^{\tau_{c_{\ell,j},j}}
    \braket{i, r_{i,k}, \sigma, 1}
    {c_{\ell,j}, j, \sigma, \mathbf{1}[\sigma<m_{c_{\ell,j},j}]} \\
    &=
    \frac{1}{\bar{s}B}
    \sum_{k,\ell=0}^{\bar{s}-1}\sum_{\sigma=0}^{B-1}
    (-1)^{\tau_{c_{\ell,j},j}}
    \mathbf{1}[r_{i,k}=j]\,
    \mathbf{1}[c_{\ell,j}=i]\,
    \mathbf{1}[\sigma<m_{c_{\ell,j},j}] \\
    &=
    \frac{1}{\bar{s}B}
    \sum_{\sigma=0}^{B-1}
    (-1)^{\tau_{ij}}
    \mathbf{1}[\sigma<m_{ij}]
    =
    \frac{m_{ij}(-1)^{\tau_{ij}}}{\bar{s}B}
    =
    \frac{\tilde a_{ij}}{\bar{s}}.
\end{align*}
If \((i,j)\notin\supp(A)\), then \(m_{ij}=0\). Thus the factor
\(\mathbf{1}[\sigma<m_{ij}]\) is zero, even if \((i,j)\) appears through
padded zero entries, and the overlap is zero. Therefore
$$
(\bra{0^a}\otimes I)U_{\rm BE}(\ket{0^a}\otimes I)=\tilde A/\bar{s},
$$
where $\tilde A=(\tilde a_{ij})$ is the entrywise approximation defined in
\eq{def-tilde-a}.

Since $|\tilde a_{ij}-a_{ij}|\le 2^{-b}$ for every entry, and
$A-\tilde A$ is still row- and column-$s$-sparse, every row and every column of
$A-\tilde A$ has at most $s$ nonzero entries, each of magnitude at most
$2^{-b}$. The operator norm is bounded by the geometric mean of the maximum
absolute row sum and column sum. Therefore,
$$
\|A-\tilde A\|
\le
\sqrt{\Bigl(\max_i\sum_j|a_{ij}-\tilde a_{ij}|\Bigr)
       \Bigl(\max_j\sum_i|a_{ij}-\tilde a_{ij}|\Bigr)}
\le
s2^{-b}
\le
\varepsilon_{\rm BE}.
$$
Hence $U_{\rm BE}$ is an
$(\bar{s}, O(n+\log s+b), \varepsilon_{\rm BE})$-block-encoding of $A$. Since
\(b=\lceil \log_2(4s/\varepsilon_{\rm BE})\rceil\), this is
$$
(\bar{s}, O(n+\log(s/\varepsilon_{\rm BE})), \varepsilon_{\rm BE}).
$$

We next analyze the \T\ count. In the row subroutine,
\lin{be-row-uniform} and \lin{be-row-q} cost no \T\ gates. The dense QROM in
\lin{be-row-dense} has table size \(2^n \bar{s}=O(2^n s)\) and message length \(n\),
so by \Cref{thm:LB-dense-qrom} it costs \(O(\sqrt{2^n s\,n})\). The promised
sparse QROM in \lin{be-row-inverse} has data support size \(2^n \bar{s}=O(2^n s)\) and
message length \(O(\log s)\), so by \Cref{thm:Algorithm-promised-sparse-QROM} it costs
\(O(\sqrt{2^n s\log s})\). Since \(\log s\le n\), the row subroutine has \T\
count \(O(\sqrt{2^n s\,n})\).

In the column subroutine, \lin{be-column-uniform} and
\lin{be-column-phase} cost no \T\ gates, while the comparator in
\lin{be-column-compare} costs \(O(b)\) \T\ gates by standard linear-size
adder constructions~\cite{gidney2018halvingcostof}. The dense QROM in
\lin{be-column-support} and the promised sparse QROM in
\lin{be-column-inverse} have the same total cost
\(O(\sqrt{2^n s\,n})\) as in the row subroutine. The dense value QROM in
\lin{be-column-value}, together with its inverse in
\lin{be-column-value-uncompute}, has table size \(2^n \bar{s}=O(2^n s)\) and message length
\(b+2\), including the sign bit. By \Cref{thm:LB-dense-qrom}, these two lines cost
\(O(\sqrt{2^n s\,b}+b)\) \T\ gates. Hence, the column subroutine has total
\T\ count
\[
    O\br{ \sqrt{2^n s\,n} + \sqrt{2^n s\,b} + b }.
\]

Finally, substituting
\(b=\lceil \log_2(4s/\varepsilon_{\rm BE})\rceil
=O(\log(s/\varepsilon_{\rm BE}))\), and using $\log s\le n$, gives the claimed \T\ count
\[
O\br{
    \sqrt{2^n s(n+\log(1/\varepsilon_{\rm BE}))}
    +\log(1/\varepsilon_{\rm BE})
}.
\]
The same estimates also bound the number of qubits. The data and precision
registers contribute \(O(n+b)\) qubits. The QROMs in
\lin{be-row-dense} and \lin{be-column-support} use
\(
    O(\sqrt{2^n s\,n}+n+\log s)
\)
ancillas by \Cref{thm:LB-dense-qrom}. The QROM in
\lin{be-column-value}, and its inverse in \lin{be-column-value-uncompute}, use
\(
    O(\sqrt{2^n s\,b}+n+\log s+b)
\)
ancillas. The promised sparse QROMs in \lin{be-row-inverse} and
\lin{be-column-inverse} have support size \(2^n \bar{s}=O(2^n s)\) and message
length \(O(\log s)\), so \Cref{thm:Algorithm-promised-sparse-QROM} gives
\(
    O(\sqrt{2^n s\log s}+n+\log s)
\)
ancillas for each of these two calls. Substituting
\(b=O(\log(s/\varepsilon_{\rm BE}))\) gives the claimed qubit bound.
\end{proof}

We will use the following consequence of the construction of
Kim and Laakkonen~\cite[Theorem~6]{KimLaakkonen2025controlled}.
\begin{theorem}[Controlled Clifford+\T\ circuits]
\label{thm:controlled-clifford-t}
Let \(U\) be an \(n\)-qubit Clifford\,$+$\,\T\ circuit with \T\ count \(t\). Then the
controlled unitary
\(
    \ket{0}\bra{0}\otimes I+
    \ket{1}\bra{1}\otimes U
\)
can be implemented with \T\ count \(O(t+n)\).
\end{theorem}

We apply \Cref{thm:controlled-clifford-t} to combine block encodings of
$\operatorname{Re}(A)$ and $\operatorname{Im}(A)$.
\begin{lemma}[Complex entries]
\label{lem:complex-sparse-block-encoding}
Let \(A\in\mathbb{C}^{2^n\times 2^n}\) be a row- and column-\(s\)-sparse matrix
with \(|a_{ij}|\le1\), and set \(\bar{s}=2^{\lceil\log s\rceil}\). There exists a
\((2\bar{s},O(n+\log(s/\varepsilon_{\rm BE})),\varepsilon_{\rm BE})\)-block-encoding
of \(A\) with \T\ count and qubit count both bounded by
\[
O\br{
    \sqrt{2^n s(n+\log(1/\varepsilon_{\rm BE}))}
    +\log(1/\varepsilon_{\rm BE})
}.
\]
\end{lemma}

\begin{proof}
Let \(U_R\) and \(U_I\) be the block encodings obtained by applying
\Cref{thm:block-encoding} to \(\operatorname{Re}(A)\) and
\(\operatorname{Im}(A)\), respectively, each with error
\(\varepsilon_{\rm BE}/2\). Thus, for the same value of \(a\),
\[
    (\bra{0^a}\otimes I)U_R(\ket{0^a}\otimes I)
    =\frac{\widetilde A_R}{\bar{s}},
    \qquad
    (\bra{0^a}\otimes I)U_I(\ket{0^a}\otimes I)
    =\frac{\widetilde A_I}{\bar{s}},
\]
where
\[
    \|\widetilde A_R-\operatorname{Re}(A)\|\le\varepsilon_{\rm BE}/2,
    \qquad
    \|\widetilde A_I-\operatorname{Im}(A)\|\le\varepsilon_{\rm BE}/2.
\]
Using one additional selector qubit, define
\[
    U_A
    =
    (H\otimes I)
    \left(
        \ket{0}\bra{0}\otimes U_R
        +
        \ket{1}\bra{1}\otimes iU_I
    \right)
    (H\otimes I).
\]
Then
\[
\begin{aligned}
    (\bra{0^{a+1}}\otimes I)U_A(\ket{0^{a+1}}\otimes I)=
    \frac{1}{2}
    (\bra{0^a}\otimes I)
    \left(U_R+iU_I\right)
    (\ket{0^a}\otimes I) =
    \frac{\widetilde A_R+i\widetilde A_I}{2\bar{s}}.
\end{aligned}
\]
Thus \(U_A\) is a \(2\bar{s}\)-normalized block encoding of \(A\) with error at most
\(\varepsilon_{\rm BE}\). The Hadamards and the phase gate
$\operatorname{diag}(1,i)$ on the additional qubit are Clifford gates, and
$U_A$ uses one controlled call to each of
\(U_R\) and \(U_I\). By \Cref{thm:controlled-clifford-t}, each controlled real block
encoding has the same asymptotic \T\ count as in \Cref{thm:block-encoding},
because the qubit count is bounded by the same expression. Hence $U_A$ has
the claimed \T\ count and qubit count.
\end{proof}

The same asymptotic bound holds for controlled complex block encodings.
\begin{theorem}[Controlled sparse block encoding]
\label{thm:controlled-block-encoding}
Let \(A\in\mathbb{C}^{2^n\times 2^n}\) be a row- and column-\(s\)-sparse matrix
with \(|a_{ij}|\le1\). Then the controlled
version of the block encoding in \Cref{lem:complex-sparse-block-encoding} can
be implemented with \T\ count
\[
O\br{
    \sqrt{2^n s(n+\log(1/\varepsilon_{\rm BE}))}
    +\log(1/\varepsilon_{\rm BE})
}.
\]
\end{theorem}

\begin{proof}
By \Cref{lem:complex-sparse-block-encoding}, the uncontrolled block encoding
has \T\ count and qubit count both bounded by
\[
O\br{
    \sqrt{2^n s(n+\log(1/\varepsilon_{\rm BE}))}
    +\log(1/\varepsilon_{\rm BE})
}.
\]
Applying \Cref{thm:controlled-clifford-t} gives the claimed bound.
\end{proof}

\subsubsection{Quantum singular value transformation}

We next present quantum singular value transformation (QSVT), introduced by
Gily\'en, Su, Low, and Wiebe~\cite{GSLW19}, as a representative application of
block encoding. QSVT generalizes the qubitization-based approach to optimal
Hamiltonian simulation of Low and Chuang~\cite{low2019hamiltonian}, supports
matrix inversion and related quantum matrix-arithmetic tasks, and has also
been used in quantum machine-learning settings such as regularized least
squares~\cite{Chakraborty2022QRLs}.

For a matrix \(A\) with singular value decomposition
\(A=\sum_i \sigma_i\ket{u_i}\bra{v_i}\), define the singular value
transformation of \(A\) with a degree-\(d\) polynomial \(P\) by
\[
    P^{(\mathrm{SV})}(A)
    =
    \begin{cases}
        \sum_i P(\sigma_i)\ket{u_i}\bra{v_i}, & \text{if } d \text{ is odd}, \\
        \sum_i P(\sigma_i)\ket{v_i}\bra{v_i}, & \text{if } d \text{ is even}.
    \end{cases}
\]
Given a block encoding of \(A\), QSVT implements a block encoding of
\(P^{(\mathrm{SV})}(A)\) using \(O(\deg(P))\) calls to the block encoding,
provided that \(P\) is bounded on \([-1,1]\) and has definite parity.

\begin{theorem}[QSVT,~{\cite[Corollary~11]{GSLW19}}]
\label{thm:qsvt}
Let \(A\) be an operator with \(\|A\|\le \alpha\), and let \(U\) be an
\((\alpha,a,0)\)-block-encoding of \(A\). Define
\(
    \Pi = \ket{0^a}\bra{0^a}\otimes I
\)
and
\(
    R_\Pi = 2\Pi-I.
\)
Let \(P\in \mathbb{R}[x]\) be a polynomial of degree \(d\) such that
\(|P(x)|\le 1\) for all \(x\in[-1,1]\) and
\(P(-x)=(-1)^dP(x)\). Then there exists a phase vector
\(\Phi=(\phi_1,\ldots,\phi_d)\in\mathbb{R}^d\) such that, defining
\begin{align}
\label{eq:qsvt-sequence}
    U_\Phi =
    \begin{cases}
        e^{i\phi_1 R_\Pi} U
        {\textstyle\prod_{j=1}^{(d-1)/2}}
        \left( e^{i\phi_{2j} R_\Pi} U^\dagger
        e^{i\phi_{2j+1} R_\Pi} U \right),
        & \text{if } d \text{ is odd}, \\
        {\textstyle\prod_{j=1}^{d/2}}
        \left( e^{i\phi_{2j-1} R_\Pi} U^\dagger
        e^{i\phi_{2j} R_\Pi} U \right),
        & \text{if } d \text{ is even},
    \end{cases}
\end{align}
and defining \(U_{-\Phi}\) analogously by negating all phases, we have
\[
\begin{aligned}
    (\bra{0^{a+1}}\otimes I)
    (H\otimes I)
    \left(
        \ket{0}\bra{0}\otimes U_\Phi
        +
        \ket{1}\bra{1}\otimes U_{-\Phi}
    \right)
    (H\otimes I)
    (\ket{0^{a+1}}\otimes I)
    =
    P^{(\mathrm{SV})}(A/\alpha).
\end{aligned}
\]
Hence, the unitary
\[
    (H\otimes I)
    \left(
        \ket{0}\bra{0}\otimes U_\Phi
        +
        \ket{1}\bra{1}\otimes U_{-\Phi}
    \right)
    (H\otimes I)
\]
is a \((1,a+1,0)\)-block-encoding of \(P^{(\mathrm{SV})}(A/\alpha)\).
\end{theorem}

We will also use the following robustness bound for singular value
transformation, which controls the error caused by applying QSVT to an
approximate block encoding.
\begin{lemma}[Robustness of singular value transformation, {\cite[Lemma~22]{GSLW19}}]
\label{lem:qsvt-robustness}
Let \(P\in\mathbb{C}[x]\) be a degree-\(d\) polynomial satisfying the
boundedness and parity assumptions needed for singular value transformation,
and let \(\|A\|,\|\tilde A\|\le 1\). Then
\[
    \big\|
        P^{(\mathrm{SV})}(A)-P^{(\mathrm{SV})}(\tilde A)
    \big\|
    \le
    4d\sqrt{\|A-\tilde A\|}.
\]
\end{lemma}

We will use two elementary implementation costs below. A reflection about
\(\ket{0^a}\) can be implemented by a multi-controlled Toffoli and Clifford
gates with \(O(a)\) \T\ gates~\cite[Corollary~7.4]{adriano1995elementarygates}.
A controlled single-qubit rotation with precision \(\rho\) can be synthesized
using \(O(\log(1/\rho))\) \T\ gates~\cite{RossSelinger2014}.

Combining \Cref{lem:complex-sparse-block-encoding,thm:controlled-block-encoding}
with QSVT gives the following \T-count upper bound for a sparse matrix.
In this and the following applications, matrix entries are approximated to the
precision required by the final target error.
\begin{lemma}\label{lem:block-encoding-for-QSVT}
Given a row- and column-\(s\)-sparse matrix
\(A \in \mathbb{C}^{2^n \times 2^n}\) with \(\|A\|\le 1\), let
\(P\in \mathbb{R}[x]\) be a degree-\(d\)
polynomial such that \(|P(x)|\le 1\) for all \(x\in[-1,1]\) and
\(P(-x)=(-1)^dP(x)\). Let \(\alpha\) be the normalization factor in the
complex sparse block encoding of \(A\), so \(\alpha=\Theta(s)\). Then there is a
\((1,O(n+\log(sd/\varepsilon)),\varepsilon)\)-block-encoding of
\(P^{(\mathrm{SV})}(A/\alpha)\) with \T\ count
\[
    O\left(d \left(
        \sqrt{2^n s(n+\log(sd/\varepsilon))}+\log(sd/\varepsilon)
    \right)\right).
\]
\end{lemma}

\begin{proof}
Set \(\delta_{\rm BE}=O(\varepsilon^2/d^2)\). By
\Cref{lem:complex-sparse-block-encoding,thm:controlled-block-encoding}, one
controlled or uncontrolled use of the sparse block encoding of \(A\) with error
\(\delta_{\rm BE}\) has \T\ count
\[
    O\left(
        \sqrt{2^n s(n+\log(sd/\varepsilon))}+\log(sd/\varepsilon)
    \right),
\]
and has \(O(n+\log(sd/\varepsilon))\) block-encoding ancillas, in addition
to the clean QROM workspace. Let \(\tilde A/\alpha\) be
the exact matrix encoded by this unitary. Then
\(\|A-\tilde A\|\le\delta_{\rm BE}\), so
\(\|A/\alpha-\tilde A/\alpha\|\le\delta_{\rm BE}/\alpha\le\delta_{\rm BE}\).
Thus
\Cref{lem:qsvt-robustness} implies
\[
    \big\|
        P^{(\mathrm{SV})}(A/\alpha)
        -
        P^{(\mathrm{SV})}(\tilde A/\alpha)
    \big\|
    =
    O(\varepsilon).
\]
By \Cref{thm:qsvt}, QSVT uses \(O(d)\) calls to the block encoding and its
inverse, together with \(O(d)\) unitaries of the form \(e^{i\phi R_\Pi}\), where
\(\Pi=\ket{0^a}\bra{0^a}\otimes I\). To implement \(e^{i\phi R_\Pi}\), compute
whether the \(a\) block-encoding ancillas are in \(\ket{0^a}\) using a
multi-controlled Toffoli, apply one controlled single-qubit rotation, and
uncompute. Using the elementary costs above and synthesizing the rotation to
precision \(O(\varepsilon/d)\), each \(e^{i\phi R_\Pi}\) costs
\(O(n+\log(sd/\varepsilon))\) \T\ gates, which is bounded by the displayed per-call
cost. Multiplying by \(O(d)\) gives the claimed bound.
\end{proof}

\subsubsection{Sparse Hamiltonian simulation}
The following theorem follows by substituting our sparse block encoding into
the Hamiltonian-simulation theorem of Gily\'en, Su, Low, and
Wiebe~\cite[Theorem~58]{GSLW19}.
\begin{theorem}[Sparse Hamiltonian simulation]
\label{thm:sparse-hamiltonian-simulation}
Let \(H\in\mathbb{C}^{2^n\times 2^n}\) be a Hermitian row- and
column-\(s\)-sparse matrix with \(|H_{ij}|\le 1\). For any
\(t\in\mathbb{R}\) with \(|t|\ge 1\), and any \(0<\varepsilon<1\), there is a
\((1,O(n+\log(s|t|/\varepsilon)),\varepsilon)\)-block-encoding of \(e^{-iHt}\) with
\T\ count
\[
    O\left(
        (s|t|+\log(1/\varepsilon))
        \left(
            \sqrt{2^n s(n+\log(s|t|/\varepsilon))}
            + \log\frac{s|t|}{\varepsilon}
        \right)
    \right).
\]
\end{theorem}

\begin{proof}
Let \(\alpha\) be the normalization factor of the sparse block encoding of
\(H\), so \(\alpha\le 4s\), and set \(\tau=\alpha |t|=O(s|t|)\). In
Theorem~58, Gily\'en, Su, Low, and
Wiebe~\cite{GSLW19} give a block encoding of \(e^{-iHt}\) using
\[
    q=
    O\left(
        \tau+
        \frac{\log(1/\varepsilon)}
        {\log\!\left(e+\log(1/\varepsilon)/\tau\right)}
    \right)
    =
    O(s|t|+\log(1/\varepsilon))
\]
queries to the controlled block encoding of \(H\) and its inverse,
plus \(O(q(n+\log(q/\varepsilon)))\) additional one- and two-qubit gates.

We implement each controlled query using
\Cref{thm:controlled-block-encoding}. Choose the block-encoding error
\(\delta_{\rm BE}=O(\varepsilon/|t|)\). Then
\(\log(s/\delta_{\rm BE})=O(\log(s|t|/\varepsilon))\), and one controlled query has
\T\ count
\[
    O\left(
        \sqrt{2^n s(n+\log(s|t|/\varepsilon))}
        +
        \log(s|t|/\varepsilon)
    \right).
\]
Rounding conjugate entries consistently makes the exactly encoded matrix
$\tilde H$ Hermitian. Replacing \(H\) by \(\tilde H\) contributes error
$O(\varepsilon)$, since
\(\|e^{-iHt}-e^{-i\tilde Ht}\|\le |t|\,\|H-\tilde H\|\).
The additional operations in Theorem~58 consist of \(O(q)\) reflections about
block-encoding ancillas and controlled single-qubit rotations. Using the
elementary costs above, and synthesizing each rotation to precision
\(O(\varepsilon/q)\), these operations contribute \(O(q(n+\log(q/\varepsilon)))\) \T\
gates. This is absorbed in
\[
    O\left(
        q\left(
            \sqrt{2^n s(n+\log(s|t|/\varepsilon))}
            + \log(s|t|/\varepsilon)
        \right)
    \right),
\]
because \(n\le \sqrt{2^n s n}\). Substituting the bound on \(q\) gives the
claimed \T\ count.
\end{proof}

\subsubsection{Solving sparse linear systems}

We next estimate the \T\ count of the quantum linear-system solver of
Costa, An, Sanders, Su, Babbush, and Berry~\cite{Costa2022optimalQLSS}, using
dense state preparation and our sparse block-encoding construction.

\begin{theorem}[Sparse quantum linear-system solving]
\label{thm:sparse-qlss}
Let $A\in\mathbb{C}^{2^n\times 2^n}$ be an invertible row- and
column-$s$-sparse matrix with $\|A\|\le 1$ and $\|A^{-1}\|=\kappa$. Let
\(\ket{b}\) be an arbitrary classically specified \(n\)-qubit input state.
For any \(0<\varepsilon<1\), the normalized solution state
\(A^{-1}\ket{b}/\|A^{-1}\ket{b}\|\) can be prepared to error \(\varepsilon\)
with \T\ count
\begin{align*}
    O\left(
        s\kappa\log(1/\varepsilon)
        \left(
            \sqrt{2^n s(n+\log(s\kappa/\varepsilon))}
            +\log(s\kappa/\varepsilon)
        \right)
    \right).
\end{align*}
\end{theorem}

\begin{proof}
Costa, An, Sanders, Su, Babbush, and Berry~\cite[Theorem~11]{Costa2022optimalQLSS}
solve the linear-system problem using
\(O(\tau\log(1/\varepsilon))\) oracle calls when the normalized matrix has
singular values in $[1/\tau,1]$. Their query bound includes the block
encoding of $A$, the state-preparation oracle for $\ket b$, their inverses,
and controlled versions of each. In our implementation the
block encoding normalizes \(A\) by a factor \(\alpha\le 4s\), so
\(\tau=\|(A/\alpha)^{-1}\|=O(s\kappa)\).

Choose the sparse block-encoding error and the state-preparation error to be
\(\eta_A=\Theta(\varepsilon/\kappa)\) and
\(\eta_b=\Theta(\varepsilon/\kappa)\), respectively. Suppose the implemented
block encoding encodes \(\tilde A\) with \(\|A-\tilde A\|\le\eta_A\), and the
implemented state-preparation oracle prepares \(\ket{\tilde b}\) with
\(\|\ket{\tilde b}-\ket b\|\le\eta_b\), after choosing the global phase of
\(\ket{\tilde b}\). This vector-norm guarantee follows, up to a constant
factor, from the trace-distance guarantee for pure states. Then \(\tilde A\)
is invertible for a
sufficiently small constant in \(\eta_A\), and
\(\|\tilde A^{-1}\|=O(\kappa)\). Writing
\(x=A^{-1}\ket b\) and \(\tilde x=\tilde A^{-1}\ket{\tilde b}\), we have
\[
    \|\tilde x-x\|
    \le
    \|\tilde A^{-1}\|\,\|\ket{\tilde b}-\ket b\|
    +
    \|\tilde A^{-1}\|\,\|A-\tilde A\|\,\|x\|.
\]
Since \(\|A\|\le1\), we have \(\|x\|\ge1\). Hence
\[
    \|\tilde x-x\|
    \le
    O(\kappa\eta_b+\kappa\eta_A\|x\|)
    =
    O(\varepsilon)\|x\|.
\]
It follows that
\[
    \left\|
        \frac{\tilde x}{\|\tilde x\|}
        -
        \frac{x}{\|x\|}
    \right\|
    =
    O(\varepsilon).
\]

We now bound the \T\ cost of these oracles. By
\Cref{thm:controlled-block-encoding}, a controlled block encoding of \(A\) with
error \(O(\eta_A)\) has \T\ count
\[
    C_A=
    O\!\left(
        \sqrt{2^n s(n+\log(s\kappa/\varepsilon))}
        +\log(s\kappa/\varepsilon)
    \right).
\]
Let $U_b$ prepare $\ket{\tilde b}$ from $\ket{0^n}$ with its ancillas
returned to zero. By the dense state-preparation theorem
\Cref{thm:dense-state-prepare} and \Cref{thm:controlled-clifford-t},
controlled or uncontrolled uses of \(U_b\) cost
\[
    C_b=
    O\!\left(
        \sqrt{2^n\log(s\kappa/\varepsilon)}
        +\log(s\kappa/\varepsilon)
    \right)
\]
\T\ gates. Its qubit count is also $O(C_b)$, so adding a control incurs no
asymptotic overhead.
Since \(s\ge1\), \(C_b=O(C_A)\). The non-oracle
operations in the construction of Costa, An, Sanders, Su, Babbush, and Berry
consist of
\(O(\tau\log(1/\varepsilon))\) reflections about block-encoding ancillas and
controlled single-qubit rotations. Using the elementary costs above, and
synthesizing each rotation to precision
\(O(\varepsilon/(\tau\log(1/\varepsilon)))\), these costs are
\(O(C_A)\) per use. The one-time preparation of \(\ket b\) is also absorbed by
the same bound. Multiplying \(C_A\) by
\(O(\tau\log(1/\varepsilon))=O(s\kappa\log(1/\varepsilon))\) gives the claimed
\T\ count.
\end{proof}

\subsection{Quantum rejection sampling}
In quantum rejection sampling~\cite{QRS}, we change the amplitudes of a label
register while preserving the unknown states associated with its labels.
We are given an oracle that prepares
\(\ket{\pi^{\xi}}=\sum_{k=1}^{N}\pi_k\ket{\xi_k}\ket{k}\),
where \(N\) is the number of labels, \(\pi\) is a
known nonnegative amplitude vector, and the states \(\{\ket{\xi_k}\}_{k=1}^{N}\) are
unknown. Given another known nonnegative amplitude vector \(\sigma\), the goal is to prepare
\(\ket{\sigma^{\xi}}=\sum_{k=1}^{N}\sigma_k\ket{\xi_k}\ket{k}\).
Here $\|\pi\|_2=\|\sigma\|_2=1$.

The method of Ozols, Roetteler, and Roland~\cite{QRS} first replaces the
target amplitude vector $\sigma$ by a vector \(\epsilon\) satisfying
\(0\le \epsilon_k\le \pi_k\) for all \(k\), and then prepares
\(\ket{\epsilon^{\xi}}=\frac{1}{\|\epsilon\|_2}
\sum_{k=1}^{N}\epsilon_k\ket{\xi_k}\ket{k}\).
The postselection described below succeeds with probability $\|\epsilon\|_2^2$, and we require
$\braket{\sigma^\xi}{\epsilon^\xi}\ge p$ for a prescribed $p\in[0,1]$.
Since the label states \(\ket{k}\) are orthogonal, $
    \braket{\sigma^\xi}{\epsilon^\xi}
    = \left\langle \sigma, \epsilon/\|\epsilon\|_2 \right\rangle$.
Thus $\epsilon$ is chosen by solving
\begin{align*}
     \max \ \ & \|\epsilon\|_2^2
    \\
    \text{subject to} \ \ & 0\le \epsilon_k\le \pi_k,  \quad \text{for all } k,
    \\
    & \braket{\sigma^\xi}{\epsilon^\xi}\ge p.
\end{align*}
An explicit solution is given in \cite{QRS}.

To prepare \(\ket{\epsilon^\xi}\) for the chosen $\epsilon$, define
\[
    \rho_k
    =
    \begin{cases}
        \epsilon_k/\pi_k, & \pi_k>0,\\
        0, & \pi_k=0.
    \end{cases}
\]
Starting from \(\ket{\pi^\xi}\ket0\), we apply the map
\[
    \ket{k}\ket0
    \mapsto
    \ket{k}
    \left(
        \sqrt{1-\rho_k^2}\ket0+\rho_k\ket1
    \right),
\]
which produces
\[
    \sum_{k=1}^{N}\pi_k\ket{\xi_k}\ket{k}
    \left(
        \sqrt{1-\rho_k^2}\ket0+\rho_k\ket1
    \right)
    =
    \sum_{k=1}^{N}\ket{\xi_k}\ket{k}
    \left(
        \sqrt{\pi_k^2-\epsilon_k^2}\ket0+\epsilon_k\ket1
    \right).
\]
Postselecting the last qubit on \(\ket1\) succeeds with probability
\(\|\epsilon\|_2^2\) and yields \(\ket{\epsilon^\xi}\ket1\). Standard
amplitude amplification then boosts the success probability using
\(O(1/\|\epsilon\|_2)\) applications of the unitary before postselection and its inverse, together
with the black-box state-preparation unitary for \(\ket{\pi^\xi}\).

We count the \T\ gates for this amplitude-transformation step, treating
the preparation of $\ket{\pi^\xi}$ as a given
oracle. Suppose $\pi$ is $s$-sparse with support $S$. Only labels $k\in S$
occur, so promised sparse QROM suffices to supply the ratios $\rho_k$.
The registers containing $\ket{\xi_k}$ are left unchanged.

\begin{theorem}
\label{thm:sparse-qrs-filter}
Suppose \(S=\supp(\pi)\), \(s=|S|\), and the given state-preparation
unitary acts on $q$ qubits, including its workspace. Then,
for any \(0<\delta<1\), the postselection subroutine for $\ket{\epsilon^\xi}$ can be implemented
with the following guarantee: conditioned on successful postselection, its
output state has trace-distance error \(O(\delta/\|\epsilon\|_2)\). Its \T\
count is
\[
    O\!\left(
        \min\{s,\sqrt{s\log(2/\delta)}\}
        +\log(2/\delta)
    \right).
\]
Consequently, the state \(\ket{\epsilon^\xi}\) can be prepared to
trace-distance error \(O(\delta/\|\epsilon\|_2)\) using \T\ count
\[
    O\!\left(
        \frac{
            \min\{s,\sqrt{s\log(2/\delta)}\}
            +\log(2/\delta)+q
        }{\|\epsilon\|_2}
    \right),
\]
together with
\(O(1/\|\epsilon\|_2)\) calls to the black-box preparation of
\(\ket{\pi^\xi}\) and its inverse.
\end{theorem}

Let \(m=\lceil\log(2/\delta)\rceil\). For each \(k\in S\), define
\[
    \rho_k^+
    =
    \begin{cases}
        \rho_k+4\cdot 2^{-m}, & \rho_k\le 1-4\cdot 2^{-m},\\
        1, & \rho_k>1-4\cdot 2^{-m}.
    \end{cases}
\]
Set \(d_k=\lfloor 2^m\rho_k^+\rfloor\), \(\tilde\rho_k=2^{-m}d_k\), and
\(\tilde\epsilon_k=\pi_k\tilde\rho_k\). We use the comparator-based method of
Sanders, Low, Scherer, and Berry~\cite{sanders2019black-box}. The algorithm is
written in \Cref{alg:sparse-qrs}.

\LinesNumbered
\begin{algorithm}[!htb]
\caption{Quantum rejection sampling}
\label{alg:sparse-qrs}
\DontPrintSemicolon
\KwInput{The support \(S=\supp(\pi)\) and the integers
\(\{d_k\}_{k\in S}\).}
\KwRegister{\(\reg{I}_1,\reg{I}_2,\reg{D},\reg{R},\reg{F}\), initialized to $\sum_{k\in S}\pi_k\ket{\xi_k}_{\reg{I}_1}\ket{k}_{\reg{I}_2}
\ket{0^{m+1}}_{\reg{D}}\ket{0^m}_{\reg{R}}\ket0_{\reg{F}}$.}
\KwOutput{Prepare \(\ket{\tilde\epsilon^\xi}\).}

\AlgoDisplayStep{Load \(d_k\) with the promised sparse QROM in \Cref{alg:promised-qrom}\nllabel{lin:qrs-load-amplitudes}}{
\ket{k}_{\reg{I}_2}\ket{0^{m+1}}_{\reg{D}}
\mapsto
\ket{k}_{\reg{I}_2}\ket{d_k}_{\reg{D}} .
}
\AlgoDisplayStep{Prepare $\reg{R}$ in a uniform superposition\nllabel{lin:qrs-uniform-ref}}{
\ket{0^m}_{\reg{R}}
\mapsto
2^{-m/2}\sum_{z} \ket z_{\reg{R}} .
}
\AlgoDisplayStep{Compare \(z\) with \(d_k\) and write the result to \(\reg{F}\)\nllabel{lin:qrs-compare}}{
\ket{d_k}_{\reg{D}}\ket z_{\reg{R}}\ket0_{\reg{F}}
\mapsto
\ket{d_k}_{\reg{D}}\ket z_{\reg{R}}\ket{\mathbf 1[z<d_k]}_{\reg{F}} .
}

\AlgoDisplayStep{Apply \(H^{\otimes m}\) to \(\reg{R}\) and uncompute \(\reg{D}\)\nllabel{lin:qrs-uncompute-amplitudes}}{
\ket{k}_{\reg{I}_2}\ket{d_k}_{\reg{D}}
\mapsto
\ket{k}_{\reg{I}_2}\ket{0^{m+1}}_{\reg{D}} .
}

\AlgoDisplayStep{Postselect on \(\reg{R}=\ket{0^m}\) and \(\reg{F}=\ket1\)\nllabel{lin:qrs-postselect}}{
\sum_{k\in S}\pi_k\ket{\xi_k}\ket{k}
\left(
    \tilde\rho_k\ket{0^m}_{\reg{R}}\ket1_{\reg{F}}+\ket{\text{garbage}_k}_{\reg{R}\reg{F}}
\right)
\Longrightarrow
\ket{\tilde\epsilon^\xi}\ket{0^m}_{\reg{R}}\ket1_{\reg{F}} .
}
\end{algorithm}
\LinesNotNumbered

\begin{proof}
We analyze the \T\ count first. By \Cref{thm:Algorithm-promised-sparse-QROM},
\lin{qrs-load-amplitudes} has \T\ count
\(O(\min\{s,\sqrt{sm}\})\).
The same bound applies to \lin{qrs-uncompute-amplitudes}. The preparation and
unpreparation of the uniform superposition on $\reg{R}$ cost no \T\ gates. The comparator in \lin{qrs-compare} compares the
\((m+1)\)-bit integer \(d_k\) with the \(m\)-bit integer \(z\), and can be implemented with \(O(m)\) \T\ gates.  Hence, one
application of the postselection subroutine costs
\(O(\min\{s,\sqrt{sm}\}+m)\).

We next analyze the correctness. After
\lin{qrs-load-amplitudes}, \lin{qrs-uniform-ref}, and \lin{qrs-compare}, the
state is
\[
    \sum_{k\in S}\pi_k\ket{\xi_k}\ket{k}\ket{d_k}_{\reg{D}}
    2^{-m/2}
    \left(
        \sum_{z=0}^{d_k-1}\ket z_{\reg{R}}\ket1_{\reg{F}}
        +
        \sum_{z=d_k}^{2^m-1}\ket z_{\reg{R}}\ket0_{\reg{F}}
    \right).
\]
Applying \(H^{\otimes m}\) to \(\reg{R}\), the amplitude of
\(\ket{0^m}_{\reg{R}}\ket1_{\reg{F}}\) is
\(2^{-m}\sum_{z=0}^{d_k-1}1=2^{-m}d_k=\tilde\rho_k\).
All other components are orthogonal to the postselected subspace
\(\ket{0^m}_{\reg{R}}\ket1_{\reg{F}}\), and we denote them by \(\ket{\text{garbage}_k}_{\reg{R}\reg{F}}\).  Therefore, postselecting on
\(\reg{R}=\ket{0^m}\) and \(\reg{F}=\ket1\) produces
\[
    \ket{\tilde\epsilon^\xi}
    =
    \frac{1}{\|\tilde\epsilon\|_2}
    \sum_{k\in S}\tilde\epsilon_k\ket{\xi_k}\ket{k}.
\]

It remains to bound the error. If
\(\rho_k\le 1-4\cdot 2^{-m}\), then $ \rho_k+3\cdot 2^{-m} \le \tilde\rho_k \le \rho_k+4\cdot 2^{-m}$.
If \(\rho_k>1-4\cdot 2^{-m}\), then \(\rho_k^+=1\), \(d_k=2^m\), and
\(\tilde\rho_k=1\). Thus, in all cases,
$ \rho_k\le \tilde\rho_k $ and $|\rho_k-\tilde\rho_k|\le 4\cdot 2^{-m}=O(\delta)$.

The resulting unnormalized state
\(\|\tilde\epsilon\|_2\ket{\tilde\epsilon^\xi}\) therefore differs from the ideal
state \(\|\epsilon\|_2\ket{\epsilon^\xi}\) by
\[
\left\|
    \sum_{k\in S}\pi_k(\rho_k-\tilde\rho_k)\ket{\xi_k}\ket{k}
\right\|_2^2
=
\sum_{k\in S}\pi_k^2|\rho_k-\tilde\rho_k|^2
\le
O(\delta^2)\sum_{k\in S}\pi_k^2
\le
O(\delta^2).
\]

Since \(\tilde\rho_k\ge\rho_k\) for every \(k\), we have
\(\|\tilde\epsilon\|_2\ge\|\epsilon\|_2\). Hence,
\[
\left\|\ketbra{\epsilon^\xi}-\ketbra{\tilde\epsilon^\xi}\right\|_1
\le 2\left\|\ket{\epsilon^\xi}-\ket{\tilde\epsilon^\xi}\right\|_2
=O(\delta/\|\epsilon\|_2).
\]

Finally, \lin{qrs-postselect} succeeds with probability
\(\|\tilde\epsilon\|_2^2\). Since  \(\|\tilde\epsilon\|_2\ge\|\epsilon\|_2\), the standard
amplitude amplification uses
\(O(1/\|\tilde\epsilon\|_2)=O(1/\|\epsilon\|_2)\) applications of the subroutine.
The reflection about $\ket{0^m}_{\reg{R}}\ket1_{\reg{F}}$ costs $O(m)$ \T\ gates.
The reflection about the initial state uses the given preparation unitary,
its inverse, and a reflection about $\ket{0^q}$, costing another $O(q)$
\T\ gates per application. This gives the stated total cost.
\end{proof}

\section{Lower bounds}

In this section, we work in the adaptive Clifford+$\T$ model of
Gosset, Kothari, and Wu~\cite{david2026quantum} and establish lower bounds on the expected \T\ count
maximized over valid inputs.

\subsection{QROM}

We first vary the data on a fixed support to prove a lower bound for
promised sparse QROM. The bound allows a small constant error and matches
the exact construction in \cref{thm:Algorithm-promised-sparse-QROM}.
\begin{theorem}[\T-count lower bound for approximate promised sparse QROMs]
\label{thm:lb-promised-qrom}
If $m\ge2$, $4\le s\le2^n$, and $\varepsilon < \frac{1}{128}$, then for the
adaptive promised \(s\)-sparse QROM with $n$-bit input, $m$-bit output,
and error $\varepsilon$, the adaptive \T-count lower bound is
\[
    \Omega\br{\min\{s,\sqrt{sm}\}}.
\]
\end{theorem}

\begin{proof}[Proof of \cref{thm:lb-promised-qrom}]
Suppose every adaptive promised \(s\)-sparse QROM with error
\(\varepsilon\) can be implemented with adaptive \T\ count at most \(t\).
Let \(\bar s=2^{\lfloor\log_2 s\rfloor}\) and
\(r=\lfloor\log_2 s\rfloor \), so that \(s/2\le\bar s\le s\). We restrict attention to
QROMs supported on
\(S=\set{0,1}^r\times\set{0}^{n-r}\), which form a subclass of the
\(s\)-sparse instances.
We use $z || 0^{n-r}$ to denote the concatenation of the string
$z\in\set{0,1}^r$ with $n-r$ trailing zeros.
For sufficiently large $s$, choose $k=\min\{m,\lfloor c_0\bar s\rfloor\}$,
where $c_0>0$ is a sufficiently small constant. We use data words whose
first bit is one, whose next $k-1$ bits vary, and whose remaining
$m-k$ bits are zero. Discarding the fixed padding gives a $k$-bit
output with the same error and cost.
By \cref{lem:alphabet-packing}, there exists a family
\(\mathcal{C}\subseteq(\{1\}\times\set{0,1}^{k-1})^S\) of functions
such that any two distinct functions differ in at least \(\bar s/4\)
coordinates, and
$\log_2 |\mathcal{C}| = \Omega(\bar s k)=\Omega(sk)$.

For each function $c\in\mathcal{C}$, define the reduced data vector $d^{c}$ by
$
    d_x^{c}
    =
    \begin{cases}
        c(x), & x\in S,\\
        0^k, & x\notin S.
    \end{cases}
$
and define the states
$$
    \ket{\phi}
    = \frac{1}{\sqrt{\bar s}}
    \sum_{z\in\set{0,1}^r}\ket{z}\ket{0^k}, \qquad
        \ket{\phi_c}
        =
    \frac{1}{\sqrt{\bar s}}
    \sum_{z\in\set{0,1}^r}\ket{z}\ket{c(z || 0^{n-r})}.
$$
The ideal promised sparse QROM $O_{d^c}^{\mathrm{prom}}$ satisfies
$$
O_{d^c}^{\mathrm{prom}} \br{ \ket{\phi} \ket{0^{n-r}} \ket{0^a} } = \ket{\phi_c} \ket{0^{n-r}} \ket{0^a}.
$$
Thus, each $\ket{\phi_c}$ can be prepared from an ideal promised sparse QROM implementation with the same \T\ count.

We now apply the same construction to approximate promised sparse QROMs. For each \(c\in\mathcal C\), let \(\mathcal A_c\) be an adaptive promised sparse QROM implementation for \(d^c\), and let \(\mathcal E_c\) be the induced channel. For simplicity, let
\(a' = n-r+a\). Since \(\ket{\phi}\ket{0^{a'}}\in \mathcal H_S\), by
\cref{def:adaptive-promised-sparse-QROM} we have
$$
    \left\|
        \mathcal E_c\br{\ket{\phi,0^{a'}}\bra{\phi,0^{a'}}}
        -
        \ket{\phi_c,0^{a'}}\bra{\phi_c,0^{a'}}
    \right\|_1
    \le
    \varepsilon.
$$
Postselect the last \(a'\) qubits on \(0^{a'}\). Let \(p_c\) be the
success probability, and let \(\varphi_c\) be the resulting density operator on the
first \(r+k\) qubits. Since the ideal output always has these qubits equal to
\(0^{a'}\), we have \(p_c\ge 1-\varepsilon/2 > 1/2\). Repeating until success
therefore prepares \(\varphi_c\) with expected \T\ count at most \(2t\).
Measuring the last \(a'\) qubits is trace-norm contractive, so
\begin{align*}
    &\left\|
        p_c \cdot \varphi_c\otimes \ketbra{0^{a'}}
        -
        \ketbra{\phi_c,0^{a'}}
    \right\|_1 \leq \varepsilon \\
    &\Rightarrow 
    p_c
    \left\|
        \varphi_c\otimes \ketbra{0^{a'}}
        -
        \ketbra{\phi_c,0^{a'}}
    \right\|_1
    -
    (1-p_c) \leq \varepsilon \\
    &\Rightarrow
    \frac12
    \left\|
        \varphi_c
        -
        \ketbra{\phi_c}
    \right\|_1
    -
    \frac{\varepsilon}{2} \leq \varepsilon \\
    &\Rightarrow \frac12 \left\| \varphi_c - \ketbra{\phi_c} \right\|_1 \le \frac{3\varepsilon}{2}
\end{align*}
where in the third line we use $p_c \ge 1-\varepsilon/2 > 1/2$.

The preparation of $\varphi_c$ has trace-distance error at most
$3\varepsilon/2$ from $\ketbra{\phi_c}$ and expected \T\ count at most $2t$.
By \cref{lem:gkw-pauli-postselection}, each function $c$ gives a state
$\ket{\psi_c}$ and a normal form such that
$$
    \ket{\psi_c}\ket{0^{4t}}
    \propto
    C(I+P_{4t})\cdots(I+P_1)
    \br{\ket{0^{r+k}}\ket{T}^{\otimes 4t}}
$$
and
$$ \frac12 \left\| \ketbra{\psi_c} - \ketbra{\phi_c} \right\|_1 \le 3\sqrt{\varepsilon}.$$
Recall that \(c\) is chosen from the family \(\mathcal{C}\). For distinct \(c,c'\in\mathcal{C}\), since their Hamming distance is at
least $\bar s/4$, they agree on at most $3\bar s/4$ coordinates. Therefore,
$$
    \frac12
    \left\|
        \ketbra{\phi_c}
        -
        \ketbra{\phi_{c'}}
    \right\|_1
    =
    \sqrt{1-\abs{\braket{\phi_c}{\phi_{c'}}}^2}
    \ge
    \frac{\sqrt{7}}{4}.
$$
If two distinct functions \(c,c'\) gave the same density operator
$\ketbra{\psi_c}=\ketbra{\psi_{c'}}$, then by the
triangle inequality we would have
\[
    \frac12
    \left\|
        \ketbra{\phi_c}
        -
        \ketbra{\phi_{c'}}
    \right\|_1
    \leq  \frac{1}{2} \br{\left\|
        \ketbra{\phi_c}
        -
        \ketbra{\psi_{c}}
    \right\|_1 +  \left\|
        \ketbra{\phi_{c'}}
        -
        \ketbra{\psi_{c'}}
    \right\|_1} \leq 
    6\sqrt{\varepsilon}.
\]
Since \(\varepsilon<1/128\), we have
\(6\sqrt{\varepsilon}<6/\sqrt{128}<\sqrt7/4\), a contradiction.
Thus distinct functions require distinct normal forms.

It remains to count the number of such normal forms. Here, $C$ is a Clifford
unitary on $r+k+4t$ qubits, and each $P_i$ is a Hermitian Pauli on
$r+k+4t$ qubits. Therefore, the total number of normal forms is at most
$$
    2^{O((r+k+4t)^2)}
    \br{2\cdot 4^{r+k+4t}}^{4t}
    =
    2^{O((\log s+k)^2+t^2)}.
$$
Since distinct functions give distinct normal forms, we get
$|\mathcal{C}| \le 2^{O((\log s+k)^2+t^2)}$.
Combining this with
$\log_2 |\mathcal{C}|=\Omega(\bar s k)=\Omega(sk)$ yields
$$
    t^2 \ge \Omega(sk)-O\br{(\log s+k)^2}.
$$
Since $k\le c_0\bar s$, choosing $c_0$ sufficiently small absorbs the
$k^2$ term, and $(\log s)^2=o(s)$ absorbs the remaining correction.
Thus $t=\Omega(\sqrt{sk})=\Omega(\min\{s,\sqrt{sm}\})$.

For bounded $s\ge4$, use the four-address table
$d_{z_1z_2}=(1,z_1z_2,0^{m-2})$, with all other entries zero.
On a uniform superposition over these addresses, discarding the fixed bits
leaves $\ket{\chi}=\frac12\sum_{z_1,z_2}\ket{z_1,z_2,z_1z_2}$, the
three-qubit Toffoli state.
Every Pauli expectation other than those of $\pm I$ has absolute value at
most $1/2$~\cite[Eq.~(39)]{gosset2025multi}. If a stabilizer state
$\ket\sigma$ is fixed by a nonidentity Hermitian Pauli $P$, then
\[
    |\braket{\sigma}{\chi}|^2
    \le\bra\chi\frac{I+P}{2}\ket\chi\le\frac34.
\]
Let $p_0$ be the probability that an implementation uses no $\T$ gates.
Conditioned on this event, the reduced output is a mixture of stabilizer
states. Its overall output $\rho$ therefore satisfies
$\bra\chi\rho\ket\chi\le1-p_0/4$. The error bound also gives
$\bra\chi\rho\ket\chi\ge1-\varepsilon/2$, so $p_0\le2\varepsilon$.
The expected \T\ count is at least $1-p_0=\Omega(1)$, proving the
remaining cases.
\end{proof}

Taking $s=2^n$ in \cref{thm:lb-promised-qrom} gives the dense QROM lower
bound, since a dense QROM must in particular act correctly on its support.

\begin{corollary}[\T-count lower bound for approximate dense QROMs]
\label{thm:lb-dense-qrom}
If $n\ge2$, $m\ge2$, and $\varepsilon < \frac{1}{128}$, then for the adaptive dense QROM with $n$-bit input,
$m$-bit output, and error $\varepsilon$, the adaptive \T\ count lower bound is
\[
    \Omega\br{\min\{2^n,\sqrt{2^n m}\}}.
\]
\end{corollary}

We next prove the lower bound for exact sparse QROMs.

\begin{theorem}[\T-count lower bound for exact sparse QROM]
\label{thm:lb-exact-sparse-qrom-support}
For $n\ge2$, $m\ge2$, and $4\le s\le2^n$, the adaptive \T-count lower
bound for exact $s$-sparse QROM with $n$-bit input and $m$-bit output is
\[
    \Omega\!\left(n+\min\{s,\sqrt{s(m+\log(2^{n+1}/s))}\}\right).
\]
\end{theorem}

\begin{proof}[Proof of \cref{thm:lb-exact-sparse-qrom-support}]
Suppose every such QROM can be implemented with expected \T\ count
at most $t$. We first show that $t=\Omega(n)$. Take a table whose
only nonzero entry is $d_0=10^{m-1}$. On a uniform address input,
and after discarding the last $m-1$ output bits, it prepares
$\ket{\phi}=2^{-n/2}\sum_x\ket{x}\ket{\mathbf1[x=0]}$.
For a $q$-qubit pure state $\ket\psi$, its stabilizer nullity is
$q-\log_2|\{P:P\ket\psi=\ket\psi\}|$, where $P$ ranges over signed
Hermitian Pauli operators. It is zero on stabilizer states, is unchanged by
Clifford gates and clean ancillas, cannot increase under Pauli
postselection, and increases by at most one per $\T$ gate
\cite{beverland2020lowerbounds,gosset2025multi}.

For the state $\ket\phi$ above, any Pauli that fixes it must preserve its
computational-basis support. Its $X$ component $(a,c)$ must therefore satisfy
$\mathbf1[x+a=0]=\mathbf1[x=0]\oplus c$ for all $x$.
When $n\ge2$, $c=1$ would equate a singleton set with the complement of a
singleton set, and $c=0$ forces $a=0$. The remaining $Z$ component would
require $b\cdot x+e\mathbf1[x=0]$ to be constant over $\mathbb F_2$.
If $e=0$, this forces $b=0$; if $e=1$, it would make the singleton indicator
an affine function, impossible since a nonconstant affine function is one
on $2^{n-1}>1$ inputs. Thus only the identity fixes $\ket\phi$, whose
nullity is $n+1$. Every nonzero-probability measurement outcome of an exact
preparation gives this same pure output, so each such execution uses at
least $n+1$ \T\ gates. The expected count has the same lower bound.

We next vary the support. Let $N=2^n$ and
$\bar s=\min\{s,N/2\}$. For each $\bar s$-element set
$S\subseteq\set{0,1}^n$, set $d_x^S=10^{m-1}$ on $S$ and zero
elsewhere. Applying its QROM to a uniform address input prepares
\[
    \ket{\phi_S}=\frac1{\sqrt N}\sum_x\ket{x}\ket{\mathbf1[x\in S]}.
\]
Different supports give different exact states. The fixed output
bits and the clean work registers can be removed before applying
\cref{lem:gkw-pauli-postselection} with error zero. Counting on
$n+1$ output qubits therefore gives
\[
    t^2\ge\Omega\!\left(\log_2\binom N{\bar s}\right)-O(n^2)
    =\Omega\!\left(s\log(2N/s)\right)-O(n^2).
\]
Here the last equality follows from \cref{fact:binomial-estimate}
and $s/2\le\bar s\le s$. Combining this with $t=\Omega(n)$ absorbs
the correction and gives $t=\Omega(\sqrt{s\log(2N/s)})$.

Finally, \cref{thm:lb-promised-qrom} gives
$t=\Omega(\min\{s,\sqrt{sm}\})$. The worst-case cost is at least
each of these bounds, and hence at least a constant multiple of
their sum. Since
$\min\{s,\sqrt{s(m+\log(2N/s))}\}
\le\min\{s,\sqrt{sm}\}+\sqrt{s\log(2N/s)}$,
the claimed bound follows.
\end{proof}

For approximate sparse QROM, the following lower bound matches
\cref{thm:Algorithm-approximate-sparse-QROM}.
\begin{corollary}[\T-count lower bound for approximate sparse QROM]
\label{cor:lb-approximate-sparse-qrom}
For $m\ge2$, $4\le s\le2^n$, and $\varepsilon<1/128$, the adaptive
\T-count lower bound for $s$-sparse QROM with error $\varepsilon$ is
$\Omega(\min\{s,\sqrt{sm}\})$.
\end{corollary}
\begin{proof}
Restricting any sparse QROM implementation to the promised input
subspace gives an implementation of the same promised sparse QROM.
The claim follows from \cref{thm:lb-promised-qrom}.
\end{proof}

\subsection{Sparse state preparation}

We next prove a matching adaptive \T-count lower bound for sparse state
preparation.

\begin{theorem}[\T-count lower bound for approximate sparse state preparation]
\label{thm:sparse-state-lower-bound}
Let $n\ge2$ and $4\le s\le2^n$.
For any \(0<\varepsilon\le 1/64\), there exists an \(s\)-sparse
\(n\)-qubit state \(\ket{\psi}\) such that any adaptive Clifford\,$+$\,\T\
circuit preparing \(\ket{\psi}\) within trace-distance error \(\varepsilon\)
has expected \T\ count at least
\[
    \Omega\!\left(
        \min\{s,\sqrt{s\log(2^{n+1}/s)}\}
        +
        \sqrt{s\log(1/\varepsilon)}
        +
        \log(1/\varepsilon)
    \right).
\]
\end{theorem}

\begin{proof}
Let \(N=2^n\).
We vary the positions and values of the nonzero amplitudes separately.

We first vary the support set and assume $s\ge Cn$, for a sufficiently
large constant $C$. Choose a sufficiently large fixed constant $K$
and put $\bar s=\lfloor s/K\rfloor$.
By \cref{lem:sparse-support-packing}, there is a family
\(\mathcal{F}\) of \(\bar s\)-element subsets of \([N]\) such that
\[
    |S\cap S'|\le \frac{\bar s}{2},
    \quad
    \forall S\neq S'\in\mathcal{F},
\]
and \( |\mathcal{F}|\ge 2^{\Omega(s\log(2N/s))}\).
For each \(S\in\mathcal{F}\), define 
$\ket{S}  = \frac{1}{\sqrt{\bar s}}\sum_{x\in S}\ket{x}$.
If \(S\neq S'\), then
\[
    \frac{1}{2}
    \left\|
        \ket{S}\!\bra{S}
        -
        \ket{S'}\!\bra{S'}
    \right\|_1
    =
    \sqrt{1-|\braket{S}{S'}|^2}
    \ge
    \sqrt{1 - (|S\cap S'|/\bar s)^2}
    \ge
    \frac{\sqrt{3}}{2}.
\]
Now suppose every \(\ket{S}\), with \(S\in\mathcal{F}\), can be prepared to
error \(\varepsilon\le 1/64\) with \T\ count at most \(t\). By
\cref{lem:gkw-pauli-postselection}, each such circuit determines a
Pauli-postselection normal form, and the number of such forms is at most
\(2^{O(n^2+t^2)}\). Since the \(\sqrt{6\varepsilon}\)-balls around the states
\(\{\ket{S}:S\in\mathcal{F}\}\) are disjoint, distinct choices of positions
require distinct normal forms. Therefore,
\[
    2^{O(n^2+t^2)}
    \ge
    |\mathcal{F}|
    \ge
    2^{\Omega(s\log(2N/s))}.
\]
Equivalently, $t^2 \ge \Omega(s\log(2N/s)) - O(n^2)$.
The function $s\log(2N/s)$ is concave, so its minimum on $[Cn,N]$
occurs at an endpoint. Choosing $C$ sufficiently large and then $n$
sufficiently large absorbs the $O(n^2)$ correction. Hence
$t=\Omega(\sqrt{s\log(2N/s)})$ in this range.

It remains to consider $s<Cn$. Choose a sufficiently large constant
$C'>C$ and put $r=\lfloor s/C'\rfloor$. For sufficiently large $s$,
we have $r\le n$, $s\ge Cr$, and $\log s\le r/2$. The preceding
argument on $r$ qubits gives an $s$-sparse state requiring
$\Omega(\sqrt{s\log(2^{r+1}/s)})=\Omega(s)$ gates. Appending
$n-r$ zero qubits preserves this bound, since they may be discarded
after preparation. Also $\log(2N/s)\ge n+1-\log s=\Omega(s)$ when
$s<Cn$ and $s$ is sufficiently large, so the claimed support term is
$\Theta(s)$ in this range. We have therefore proved the support term
$\Omega(\min\{s,\sqrt{s\log(2N/s)}\})$ for sufficiently large $s$.

We next fix the support inside the first \(2^r\) basis states, where
\(r=\lfloor\log s\rfloor\), and vary the values on that support.
By the state-preparation lower bound of
Gosset, Kothari, and Wu~\cite{david2026quantum}, there exists an \(r\)-qubit state \(\ket{\chi}\)
such that every adaptive Clifford+\T\ preparation of \(\ket{\chi}\) to error
\(\varepsilon\) has adaptive \T\ count
\[
    \Omega\br{
        \sqrt{2^r\log(1/\varepsilon)}
        +
        \log(1/\varepsilon)
    }
    =
    \Omega\br{
        \sqrt{s\log(1/\varepsilon)}
        +
        \log(1/\varepsilon)
    }.
\]
Embed this state into \(n\) qubits by setting
\(\ket{\psi}=\ket{\chi}\ket{0^{n-r}}\). Then \(\ket{\psi}\) has support size
at most \(2^r\le s\). Any preparation of \(\ket{\psi}\) with smaller adaptive
\T\ count would, after discarding the last \(n-r\) qubits, yield a
preparation of \(\ket{\chi}\) with the same cost and the
same trace-distance error. Therefore,
\begin{equation*}
\label{eq:amplitude-lower-bound}
    t
    =
    \Omega\bigl(
        \sqrt{s\log(1/\varepsilon)}
        +
        \log(1/\varepsilon)
    \bigr).
\end{equation*}
The amplitude bound also covers bounded $s$ and $n$, for which the
support term is bounded. The worst-case complexity is at least each of these two lower bounds, and
hence at least a constant multiple of their sum. We therefore obtain
\[
    t
    =
    \Omega\bigl(
        \min\{s,\sqrt{s\log(2^{n+1}/s)}\}
        +
        \sqrt{s\log(1/\varepsilon)}
        +
        \log(1/\varepsilon)
    \bigr).
\]
\end{proof}

\subsection{Sparse block encoding}

We now prove a lower bound for the sparse block-encoding task.

\begin{theorem}[\T-count lower bound for approximate sparse block encodings]
\label{thm:sparse-block-encoding-lower-bound}
Fix constants \(0<c_1\le c_2\). There exists a constant \(c>0\) such that
the following holds for all \(n\ge2\). Let \(s\) be an integer
satisfying \(4\le s\le 2^n\), and let the normalization \(\alpha\) satisfy
\(c_1s\le\alpha\le c_2s\). For any
\(0<\varepsilon_{\rm BE}\le c\), there exists a row- and
column-\(s\)-sparse matrix
\(A\in\mathbb{C}^{2^n\times 2^n}\) with \(|a_{ij}|\le1\) such that any
adaptive Clifford\,$+$\,\T\ circuit implementing an
\((\alpha,a_{\rm BE},\varepsilon_{\rm BE})\)-block-encoding of \(A\) has
expected \T\ count at least
\[
    \Omega\!\left(
        \sqrt{2^n s(n+\log(1/\varepsilon_{\rm BE}))}
        +\log(1/\varepsilon_{\rm BE})
    \right).
\]
\end{theorem}

We first give a lemma which bounds the number of matrices that can be block encoded with a bounded \T\ count. 


\begin{lemma}
\label{lem:block-encoding-counting}
Let \(n>0\) and \(N=2^n\). Let \(\mathcal F\) be a finite set of
nonzero \(N\times N\) matrices, all with Frobenius norm \(\Gamma\). Suppose
that \(\varepsilon<\Gamma/\sqrt N\) and that for every distinct
\(A,A'\in\mathcal F\),
\begin{equation}
\label{eq:be-state-separation}
    \min_{\varphi\in\mathbb{R}}
    \|A-e^{\i\varphi}A'\|_F
    >
    4\sqrt N\varepsilon.
\end{equation}
For any \(\alpha>0\) and integer \(a_{\rm BE}\ge0\), if every
\(A\in\mathcal F\) has an
\((\alpha,a_{\rm BE},\varepsilon)\)-block-encoding implemented by an
adaptive Clifford\,$+$\,\T\ circuit with expected \T\ count at most an integer
\(t\), then
\[
    |\mathcal F|\le 2^{O(n^2+t^2)} .
\]
\end{lemma}
\begin{proof}
For each \(A\in\mathcal F\), let \(U_A\) be its $(\alpha, a_{\rm{BE}}, \varepsilon)$-block-encoding unitary and write
\[
    B_A
    =
    \alpha(\bra{0^{a_{\rm BE}}}\otimes I)
    U_A(\ket{0^{a_{\rm BE}}}\otimes I)
\]
for the encoded block. By the block-encoding error bound,
$\|A-B_A\|_F\le\sqrt N\|A-B_A\|\le\sqrt N\varepsilon$, so
$\|B_A\|_F\ge\Gamma-\sqrt N\varepsilon>0$.

We first express $B_A$ as a state.
Let \(\reg{Q}, \reg{R}\) be two \(n\)-qubit registers,
and \(\reg{E}\) be the block-encoding ancilla register. Applying \(U_A\) to
\(\ket{0^{a_{\rm BE}}}_{\reg{E}}\ket{\Phi_N}_{\reg{Q}\reg{R}}\), where
\(
    \ket{\Phi_N}_{\reg{Q}\reg{R}}
    =
    \frac1{\sqrt N}\sum_{j=0}^{N-1}\ket{j}_{\reg{Q}}\ket{j}_{\reg{R}},
\)
prepares the state
\[
    \ket{\Psi_A}_{\reg{E}\reg{Q}\reg{R}}
    =
    ((U_A)_{\reg{E}\reg{Q}}\otimes I_{\reg{R}})
    \ket{0^{a_{\rm BE}}}_{\reg{E}}\ket{\Phi_N}_{\reg{Q}\reg{R}}.
\]
Projecting the \(\reg{E}\) register of
\(\ket{\Psi_A}\) onto \(\ket{0^{a_{\rm BE}}}\) gives the vectorization of
\(B_A\):
\begin{align}
\label{eq:be-vectorization}
    (\bra{0^{a_{\rm BE}}}_{\reg{E}}\otimes I_{\reg{Q}\reg{R}})
    \ket{\Psi_A}_{\reg{E}\reg{Q}\reg{R}}
    &= (\bra{0^{a_{\rm BE}}}\otimes I)
    U_A(\ket{0^{a_{\rm BE}}}\otimes I) \ket{\Phi_N}\nonumber \\
    &= \frac1{\alpha \sqrt N}
    \sum_{j,k=0}^{N-1}(B_A)_{jk}\ket{j}_{\reg{Q}}\ket{k}_{\reg{R}}\nonumber\\
    &=
    \frac1{\alpha\sqrt N}\operatorname{vec}(B_A)_{\reg{Q}\reg{R}}.
\end{align}

By assumption, \(U_A\) can be
implemented by an adaptive Clifford+\T\ circuit with expected \T\ count at
most \(t\). Since \(\ket{\Phi_N}\) can be prepared by a Clifford circuit from the
all-zero state, \(\ket{\Psi_A}\) can be prepared by an adaptive
Clifford+\T\ circuit with expected \T\ count at most \(t\).
Applying \lem{gkw-adaptive-reduction} with error zero gives a postselected
circuit whose output has density operator $\ketbra{\Psi_A}$. Thus, for some
global phase \(\varphi_A\in\mathbb R\),
\begin{align*}
    C_{m+1}(I+P_m)C_m\cdots(I+P_1)C_1
    \left(
        \ket{0^{a_{\rm BE}+2n}}
        \ket{T}^{\otimes 2t}
        \ket{0^a}
    \right) \propto
    e^{\i\varphi_A}\ket{\Psi_A}\ket{0^{2t+a}},
\end{align*}
where the $C_i$ are Clifford unitaries and the $P_i$ are Hermitian Paulis.
Appending the projection $\ketbra{0^{a_{\rm{BE}} } }_{\reg{E}} \otimes I$ gives
\begin{align}
\label{eq:be-postselect}
    &(\ketbra{0^{a_{\rm{BE}} } }_{\reg{E}} \otimes I)
    C_{m+1}(I+P_m)C_m\cdots(I+P_1)C_1
    \left(
        \ket{0^{a_{\rm BE}+2n}}_{\reg{E}\reg{Q}\reg{R}}
        \ket{T}^{\otimes 2t}
        \ket{0^a}
    \right) \nonumber\\
    &\propto{}
    (\ketbra{0^{a_{\rm BE}}}_{\reg{E}}\otimes I)
    \bigl(e^{\i\varphi_A}\ket{\Psi_A}_{\reg{E}\reg{Q}\reg{R}}\ket{0^{2t+a}}\bigr)\nonumber\\
    &= e^{\i\varphi_A}\ket{0^{a_{\rm BE}}}_{\reg{E}} \bigl( (\bra{0^{a_{\rm BE}}}_{\reg{E}} \otimes I_{\reg{Q}\reg{R}})  \ket{\Psi_A}_{\reg{E}\reg{Q}\reg{R}}\bigr) \ket{0^{2t+a}}\nonumber\\
    &\propto{}
    e^{\i\varphi_A}\operatorname{vec}(B_A)_{\reg{Q}\reg{R}} \ket{0^{a_{\rm BE}}}_{\reg{E}}
    \ket{0^{2t+a}}
\end{align}
where the last line follows from \eqref{eq:be-vectorization}, with registers
reordered by SWAP gates. The appended projection is itself a product of
Pauli postselections:
\(
    \ketbra{0^{a_{\rm BE}}}_{\reg{E}} \otimes I
    =
    \prod_{j=1}^{a_{\rm BE}}
    (I+Z_j)/2\otimes I
\),
where $Z_j$ acts on the $j$-th qubit of $\reg{E}$.
We can therefore normalize the nonzero output in \eqref{eq:be-postselect} and apply
\Cref{lem:gkw-canonical-form} with $2n$ output qubits, $q=2t$ magic states,
and $a+a_{\rm BE}$ clean ancillas. This gives
\[
    e^{\i\varphi_A}\operatorname{vec}(B_A)\ket{0^{2t}}
    \propto
    C(I+P_{2t})\cdots(I+P_1)
    \big(\ket{0^{2n}}\ket{T}^{\otimes 2t}\big).
\]
There are $2^{O((2n+2t)^2)}$ Clifford unitaries up to global phase and
$(2\cdot4^{2n+2t})^{2t}$ choices of the Hermitian Paulis. Thus there are at
most $2^{O(n^2+t^2)}$ canonical-form circuits up to global phase,
independently of $a$ and $a_{\rm BE}$.

It remains to check that different matrices in \(\mathcal F\) cannot give the
same canonical-form circuit. The reverse triangle inequality gives
\begin{align*}
    \left\| \frac{\operatorname{vec}(A)}{\|A\|_F} 
    - \frac{\operatorname{vec}(B_A)}{\|B_A\|_F} \right\|_2 
    &\le 
    \left\| \frac{\operatorname{vec}(A)}{\|A\|_F} 
    - \frac{\operatorname{vec}(B_A)}{\|A\|_F} \right\|_2 + 
    \left\| \frac{\operatorname{vec}(B_A)}{\|A\|_F} 
    - \frac{\operatorname{vec}(B_A)}{\|B_A\|_F} \right\|_2 \\
    &= \frac{1}{\|A\|_F} \br{ \|A-B_A\|_F + \left| \|B_A\|_F - \|A\|_F\right| } \\
    &\leq  \frac{2\sqrt N\varepsilon}{\Gamma}.
\end{align*}

Next, for distinct \(A,A'\in\mathcal F\), suppose the same canonical-form circuit
prepares the two states obtained from \(A\) and \(A'\). Their normalized
outputs then satisfy
\[
    e^{\i\varphi_A} \frac{\operatorname{vec}(B_A)}{\|B_A\|_F}
    =
    e^{\i\varphi_{A'}} \frac{\operatorname{vec}(B_{A'})}{\|B_{A'}\|_F}.
\]
It follows that
\begin{align*}
    &\min_{\varphi\in\mathbb{R}}
    \left\| \operatorname{vec}(A) - e^{\i\varphi} \operatorname{vec}(A') \right\|_2 \\
    &\leq \left\| \operatorname{vec}(A) - e^{\i(\varphi_{A'} - \varphi_A)} \operatorname{vec}(A') \right\|_2 \\
    &= \Gamma \cdot \left\| e^{i\varphi_A} \frac{\operatorname{vec}(A)}{\|A\|_F} - e^{i\varphi_{A'}} \frac{\operatorname{vec}(A')}{\|A'\|_F} \right\|_2 \\
    &\leq  \Gamma \cdot \br{\left\| e^{i\varphi_A} \frac{\operatorname{vec}(A)}{\|A\|_F} - e^{i\varphi_A} \frac{\operatorname{vec}(B_A)}{\|B_A\|_F} \right\| + \left\|e^{i\varphi_{A'}} \frac{\operatorname{vec}(A')}{\|A'\|_F} - e^{i\varphi_{A'}} \frac{\operatorname{vec}(B_{A'})}{\|B_{A'}\|_F} \right\|_2} \\
    &\leq 4\sqrt{N} \varepsilon.
\end{align*}
By $\|\operatorname{vec}(X)\|_2=\|X\|_F$, this contradicts
\cref{eq:be-state-separation}. Hence,
the map from \(\mathcal F\) to canonical-form circuits is injective, and the
counting bound above gives the claim.
\end{proof}

Now we are ready to prove \cref{thm:sparse-block-encoding-lower-bound}.

\begin{proof}[Proof of \cref{thm:sparse-block-encoding-lower-bound}]
Let \(N=2^n\) and $\lambda=\min\{1,c_1/2\}$. Each hard matrix below has
at most $s$ nonzero entries per row and column, each of magnitude at most
$\lambda$, so its operator norm is at most $\lambda s\le\alpha$.
Suppose every instance has expected \T\ count at most $t$. The ancilla counts may be
made equal over each finite family by adding unused clean qubits.
We prove the lower bound using three hard families.

\paragraph{Varying the nonzero positions.}
By \cref{lem:sparse-matrix-support-packing}, there exists a family \(\mathcal{P}\) of subsets of \([N]\times[N]\) such that every \(P\in\mathcal{P}\) contains $r=\Theta(s)$
positions in each row and at most \(s\) positions in each column,
\[
    |P\triangle P'|\ge\gamma Ns,
    \quad
    \forall P\neq P'\in\mathcal{P},
\]
and
\[
    |\mathcal{P}|=2^{\Omega(Ns\log(2N/s))}.
\] 
For every \(P\in\mathcal{P}\),
let \(A_P\in\mathbb{R}^{N\times N}\) be the matrix with
\((A_P)_{ij}=\lambda\) for \((i,j)\in P\) and \(0\) otherwise. 
For distinct \(P,P'\in\mathcal{P}\),
$$\|A_P\|_F=\|A_{P'}\|_F=\lambda\sqrt{Nr}=\Theta(\sqrt{Ns})$$
and 
$$
    \min_{\varphi\in\mathbb{R}}
    \|A_P-e^{\i\varphi}A_{P'}\|_F
    =
    \|A_P-A_{P'}\|_F
    =
    \lambda\sqrt{|P\mathbin{\triangle}P'|}
    =
    \Omega(\sqrt{Ns}).
$$ where we use the fact that $A_P$ and $A_{P'}$ have nonnegative entries.
For sufficiently small $c$, the bound $\varepsilon_{\rm BE}\le c\sqrt{s}$
gives both $\sqrt N\varepsilon_{\rm BE}<\lambda\sqrt{Nr}$ and the
separation condition \eqref{eq:be-state-separation}. Thus
\cref{lem:block-encoding-counting} gives
$$ 2^{\Omega(Ns\log(2N/s))} \leq |\mathcal{P}|\le 2^{O(n^2+t^2)},$$
which gives $t^2 \ge \Omega(Ns\log(2N/s))-O(n^2)$. Since
$Ns\log(2N/s)=\Omega(2^n n)$, the correction is absorbed for
sufficiently large $n$, and thus
$$t = \Omega(\sqrt{2^n s \log(2^{n+1}/s)}).$$

\paragraph{Varying the entry values.}
Choose \(\lfloor N/s\rfloor\) disjoint \(s\times s\) blocks on the diagonal.
These blocks contain \(s^2\lfloor N/s\rfloor\) positions, which we group into
\(M=\lfloor s^2\lfloor N/s\rfloor/2\rfloor=\Theta(Ns)\) pairs. 
To keep the Frobenius norm fixed, assign the values
\[
    \frac{\lambda}{2}(\cos\theta_\ell,\sin\theta_\ell)
\]
to the \(\ell\)-th pair, where
\(\theta=(\theta_1,\ldots,\theta_M)\in[0,\pi/2]^M\), and set the remaining
entries to zero. The resulting matrix
\(A_\theta\) is row- and column-\(s\)-sparse. 
For \(\theta,\theta'\in[0,\pi/2]^M\), we have
\begin{align*}
    \|A_\theta-A_{\theta'}\|_F^2
    &=
    \sum_{\ell=1}^M
    \frac{\lambda^2}{4}
    \bigl((\cos\theta_\ell-\cos\theta'_\ell)^2
    +(\sin\theta_\ell-\sin\theta'_\ell)^2\bigr) \\
    &=
    \sum_{\ell=1}^M
    \lambda^2\sin^2\left(\frac{\theta_\ell-\theta'_\ell}{2}\right) \\
    &\geq \sum_{\ell=1}^M \frac{\lambda^2}{\pi^2} (\theta_\ell-\theta'_\ell)^2 = 
    \Omega(\|\theta-\theta'\|_2^2).
\end{align*} where we use the inequality $|\sin(x/2)| \geq |x/\pi|$ for $x \in [-\pi,\pi]$.

Choose \(\rho=C_0\sqrt{N}\varepsilon_{\rm BE}/\lambda\), where \(C_0\) is a
sufficiently large fixed constant, and put
\(k=\lfloor\log_2(\sqrt M/\rho)\rfloor\). Taking \(c\) sufficiently small
ensures \(\sqrt M/\rho\ge4\).
By \cref{lem:alphabet-packing}, there is a family of functions
\([M]\to\{0,1\}^k\) with minimum Hamming distance at least \(M/4\)
and logarithmic size \(\Omega(Mk)\). Interpret each \(k\)-bit value as an
integer \(j\), and replace it by \(\pi j/[2(2^k-1)]\).
The resulting family \(\mathcal C\subseteq[0,\pi/2]^M\) satisfies
\[
    \|\theta-\theta'\|_2
    \ge \frac{\pi\sqrt M}{4(2^k-1)}
    \ge \frac{\pi\rho}{4}
    \qquad (\theta\ne\theta'),
\]
and \(\log|\mathcal C|=\Omega(Mk)
=\Omega(Ns\log(\sqrt{s}/\varepsilon_{\rm BE}))\).
For distinct
\(\theta,\theta'\in\mathcal C\), the matrices \(A_\theta\) and
\(A_{\theta'}\) have nonnegative entries, and the preceding distance bounds give
\[
    \min_{\varphi\in\mathbb R}
    \|A_\theta-e^{\i\varphi}A_{\theta'}\|_F
    =
    \|A_\theta-A_{\theta'}\|_F
    \ge
    \lambda\rho/4.
\]
Choosing $C_0$ sufficiently large makes $\lambda\rho/4>4\sqrt N\varepsilon_{\rm BE}$,
as required by \eqref{eq:be-state-separation}. Also,
$\|A_\theta\|_F=\lambda\sqrt M/2$ for all $\theta$, so sufficiently small
$c$ gives $\sqrt N\varepsilon_{\rm BE}<\lambda\sqrt M/2$.
Thus \cref{lem:block-encoding-counting} gives
\[
    |\mathcal C|\le 2^{O(n^2+t^2)}.
\]

Combining this with the lower bound on \(|\mathcal C|\) gives
$t^2\ge\Omega(Ns\log(\sqrt{s}/\varepsilon_{\rm BE}))-O(n^2)$.
Since $s\ge4$ and $\varepsilon_{\rm BE}\le c$, the positive term is
$\Omega(2^n)$ and absorbs $O(n^2)$ for sufficiently large $n$. Hence
\[
    t
    =
    \Omega\!\left(
        \sqrt{2^n s\log(\sqrt{s}/\varepsilon_{\rm BE})}
    \right).
\]

\paragraph{Reduction from single-qubit state preparation.}
It remains to obtain the final logarithmic term. For this, we reduce
single-qubit state preparation to the block encoding of a rank-one sparse
matrix. Let
\(
    \bar s=2^{\lfloor\log_2 s\rfloor},
\) and \(
    r=\log_2\bar s,
\)
so that \(s/2\le\bar s\le s\). By
\cite[Lemma~5.9]{beverland2020lowerbounds} and
\cref{lem:gkw-adaptive-reduction}, there is a state
\(\ket{\chi}=\alpha_0\ket{0}+\alpha_1\ket{1}\) that requires
\(\Omega(\log(\bar s/\varepsilon_{\rm BE}))
=\Omega(\log(s/\varepsilon_{\rm BE}))\) adaptive \T\ count to prepare within
error \(\Theta(\varepsilon_{\rm BE}/\bar s)\).

Define vectors \(u\in\mathbb C^N\) and \(v\in\mathbb R^N\) by
\[
    u_i
    =
    \begin{cases}
        \alpha_0, & 0\le i<\bar s/2,\\
        \alpha_1, & \bar s/2\le i<\bar s,\\
        0, & i\ge\bar s,
    \end{cases}
    \qquad
    v_i
    =
    \begin{cases}
        1, & 0\le i<\bar s,\\
        0, & i\ge\bar s.
    \end{cases}
\]
Let \(A=\lambda uv^\dagger\). The matrix \(A\) is row- and column-\(\bar s\)-sparse,
and hence row- and column-\(s\)-sparse, and satisfies \(|a_{ij}|\le1\).
Let \(U\) be an assumed
\((\alpha,a_{\rm BE},\varepsilon_{\rm BE})\)-block-encoding of \(A\), and
write
\[
    B
    =
    \alpha(\bra{0^{a_{\rm BE}}}\otimes I)
    U(\ket{0^{a_{\rm BE}}}\otimes I).
\]
Set $
    w=\frac{A}{\alpha}\ket{0^{n-r}}\ket{+}^{\otimes r}$ and $
    \widetilde w=\frac{B}{\alpha}\ket{0^{n-r}}\ket{+}^{\otimes r}$.
A direct calculation gives
\[
    w
    =
    \frac{\lambda\bar s}{\alpha\sqrt2}
    (I_{n-r}\otimes I\otimes H^{\otimes(r-1)})
    \ket{0^{n-r}}\ket{\chi}\ket{0^{r-1}},
\]
and therefore
\[
    \|w\|_2=\frac{\lambda\bar s}{\alpha\sqrt2}
    \ge \frac{\lambda}{2c_2\sqrt2}.
\]
Moreover,
\[
    \|\widetilde w-w\|_2
    \le
    \frac{\|B-A\|}{\alpha}
    \le
    \frac{\varepsilon_{\rm BE}}{\alpha}
    =
    O(\varepsilon_{\rm BE}/s).
\]
For sufficiently small \(c\), it follows
that \(\|\widetilde w\|_2=\Omega(1)\). Thus, after applying \(U\) to
\(\ket{0^{a_{\rm BE}}}\ket{0^{n-r}}\ket{+}^{\otimes r}\), postselection of the block-encoding
ancillas onto \(\ket{0^{a_{\rm BE}}}\) succeeds with constant probability.
Repeating until success prepares
\(\widetilde w/\|\widetilde w\|_2\) with expected \T\ count \(O(t)\).

Applying \(I_{n-r}\otimes I\otimes H^{\otimes(r-1)}\) to
$w/\|w\|_2$ gives $\ket{0^{n-r}}\ket\chi\ket{0^{r-1}}$.
The same operation on $\widetilde w/\|\widetilde w\|_2$, followed by
discarding the other qubits, prepares $\ket\chi$ with trace-distance error
at most
\[
    \left\|
        \frac{\widetilde w}{\|\widetilde w\|_2}
        -\frac{w}{\|w\|_2}
    \right\|_2
    \le 2\frac{\|\widetilde w-w\|_2}{\|\widetilde w\|_2}
    =O(\varepsilon_{\rm BE}/s).
\]
Here we use contractivity of trace distance under partial
trace and its upper bound by Euclidean distance for pure states.
The single-qubit state-preparation lower bound therefore implies
\[
    t=\Omega\!\left(\log(s/\varepsilon_{\rm BE})\right).
\]
The worst-case
complexity is at least each of the three lower bounds and hence at least a
constant multiple of their sum. Therefore,
\[
    t
    =
    \Omega\!\left(
        \sqrt{2^n s\log(2^{n+1}/s)}
        +
        \sqrt{2^n s\log(\sqrt{s}/\varepsilon_{\rm BE})}
        +
        \log(s/\varepsilon_{\rm BE})
    \right).
\]
Finally,
$\log(2N/s)+\log(\sqrt{s}/\varepsilon_{\rm BE})
=n+1-\frac12\log s+\log(1/\varepsilon_{\rm BE})
=\Theta(n+\log(1/\varepsilon_{\rm BE}))$.
Since $\log s\le n\le\sqrt{Nsn}$, the displayed bound gives the
claimed expression. For the bounded values of $n$ excluded by the counting
estimates, the single-qubit precision bound also bounds the full expression
up to an absolute factor.
\end{proof}

\bibliographystyle{plainnat}
\bibliography{ref}

@article{hann2021resilience,
  title = {Resilience of quantum random access memory to generic noise},
  author = {Hann, Connor T. and Lee, Gideon and Girvin, S.M. and Jiang, Liang},
  journal = {PRX Quantum},
  volume = {2},
  issue = {2},
  pages = {020311},
  numpages = {30},
  year = {2021},
  month = {Apr},
  publisher = {American Physical Society},
  doi = {10.1103/PRXQuantum.2.020311},
  url = {https://link.aps.org/doi/10.1103/PRXQuantum.2.020311}
}

@article{jaques2025qram,
  doi = {10.22331/q-2025-12-02-1922},
  url = {https://doi.org/10.22331/q-2025-12-02-1922},
  title = {{QRAM}: a survey and critique},
  author = {Jaques, Samuel and Rattew, Arthur G.},
  journal = {{Quantum}},
  issn = {2521-327X},
  publisher = {{Verein zur F{\"{o}}rderung des Open Access Publizierens in den Quantenwissenschaften}},
  volume = {9},
  pages = {1922},
  month = dec,
  year = {2025}
}

@article{giovannetti2008quantum,
  title = {Quantum random access memory},
  author = {Giovannetti, Vittorio and Lloyd, Seth and Maccone, Lorenzo},
  journal = {Phys. Rev. Lett.},
  volume = {100},
  issue = {16},
  pages = {160501},
  numpages = {4},
  year = {2008},
  month = {Apr},
  publisher = {American Physical Society},
  doi = {10.1103/PhysRevLett.100.160501},
  url = {https://link.aps.org/doi/10.1103/PhysRevLett.100.160501}
}

@article{giovannetti2008architectures,
  title   = {Architectures for a quantum random access memory},
  author  = {Giovannetti, Vittorio and Lloyd, Seth and Maccone, Lorenzo},
  journal = {Phys. Rev. A},
  volume  = {78},
  pages   = {052310},
  year    = {2008},
  doi     = {10.1103/PhysRevA.78.052310},
  url     = {https://link.aps.org/doi/10.1103/PhysRevA.78.052310}
}

@article{zhangCircuitComplexityQuantum2024a,
  title = {Circuit complexity of quantum access models for encoding classical data},
  author = {Zhang, Xiao-Ming and Yuan, Xiao},
  year = 2024,
  month = apr,
  journal = {npj Quantum Inf.},
  volume = {10},
  number = {1},
  pages = {42},
  issn = {2056-6387},
  doi = {10.1038/s41534-024-00835-8},
  url = {https://doi.org/10.1038/s41534-024-00835-8}
}

@misc{gottesman1998heisenberg,
      title={The {H}eisenberg representation of quantum computers}, 
      author={Daniel Gottesman},
      year={1998},
      eprint={quant-ph/9807006},
      archivePrefix={arXiv},
      primaryClass={quant-ph},
      url={https://arxiv.org/abs/quant-ph/9807006}, 
}

@article{bravyi2005universal,
  title = {Universal quantum computation with ideal {Clifford} gates and noisy ancillas},
  author = {Bravyi, Sergey and Kitaev, Alexei},
  journal = {Phys. Rev. A},
  volume = {71},
  issue = {2},
  pages = {022316},
  numpages = {14},
  year = {2005},
  month = {Feb},
  publisher = {American Physical Society},
  doi = {10.1103/PhysRevA.71.022316},
  url = {https://link.aps.org/doi/10.1103/PhysRevA.71.022316}
}

@article{litinski2019magicstate,
  doi = {10.22331/q-2019-12-02-205},
  url = {https://doi.org/10.22331/q-2019-12-02-205},
  title = {Magic state distillation: not as costly as you think},
  author = {Litinski, Daniel},
  journal = {{Quantum}},
  issn = {2521-327X},
  publisher = {{Verein zur F{\"{o}}rderung des Open Access Publizierens in den Quantenwissenschaften}},
  volume = {3},
  pages = {205},
  month = dec,
  year = {2019}
}

@misc{gidney2024magicstatecultivation,
      title={Magic state cultivation: growing {T} states as cheap as {CNOT} gates}, 
      author={Gidney, Craig and Shutty, Noah and Jones, Cody},
      year={2024},
      eprint={2409.17595},
      archivePrefix={arXiv},
      primaryClass={quant-ph},
      url={https://arxiv.org/abs/2409.17595}, 
}

@InProceedings{kuperberg2013subexp,
  author =	{Kuperberg, Greg},
  title =	{Another subexponential-time quantum algorithm for the dihedral hidden subgroup problem},
  booktitle =	{8th Conference on the Theory of Quantum Computation, Communication and Cryptography (TQC 2013)},
  pages =	{20--34},
  series =	{Leibniz International Proceedings in Informatics (LIPIcs)},
  ISBN =	{978-3-939897-55-2},
  ISSN =	{1868-8969},
  year =	{2013},
  volume =	{22},
  editor =	{Severini, Simone and Brandao, Fernando},
  publisher =	{Schloss Dagstuhl -- Leibniz-Zentrum f{\"u}r Informatik},
  address =	{Dagstuhl, Germany},
  URL =		{https://drops.dagstuhl.de/entities/document/10.4230/LIPIcs.TQC.2013.20},
  URN =		{urn:nbn:de:0030-drops-43213},
  doi =		{10.4230/LIPIcs.TQC.2013.20}
}

@misc{gosset2025multi,
  title = {Multi-qubit {Toffoli} with exponentially fewer {T} gates},
  author = {Gosset, David and Kothari, Robin and Zhang, Chenyi},
  year = {2025},
  eprint = {2510.07223},
  archivePrefix = {arXiv},
  primaryClass = {quant-ph},
  url = {https://arxiv.org/abs/2510.07223}
}

@article{gidney2018halvingcostof,
  doi = {10.22331/q-2018-06-18-74},
  url = {https://doi.org/10.22331/q-2018-06-18-74},
  title = {Halving the cost of quantum addition},
  author = {Gidney, Craig},
  journal = {{Quantum}},
  issn = {2521-327X},
  publisher = {{Verein zur F{\"{o}}rderung des Open Access Publizierens in den Quantenwissenschaften}},
  volume = {2},
  pages = {74},
  month = jun,
  year = {2018}
}

@article{clader2022dense,
  author = {Clader, B. David and Dalzell, Alexander M. and Stamatopoulos, Nikitas and Salton, Grant and Berta, Mario and Zeng, William J.},
  journal = {IEEE Trans. Quantum Eng.},
  title = {Quantum resources required to block-encode a matrix of classical data},
  year = {2022},
  volume = {3},
  number = {},
  pages = {1--23},
  doi = {10.1109/TQE.2022.3231194},
  url = {https://doi.org/10.1109/TQE.2022.3231194}
}

@article{childs2017qls,
author = {Childs, Andrew M. and Kothari, Robin and Somma, Rolando D.},
title = {Quantum algorithm for systems of linear equations with exponentially improved dependence on precision},
journal = {SIAM J. Comput.},
volume = {46},
number = {6},
pages = {1920--1950},
year = {2017},
doi = {10.1137/16M1087072},
url = {https://doi.org/10.1137/16M1087072}
}

@article{sanders2019black-box,
  title = {Black-box quantum state preparation without arithmetic},
  author = {Sanders, Yuval R. and Low, Guang Hao and Scherer, Artur and Berry, Dominic W.},
  journal = {Phys. Rev. Lett.},
  volume = {122},
  issue = {2},
  pages = {020502},
  numpages = {5},
  year = {2019},
  month = {Jan},
  publisher = {American Physical Society},
  doi = {10.1103/PhysRevLett.122.020502},
  url = {https://link.aps.org/doi/10.1103/PhysRevLett.122.020502}
}

@inproceedings{david2026quantum,
  author = {Gosset, David and Kothari, Robin and Wu, Kewen},
  title = {Quantum state preparation with optimal {T}-count},
  booktitle = {Proceedings of the 2026 Annual ACM-SIAM Symposium on Discrete Algorithms (SODA)},
  pages = {3378--3406},
  year = {2026},
  organization = {SIAM},
  doi = {10.1137/1.9781611978971.122},
  url = {https://epubs.siam.org/doi/abs/10.1137/1.9781611978971.122}
}

@inproceedings{li2025nearly,
  title={Nearly optimal circuit size for sparse quantum state preparation},
  author={Li, Lvzhou and Luo, Jingquan},
  booktitle={52nd International Colloquium on Automata, Languages, and Programming (ICALP 2025)},
  series={Leibniz International Proceedings in Informatics (LIPIcs)},
  volume={334},
  pages={113:1--113:19},
  year={2025},
  publisher={Schloss Dagstuhl -- Leibniz-Zentrum f{\"u}r Informatik},
  doi={10.4230/LIPIcs.ICALP.2025.113},
  url={https://drops.dagstuhl.de/entities/document/10.4230/LIPIcs.ICALP.2025.113}
}

@misc{LuoLi26,
  title = {Sparse Quantum State Preparation with Sublinear {T}-Count},
  author = {Luo, Jingquan and Li, Lvzhou},
  year = {2026},
  eprint = {2608.00414},
  archivePrefix = {arXiv},
  primaryClass = {quant-ph},
  url = {https://arxiv.org/abs/2608.00414}
}

@inproceedings{berry2015simulating,
author = {Berry, Dominic W. and Childs, Andrew M. and Kothari, Robin},
title = {Hamiltonian simulation with nearly optimal dependence on all parameters},
year = {2015},
isbn = {9781467381918},
publisher = {IEEE Computer Society},
address = {USA},
booktitle = {2015 IEEE 56th Annual Symposium on Foundations of Computer Science},
url = {https://doi.org/10.1109/FOCS.2015.54},
doi = {10.1109/FOCS.2015.54},
pages = {792--809},
numpages = {18},
series = {FOCS '15}
}

@article{beverland2020lowerbounds,
doi = {10.1088/2058-9565/ab8963},
url = {https://doi.org/10.1088/2058-9565/ab8963},
year = {2020},
month = {may},
publisher = {IOP Publishing},
volume = {5},
number = {3},
pages = {035009},
author = {Beverland, Michael and Campbell, Earl and Howard, Mark and Kliuchnikov, Vadym},
title = {Lower bounds on the non-{Clifford} resources for quantum computations},
journal = {Quantum Sci. Technol.}
}

@article{Low2024tradingtgatesdirty,
  doi = {10.22331/q-2024-06-17-1375},
  url = {https://doi.org/10.22331/q-2024-06-17-1375},
  title = {Trading {T} gates for dirty qubits in state preparation and unitary synthesis},
  author = {Low, Guang Hao and Kliuchnikov, Vadym and Schaeffer, Luke},
  journal = {{Quantum}},
  issn = {2521-327X},
  publisher = {{Verein zur F{\"{o}}rderung des Open Access Publizierens in den Quantenwissenschaften}},
  volume = {8},
  pages = {1375},
  month = jun,
  year = {2024}
}

@misc{rupprecht2026sparsequantumstatepreparation,
  title={Sparse quantum state preparation with improved {Toffoli} cost}, 
  author={Rupprecht, Felix and W\"olk, Sabine},
  year={2026},
  eprint={2601.09388},
  archivePrefix={arXiv},
  primaryClass={quant-ph},
  url={https://arxiv.org/abs/2601.09388}, 
}

@article{zhang2022state,
  title = {Quantum state preparation with optimal circuit depth: implementations and applications},
  author = {Zhang, Xiao-Ming and Li, Tongyang and Yuan, Xiao},
  journal = {Phys. Rev. Lett.},
  volume = {129},
  issue = {23},
  pages = {230504},
  numpages = {6},
  year = {2022},
  month = {Nov},
  publisher = {American Physical Society},
  doi = {10.1103/PhysRevLett.129.230504},
  url = {https://link.aps.org/doi/10.1103/PhysRevLett.129.230504}
}

@ARTICLE{sun2023state,
  author={Sun, Xiaoming and Tian, Guojing and Yang, Shuai and Yuan, Pei and Zhang, Shengyu},
  journal={IEEE Transactions on Computer-Aided Design of Integrated Circuits and Systems}, 
  title={Asymptotically optimal circuit depth for quantum state preparation and general unitary synthesis}, 
  year={2023},
  volume={42},
  number={10},
  pages={3301-3314},
  doi={10.1109/TCAD.2023.3244885},
  url={https://doi.org/10.1109/TCAD.2023.3244885}
}

@article{yeo2025state,
  title = {Reducing circuit depth in quantum state preparation for quantum simulation using measurements and feedforward},
  author = {Yeo, Hyeonjun and Kim, Ha Eum and Sohn, IlKwon and Jeong, Kabgyun},
  journal = {Phys. Rev. Appl.},
  volume = {23},
  issue = {5},
  pages = {054066},
  numpages = {23},
  year = {2025},
  month = {May},
  publisher = {American Physical Society},
  doi = {10.1103/PhysRevApplied.23.054066},
  url = {https://link.aps.org/doi/10.1103/PhysRevApplied.23.054066}
}

@misc{zi2025state,
      title={Constant-depth quantum circuits for arbitrary quantum state preparation via measurement and feedback}, 
      author={Wei Zi and Junhong Nie and Xiaoming Sun},
      year={2025},
      eprint={2503.16208},
      archivePrefix={arXiv},
      primaryClass={quant-ph},
      url={https://arxiv.org/abs/2503.16208}, 
}

@article{low2019hamiltonian,
  doi = {10.22331/q-2019-07-12-163},
  url = {https://doi.org/10.22331/q-2019-07-12-163},
  title = {Hamiltonian simulation by qubitization},
  author = {Low, Guang Hao and Chuang, Isaac L.},
  journal = {{Quantum}},
  issn = {2521-327X},
  publisher = {{Verein zur F{\"{o}}rderung des Open Access Publizierens in den Quantenwissenschaften}},
  volume = {3},
  pages = {163},
  month = jul,
  year = {2019}
}

@article{campsExplicitQuantumCircuits2024a,
  title = {Explicit quantum circuits for block encodings of certain sparse matrices},
  author = {Camps, Daan and Lin, Lin and Van Beeumen, Roel and Yang, Chao},
  year = 2024,
  month = mar,
  journal = {SIAM J. Matrix Anal. Appl.},
  volume = {45},
  number = {1},
  pages = {801--827},
  publisher = {{Society for Industrial and Applied Mathematics}},
  issn = {0895-4798},
  doi = {10.1137/22M1484298},
  url = {https://doi.org/10.1137/22M1484298},
  urldate = {2026-07-03}
}

@misc{zecchi2026blockencodingsparsematrices,
      title={Block encoding of sparse matrices with a periodic diagonal structure}, 
      author={Zecchi, Alessandro Andrea and Sanavio, Claudio and Cappelli, Luca and Perotto, Simona and Roggero, Alessandro and Succi, Sauro},
      year={2026},
      eprint={2602.10589},
      archivePrefix={arXiv},
      primaryClass={quant-ph},
      url={https://arxiv.org/abs/2602.10589}, 
}

@article{Sunderhauf2024blockencoding,
  doi = {10.22331/q-2024-01-11-1226},
  url = {https://doi.org/10.22331/q-2024-01-11-1226},
  title = {Block-encoding structured matrices for data input in quantum computing},
  author = {S{\"{u}}nderhauf, Christoph and Campbell, Earl and Camps, Joan},
  journal = {{Quantum}},
  issn = {2521-327X},
  publisher = {{Verein zur F{\"{o}}rderung des Open Access Publizierens in den Quantenwissenschaften}},
  volume = {8},
  pages = {1226},
  month = jan,
  year = {2024}
}

@article{Yang2025dictionarybased,
  doi = {10.22331/q-2025-07-22-1805},
  url = {https://doi.org/10.22331/q-2025-07-22-1805},
  title = {Dictionary-based block encoding of sparse matrices with low subnormalization and circuit depth},
  author = {Yang, Chunlin and Li, Zexian and Yao, Hongmei and Fan, Zhaobing and Zhang, Guofeng and Liu, Jianshe},
  journal = {{Quantum}},
  issn = {2521-327X},
  publisher = {{Verein zur F{\"{o}}rderung des Open Access Publizierens in den Quantenwissenschaften}},
  volume = {9},
  pages = {1805},
  month = jul,
  year = {2025}
}

@article{QRS,
author = {Ozols, Maris and Roetteler, Martin and Roland, J\'{e}r\'{e}mie},
title = {Quantum rejection sampling},
year = {2013},
issue_date = {August 2013},
publisher = {Association for Computing Machinery},
address = {New York, NY, USA},
volume = {5},
number = {3},
issn = {1942-3454},
url = {https://doi.org/10.1145/2493252.2493256},
doi = {10.1145/2493252.2493256},
journal = {ACM Trans. Comput. Theory},
month = aug,
articleno = {11},
numpages = {33}
}

@article{adriano1995elementarygates,
  title = {Elementary gates for quantum computation},
  author = {Barenco, Adriano and Bennett, Charles H. and Cleve, Richard and DiVincenzo, David P. and Margolus, Norman and Shor, Peter and Sleator, Tycho and Smolin, John A. and Weinfurter, Harald},
  journal = {Phys. Rev. A},
  volume = {52},
  issue = {5},
  pages = {3457--3467},
  numpages = {0},
  year = {1995},
  month = {Nov},
  publisher = {American Physical Society},
  doi = {10.1103/PhysRevA.52.3457},
  url = {https://link.aps.org/doi/10.1103/PhysRevA.52.3457}
}

@misc{KimLaakkonen2025controlled,
      title={Any {Clifford+T} circuit can be controlled with constant {T}-depth overhead}, 
      author={Kim, Isaac H. and Laakkonen, Tuomas},
      year={2025},
      eprint={2512.24982},
      archivePrefix={arXiv},
      primaryClass={quant-ph},
      url={https://arxiv.org/abs/2512.24982}, 
}

@article{HHL2009,
  author = {Harrow, Aram W. and Hassidim, Avinatan and Lloyd, Seth},
  title = {Quantum algorithm for solving linear systems of equations},
  journal = {Phys. Rev. Lett.},
  volume = {103},
  number = {15},
  pages = {150502},
  year = {2009},
  doi = {10.1103/PhysRevLett.103.150502},
  url = {https://doi.org/10.1103/PhysRevLett.103.150502}
}

@article{Wiebe2012datafitting,
  author = {Wiebe, Nathan and Braun, Daniel and Lloyd, Seth},
  title = {Quantum data fitting},
  journal = {Phys. Rev. Lett.},
  volume = {109},
  pages = {050505},
  year = {2012},
  doi = {10.1103/PhysRevLett.109.050505},
  url = {https://doi.org/10.1103/PhysRevLett.109.050505}
}

@article{Clader2013preconditioned,
  author = {Clader, B. David and Jacobs, Bryan C. and Sprouse, Chad R.},
  title = {Preconditioned quantum linear system algorithm},
  journal = {Phys. Rev. Lett.},
  volume = {110},
  pages = {250504},
  year = {2013},
  doi = {10.1103/PhysRevLett.110.250504},
  url = {https://doi.org/10.1103/PhysRevLett.110.250504}
}

@article{Temme2011metropolis,
  title = {Quantum {Metropolis} sampling},
  author = {Temme, K. and Osborne, T. J. and Vollbrecht, K. G. and Poulin, D. and Verstraete, F.},
  year = 2011,
  month = mar,
  journal = {Nature},
  volume = {471},
  number = {7336},
  pages = {87--90},
  issn = {1476-4687},
  doi = {10.1038/nature09770},
  url = {https://doi.org/10.1038/nature09770}
}

@article{Montanaro2016fem,
  author = {Montanaro, Ashley and Pallister, Sam},
  title = {Quantum algorithms and the finite element method},
  journal = {Phys. Rev. A},
  volume = {93},
  pages = {032324},
  year = {2016},
  doi = {10.1103/PhysRevA.93.032324},
  url = {https://doi.org/10.1103/PhysRevA.93.032324}
}

@article{Kassal2010chemistry,
  author = {Kassal, Ivan and Whitfield, James D. and Perdomo-Ortiz, Alejandro and Yung, Man-Hong and Aspuru-Guzik, Al{\'a}n},
  title = {Simulating chemistry using quantum computers},
  journal = {Annu. Rev. Phys. Chem.},
  volume = {62},
  pages = {185--207},
  year = {2011},
  doi = {10.1146/annurev-physchem-032210-103512},
  url = {https://doi.org/10.1146/annurev-physchem-032210-103512}
}

@article{Babbush2015fermions,
  author = {Babbush, Ryan and Berry, Dominic W. and Kivlichan, Ian D. and Wei, Annie Y. and Love, Peter J. and Aspuru-Guzik, Al{\'a}n},
  title = {Exponentially more precise quantum simulation of fermions {I}: quantum chemistry in second quantization},
  journal = {New J. Phys.},
  volume = {18},
  pages = {033032},
  year = {2016},
  doi = {10.1088/1367-2630/18/3/033032},
  url = {https://doi.org/10.1088/1367-2630/18/3/033032}
}

@article{Berry2017lode,
  author = {Berry, Dominic W. and Childs, Andrew M. and Ostrander, Aaron and Wang, Guoming},
  title = {Quantum algorithm for linear differential equations with exponentially improved dependence on precision},
  journal = {Commun. Math. Phys.},
  volume = {356},
  pages = {1057--1081},
  year = {2017},
  doi = {10.1007/s00220-017-3002-y},
  url = {https://doi.org/10.1007/s00220-017-3002-y}
}

@misc{Chevignard2026DGS,
  author = {Chevignard, Cl\'{e}mence and Shen, Yixin and Schrottenloher, Andr\'{e}},
  title = {Quantum algorithm for discrete {Gaussian} sampling},
  year = {2026},
  eprint = {2605.20133},
  archivePrefix = {arXiv},
  primaryClass = {quant-ph},
  url = {https://arxiv.org/abs/2605.20133}
}

@misc{Ling2026QRSTrapdoor,
  author = {Ling, Cong and Yan, Hao and Zhao, Nicholas},
  title = {Improved dual attack and trapdoor sampling via quantum rejection sampling},
  year = {2026},
  eprint = {2605.24798},
  archivePrefix = {arXiv},
  primaryClass = {quant-ph},
  url = {https://arxiv.org/abs/2605.24798}
}

@article{RossSelinger2014,
author = {Ross, Neil J. and Selinger, Peter},
title = {Optimal ancilla-free {Clifford+T} approximation of {Z}-rotations},
year = {2016},
issue_date = {September 2016},
publisher = {Rinton Press, Incorporated},
address = {Paramus, NJ},
volume = {16},
number = {11--12},
issn = {1533-7146},
journal = {Quantum Info. Comput.},
month = sep,
pages = {901--953},
numpages = {53},
url={https://dl.acm.org/doi/abs/10.5555/3179330.3179331}
}

@inproceedings{GSLW19,
author = {Gily\'{e}n, Andr\'{a}s and Su, Yuan and Low, Guang Hao and Wiebe, Nathan},
title = {Quantum singular value transformation and beyond: exponential improvements for quantum matrix arithmetics},
year = {2019},
isbn = {9781450367059},
publisher = {Association for Computing Machinery},
address = {New York, NY, USA},
url = {https://doi.org/10.1145/3313276.3316366},
doi = {10.1145/3313276.3316366},
booktitle = {Proceedings of the 51st Annual ACM SIGACT Symposium on Theory of Computing},
pages = {193--204},
numpages = {12},
location = {Phoenix, AZ, USA},
series = {STOC 2019}
}

@article{babbush2018encoding,
  title = {Encoding electronic spectra in quantum circuits with linear {T} complexity},
  author = {Babbush, Ryan and Gidney, Craig and Berry, Dominic W. and Wiebe, Nathan and McClean, Jarrod and Paler, Alexandru and Fowler, Austin and Neven, Hartmut},
  journal = {Phys. Rev. X},
  volume = {8},
  issue = {4},
  pages = {041015},
  numpages = {36},
  year = {2018},
  month = {Oct},
  publisher = {American Physical Society},
  doi = {10.1103/PhysRevX.8.041015},
  url = {https://link.aps.org/doi/10.1103/PhysRevX.8.041015}
}

@article{dimatteo2019qramft,
  title         = {Fault-tolerant resource estimation of quantum random-access memories},
  author        = {{Di Matteo}, Olivia and Gheorghiu, Vlad and Mosca, Michele},
  journal       = {IEEE Trans. Quantum Eng.},
  volume        = {1},
  pages         = {1--13},
  year          = {2020},
  doi           = {10.1109/TQE.2020.2965803},
  url           = {https://doi.org/10.1109/TQE.2020.2965803}
}

@article{paler2020parallelisingqram,
  title   = {Parallelising the Queries in Bucket Brigade Quantum {RAM}},
  author  = {Paler, Alexandru and Oumarou, Oumarou and Basmadjian, Robert},
  journal = {Phys. Rev. A},
  volume  = {102},
  pages   = {032608},
  year    = {2020},
  doi     = {10.1103/PhysRevA.102.032608},
  url     = {https://link.aps.org/doi/10.1103/PhysRevA.102.032608}
}

@misc{haner2022tablelookup,
  title         = {Space-time optimized table lookup},
  author        = {H{\"a}ner, Thomas and Kliuchnikov, Vadym and Roetteler, Martin and Soeken, Mathias},
  year          = {2022},
  eprint        = {2211.01133},
  archivePrefix = {arXiv},
  primaryClass  = {quant-ph},
  url           = {https://arxiv.org/abs/2211.01133}
}

@article{zhu2024qlut,
  title = {Unified architecture for quantum lookup tables},
  author = {Zhu, Shuchen and Sundaram, Aarthi and Low, Guang Hao},
  journal = {Phys. Rev. Res.},
  volume = {7},
  issue = {4},
  pages = {043230},
  numpages = {20},
  year = {2025},
  month = {Dec},
  publisher = {American Physical Society},
  doi = {10.1103/d896-mktn},
  url = {https://link.aps.org/doi/10.1103/d896-mktn}
}

@article{cesa2025resourcestateqram,
  title     = {Resource-state quantum {RAM} for fast and error-correctable queries},
  author    = {Cesa, Francesco and Bernien, Hannes and Pichler, Hannes},
  journal   = {Nat. Commun.},
  year      = {2026},
  month     = jun,
  doi       = {10.1038/s41467-026-73275-x},
  url       = {https://www.nature.com/articles/s41467-026-73275-x}
}

@inproceedings{xu2025fattreeqram,
  title     = {{Fat-Tree QRAM}: A High-Bandwidth Shared Quantum Random Access Memory for Parallel Queries},
  author    = {Xu, Shifan and Lu, Alvin and Ding, Yongshan},
  booktitle = {Proceedings of the 30th ACM International Conference on Architectural Support for Programming Languages and Operating Systems, Volume 2},
  pages     = {390--406},
  year      = {2025},
  publisher = {ACM},
  address   = {New York, NY, USA},
  doi       = {10.1145/3676641.3716256},
  url       = {https://dl.acm.org/doi/10.1145/3676641.3716256}
}

@misc{motlagh2026halving,
  title         = {Halving the cost of {QROM}},
  author        = {Motlagh, Danial and Pocrnic, Matthew},
  year          = {2026},
  eprint        = {2605.20334},
  archivePrefix = {arXiv},
  primaryClass  = {quant-ph},
  url           = {https://arxiv.org/abs/2605.20334}
}

@article{Costa2022optimalQLSS,
  title = {Optimal scaling quantum linear-systems solver via discrete adiabatic theorem},
  author = {Costa, Pedro C.S. and An, Dong and Sanders, Yuval R. and Su, Yuan and Babbush, Ryan and Berry, Dominic W.},
  journal = {PRX Quantum},
  volume = {3},
  issue = {4},
  pages = {040303},
  numpages = {54},
  year = {2022},
  month = {Oct},
  publisher = {American Physical Society},
  doi = {10.1103/PRXQuantum.3.040303},
  url = {https://link.aps.org/doi/10.1103/PRXQuantum.3.040303}
}

@article{Chakraborty2022QRLs,
  doi = {10.22331/q-2023-04-27-988},
  url = {https://doi.org/10.22331/q-2023-04-27-988},
  title = {Quantum regularized least squares},
  author = {Chakraborty, Shantanav and Morolia, Aditya and Peduri, Anurudh},
  journal = {{Quantum}},
  issn = {2521-327X},
  publisher = {{Verein zur F{\"{o}}rderung des Open Access Publizierens in den Quantenwissenschaften}},
  volume = {7},
  pages = {988},
  month = apr,
  year = {2023}
}

@inproceedings{pagh1999hash,
  author="Pagh, Rasmus",
editor="Dehne, Frank
and Sack, J{\"o}rg-R{\"u}diger
and Gupta, Arvind
and Tamassia, Roberto",
title="Hash and displace: efficient evaluation of minimal perfect hash functions",
booktitle="Algorithms and Data Structures",
year="1999",
publisher="Springer Berlin Heidelberg",
address="Berlin, Heidelberg",
pages="49--54",
isbn="978-3-540-48447-9",
url="https://doi.org/10.1007/3-540-48447-7_5"
}

@inproceedings{esposito2020recsplit,
  author = {Esposito, Emmanuel and {Mueller Graf}, Thomas and Vigna, Sebastiano},
  title = {{RecSplit}: minimal perfect hashing via recursive splitting},
  booktitle = {2020 Proceedings of the Symposium on Algorithm Engineering and Experiments (ALENEX)},
  pages = {175--185},
  year = {2020},
  organization = {SIAM},
  doi = {10.1137/1.9781611976007.14},
  url = {https://epubs.siam.org/doi/abs/10.1137/1.9781611976007.14}
}

@article{ramacciotti2024state,
  title = {Simple quantum algorithm to efficiently prepare sparse states},
  author = {Ramacciotti, Debora and Lefterovici, Andreea I. and Rotundo, Antonio F.},
  journal = {Phys. Rev. A},
  volume = {110},
  issue = {3},
  pages = {032609},
  numpages = {10},
  year = {2024},
  month = {Sep},
  publisher = {American Physical Society},
  doi = {10.1103/PhysRevA.110.032609},
  url = {https://link.aps.org/doi/10.1103/PhysRevA.110.032609}
}

@article{mao2024state,
  title = {Toward optimal circuit size for sparse quantum state preparation},
  author = {Mao, Rui and Tian, Guojing and Sun, Xiaoming},
  journal = {Phys. Rev. A},
  volume = {110},
  issue = {3},
  pages = {032439},
  numpages = {9},
  year = {2024},
  month = {Sep},
  publisher = {American Physical Society},
  doi = {10.1103/PhysRevA.110.032439},
  url = {https://link.aps.org/doi/10.1103/PhysRevA.110.032439}
}

@INPROCEEDINGS{gleinig2021state,
  author={Gleinig, Niels and Hoefler, Torsten},
  booktitle={2021 58th ACM/IEEE Design Automation Conference (DAC)}, 
  title={An Efficient Algorithm for Sparse Quantum State Preparation}, 
  year={2021},
  volume={},
  number={},
  pages={433-438},
  doi={10.1109/DAC18074.2021.9586240},
    url={https://doi.org/10.1109/DAC18074.2021.9586240}
}

@article{deveras2022double,
  title = {Double sparse quantum state preparation},
  author = {{de Veras}, Tiago M. L. and {da Silva}, Leon D. and {da Silva}, Adenilton J.},
  year = 2022,
  month = jun,
  journal = {Quantum Information Processing},
  volume = {21},
  number = {6},
  pages = {204},
  issn = {1573-1332},
  doi = {10.1007/s11128-022-03549-y},
  url = {https://doi.org/10.1007/s11128-022-03549-y}
}

@article{mozafari2022state,
  title = {Efficient deterministic preparation of quantum states using decision diagrams},
  author = {Mozafari, Fereshte and De Micheli, Giovanni and Yang, Yuxiang},
  journal = {Phys. Rev. A},
  volume = {106},
  issue = {2},
  pages = {022617},
  numpages = {9},
  year = {2022},
  month = {Aug},
  publisher = {American Physical Society},
  doi = {10.1103/PhysRevA.106.022617},
  url = {https://link.aps.org/doi/10.1103/PhysRevA.106.022617}
}

@article{malvetti2021quantumcircuits,
  doi = {10.22331/q-2021-03-15-412},
  url = {https://doi.org/10.22331/q-2021-03-15-412},
  title = {Quantum circuits for sparse isometries},
  author = {Malvetti, Emanuel and Iten, Raban and Colbeck, Roger},
  journal = {{Quantum}},
  issn = {2521-327X},
  publisher = {{Verein zur F{\"{o}}rderung des Open Access Publizierens in den Quantenwissenschaften}},
  volume = {5},
  pages = {412},
  month = {Mar},
  year = {2021}
}
\end{document}